\documentclass[format=acmsmall,screen=true]{acmart}
\usepackage{graphicx}
\usepackage[show]{chato-notes}
\usepackage{amsmath,amssymb}
\usepackage{epsfig}
\usepackage{algorithm}
\usepackage{algorithmic}
\usepackage{url}
\usepackage{hyperref}
\usepackage{rotating}
\usepackage{tikz}
\usetikzlibrary{trees}
\usetikzlibrary{automata,positioning}
\usepackage{pdflscape}
\usepackage{tabularx}
\usepackage{xspace}

\hypersetup{
    bookmarks=true,         
    unicode=false,          
    pdftoolbar=true,        
    pdfmenubar=true,        
    pdffitwindow=false,     
    pdfstartview={FitH},    
    pdftitle={My title},    
    pdfauthor={Author},     
    pdfsubject={Subject},   
    pdfcreator={Creator},   
    pdfproducer={Producer}, 
    pdfnewwindow=true,      
    colorlinks=true,       
    linkcolor=black,          
    citecolor=black,        
    filecolor=black,      
    urlcolor=black           
}

\newcommand{\hide}[1]{}

\newcommand{\spara}[1]{\smallskip\noindent{\bf #1}}

\newtheorem{mydefinition}{Definition}
\newtheorem{mytheorem}{Theorem}
\newtheorem{mycorollary}{Corollary}
\newtheorem{mylemma}{Lemma}
\newtheorem{problem}{Problem}
\newtheorem{fact}{Fact}

\DeclareMathOperator*{\argmax}{arg\,max}

\newcommand{\NPhard}{$\mathbf{NP}$-hard}

\renewcommand{\vec}[1]{\mathbf{#1}}

\newcommand{\mlkcore}{\mbox{\ensuremath{\vec{k}}-core}}
\newcommand{\mlkcores}{{\mlkcore}s}

\renewcommand{\deg}{\mbox{\ensuremath{deg}}}

\newcommand{\squishlist}{
 \begin{list}{$\bullet$}
  {  \setlength{\itemsep}{0pt}
     \setlength{\parsep}{3pt}
     \setlength{\topsep}{3pt}
     \setlength{\partopsep}{0pt}
     \setlength{\leftmargin}{2em}
     \setlength{\labelwidth}{1.5em}
     \setlength{\labelsep}{0.5em}
} }
\newcommand{\squishlisttight}{
 \begin{list}{$\bullet$}
  { \setlength{\itemsep}{0pt}
    \setlength{\parsep}{0pt}
    \setlength{\topsep}{0pt}
    \setlength{\partopsep}{0pt}
    \setlength{\leftmargin}{2em}
    \setlength{\labelwidth}{1.5em}
    \setlength{\labelsep}{0.5em}
} }

\newcommand{\squishdesc}{
 \begin{list}{}
  {  \setlength{\itemsep}{0pt}
     \setlength{\parsep}{3pt}
     \setlength{\topsep}{3pt}
     \setlength{\partopsep}{0pt}
     \setlength{\leftmargin}{1em}
     \setlength{\labelwidth}{1.5em}
     \setlength{\labelsep}{0.5em}
} }

\newcommand{\squishend}{
  \end{list}
}

\newcommand{\mlcoredecomposition}{\mbox{\textsc{Multilayer Core Decomposition}}}
\newcommand{\mldensestsubgraph}{\mbox{\textsc{Multilayer Densest Subgraph}}}
\newcommand{\corenessvec}{coreness vector}
\newcommand{\kcorealg}{\mbox{$\vec{k}$-\textsf{core}}}
\newcommand{\kcorespathalg}{\mbox{$\vec{k}$-\textsf{coresPath}}\xspace}

\newcommand{\bfs}{\textsc{bfs-ml}-\textsf{cores}}
\newcommand{\dfs}{\textsc{dfs-ml}-\textsf{cores}}
\newcommand{\hybrid}{\textsc{hybrid-ml}-\textsf{cores}}
\newcommand{\rmaxsub}{\textsc{rim-ml}-\textsf{cores}}
\newcommand{\rmax}{\textsc{im-ml}-\textsf{cores}}
\newcommand{\rmaxcore}{\textsf{Inner-mostCore}}
\newcommand{\coresset}{\ensuremath{\mathbf{C}}}
\newcommand{\imcoresset}{\ensuremath{\mathbf{I}}}
\newcommand{\rightimcoresset}{\ensuremath{\mathbf{I}_r}}
\newcommand{\density}{\ensuremath{\delta}}
\newcommand{\densestalg}{\mbox{\textsc{ml}-\textsf{densest}}}
\newcommand{\SLinnermost}{\ensuremath{C^{(\mu)}}}
\newcommand{\SLdensest}{\ensuremath{S^{*}_{\mbox{\textsc{sl}}}}}
\newcommand{\eqncomment}[2]{\ensuremath{#1\{\mbox{#2}\}}}

\newcommand{\dmax}{\ensuremath{deg_{max}}}

\newcommand{\revision}[1]{#1}

\setcopyright{acmlicensed}
\acmJournal{TKDD}
\acmYear{2019}
\acmVolume{1}
\acmNumber{1}
\acmArticle{1}
\acmMonth{1}
\acmPrice{15.00}
\acmDOI{10.1145/3369872}

\begin{document}
\title[Core Decomposition in Multilayer Networks: Theory, Algorithms, and Applications]{Core Decomposition in Multilayer Networks:\\ Theory, Algorithms, and Applications}

\author{Edoardo Galimberti}
\affiliation{
 \institution{ISI Foundation}
 \city{Turin}
  \country{Italy}}
  \affiliation{
 \institution{University of Turin}
 \city{Turin}
  \country{Italy}}
\email{edoardo.galimberti@isi.it}

\author{Francesco Bonchi}
\affiliation{
 \institution{ISI Foundation}
 \city{Turin}
  \country{Italy}}
  \affiliation{
 \institution{Eurecat}
 \city{Barcelona}
  \country{Spain}}
\email{francesco.bonchi@isi.it}

\author{Francesco Gullo}
\affiliation{
 \institution{UniCredit, R\&D Department}
  \city{Rome}
  \country{Italy}}
\email{gullof@acm.org}

\author{Tommaso Lanciano}
\affiliation{
 \institution{Sapienza University of Rome}
  \city{Rome}
  \country{Italy}}
\email{lanciano@diag.uniroma1.it}

\begin{abstract}
Multilayer networks are a powerful paradigm to model complex systems, where multiple relations occur between the same entities.
Despite the keen interest in a variety of tasks, algorithms, and analyses in this type of network, the problem of extracting  dense subgraphs has remained largely unexplored so far.

As a first step in this direction, in this work
we study the problem of \emph{core decomposition of a multilayer network}.
Unlike the single-layer counterpart in which cores are all nested into one another and can be computed in linear time, the multilayer context is much more challenging as no total order exists among multilayer cores;
rather, they form a lattice whose size is exponential in the number of layers.
In this setting we devise three algorithms which differ in the way they visit the core lattice and in their pruning techniques.
We assess time and space efficiency of the three algorithms on a large variety of real-world multilayer networks.

We then move a step forward and study the problem of extracting the \emph{inner-most} (also known as \emph{maximal}) cores, i.e., the cores that are not dominated by any other core in terms of their core index in all the layers.
Inner-most cores are typically orders of magnitude less than all the cores.
Motivated by this, we devise an algorithm  that effectively exploits the maximality property and extracts inner-most cores directly, without first computing a complete decomposition. This allows for a consistent speed up over a na\"ive method that simply filters out non-inner-most ones from all the cores.

Finally, we showcase the multilayer core-decomposition tool in a variety of scenarios and problems.
We start by considering the problem of \emph{densest-subgraph extraction in multilayer networks}.
We introduce a definition of multilayer densest subgraph that trades-off between high density and number of layers in which the high density holds, and exploit multilayer core decomposition to approximate this problem with quality guarantees.
As further applications, we show how to utilize multilayer core decomposition to speed-up the extraction of \emph{frequent cross-graph quasi-cliques} and to generalize the \emph{community-search} problem to the multilayer setting.
\end{abstract}



\maketitle
\sloppy

\section{Introduction}
\label{sec:intro}

In several real-world contexts, such as social media, biological networks, financial networks, transportation systems, 
it is common to encounter multiple relations among the objects of the underlying domain.
Data in these scenarios is therefore modeled as a graph\footnote{Throughout the paper we use the terms ``network'' and ``graph'' interchangeably.} composed of a superimposition of different layers, i.e., where 
multiple edges of different types exist between a pair of vertices~\cite{DickisonMagnaniRossi2016,cai2005community, lee2015towards}.
\revision{
In the literature different terminologies have been used for graphs of this kind: \emph{multilayer networks}, \emph{multiplex networks}, \emph{multidimensional networks}, \emph{multirelational networks}, \emph{multislice networks}, and more.
No uniformity exists in this regard, and the various terms may also refer to slightly different concepts.
In this work we deal with networks composed of multiple layers, \emph{with no inter-layer links}, and hereinafter use the term ``\emph{multilayer networks}'' to refer to them.\footnote{\revision{Another popular terminology  for networks with multiple layers and no inter-layer links is ``\emph{multiplex networks}''}.}
}

Extracting dense structures from large graphs has emerged as a key graph-mining primitive in a variety of scenarios~\cite{aggarwal},
ranging from web mining~\cite{gibson}, to biology~\cite{fratkin,langston}, and
finance~\cite{Du}.
Although the literature on multilayer graphs has grown fast in the last years, the problem of extracting dense subgraphs in this type of graph has remained, surprisingly, largely unexplored.

In single-layer graphs, among the many definitions of a dense structure, \emph{core decomposition} plays a central role~\cite{bonchi18core}.
The $k$-\emph{core} of a  graph is defined as a maximal subgraph in which every vertex has at least $k$ neighbors within that subgraph.
The set of all $k$-cores of a graph $G$ forms the \emph{core decomposition} of $G$~\cite{Seidman1983k-cores}.
The appeal of core decomposition lies in the fact that it can be computed very fast, in linear time~\cite{MatulaB83,batagelj2011fast},
while, at the same time, being easily used to speed-up/approximate dense-subgraph extraction according to various other definitions.
For instance, core decomposition may speed up maximal-clique finding~\cite{EppsteinLS10}, as a
$k$-clique is guaranteed to be contained into a $(k\!-\!1)$-core, which can be significantly smaller than the original graph.
 Moreover, core decomposition can be used to approximate betweenness centrality~\cite{HealyJMA06}, and
to design linear-time approximation algorithms for the densest-subgraph~\cite{KortsarzP94} and the densest at-least-$k$-subgraph~\cite{AndersenC09} problems.

In this work we study the problem of \emph{core decomposition in multilayer networks}.
A major challenge with respect to the single-layer setting is that the number of multilayer cores is exponential in the number of layers.
We achieve this challenge by devising a number of pruning techniques to avoid from-scratch computation of all the cores and exploit such techniques to devise efficient algorithms.

Nevertheless, computing multilayer core decomposition efficiently is not enough.
Due the potentially high number of multilayer cores, a further major desideratum is to provide a data analyst with additional tools to browse through the output, and ultimately select the patterns of interest. The situation resembles that of the classic \emph{association-rule} and \emph{frequent-itemset} mining: a potentially exponential output, efficient algorithms to extract all the patterns, the need to have
concise summaries of the extracted knowledge, and the opportunity of using the extracted patterns as building blocks for more sophisticated analyses.
Following this direction, we present a number of problems and applications built on top of the multilayer core-decomposition tool.
First we focus on the problem of extracting only the \emph{maximal} or, as we call them in this work, the \emph{inner-most cores}, i.e., cores that are not ``dominated'' by any other core.
As experimentally observed, inner-most cores are orders of magnitude less than all the cores. 
Therefore, it is desirable to design algorithms that effectively exploit the maximality property and extract inner-most cores directly, without first computing a complete decomposition. Then, we show how multilayer core decomposition finds application to the problem of \emph{densest-subgraph extraction in multilayer networks} \cite{jethava2015finding,charikar2018finding}. As a further application, we exploit multilayer core decomposition to speed-up the extraction of \emph{frequent cross-graph quasi-cliques}~\cite{jiang2009mining}. Finally, we exploit multilayer core decomposition  to generalize the \emph{community-search} problem \cite{Sozio} to the multilayer setting.

\subsection{Background and related work}
\label{sec:related}
\spara{Core decomposition.}
Given a single-layer graph $G = (V,E)$ and a vertex $u \in V$, let $deg(u)$ and $deg_S(u)$ denote the degree of $u$ in $G$ and in a subgraph $S$ of $G$, respectively.
Also, given a subset  $C \subseteq V$ of vertices, let $E[C]$ denote the subset of edges induced by $C$.
%
\begin{mydefinition}[core decomposition]\label{def:kcores}
  The $k$\emph{-core} (or core of order $k$) of a single-layer graph $G = (V,E)$ is a
  \emph{maximal} subgraph $G[C_k] = (C_k,E[C_k])$ such that $\forall u \in C_k: deg_{C_k}(u) \geq k$.
  The set  $G = C_0 \supseteq C_1 \supseteq \cdots \supseteq C_{k^*}$ of all $k$-cores ($k^* = \arg\max_{k} C_k \neq \emptyset$) is the
  \emph{core decomposition} of $G$.
\end{mydefinition}
Core decomposition can be computed in linear time by iteratively removing the smallest-degree vertex and setting its core number as its degree at the time of removal~\cite{batagelj2011fast}.
Core decomposition has established itself as an important tool for network analysis and visualization~\cite{DBLP:conf/gd/BatageljMZ99,Alvarez-HamelinDBV05}
in several domains, e.g., bioinformatics~\cite{DBLP:journals/bmcbi/BaderH03,citeulike:298147},
software engineering~\cite{DBLP:journals/tjs/ZhangZCLZ10},
and social networks~\cite{Kitsak2010,GArcia2013}.
It has been studied under various settings, such as distributed~\cite{AksuDistributed2014,KhaouidKcore2015,DistributedCores1,DistributedCores2}, streaming~\cite{LiEfficient2014,StreamingCores,ZhangFast2017}, and external-memory~\cite{DiskCores,WenIO2016}, and for various types of graph, such as uncertain~\cite{bonchi14cores}, directed~\cite{DirectedCores}, weighted~\cite{WeigthedCores}, and attributed~\cite{zhang2017engagement} graphs.
Core decomposition has been studied also for temporal networks:~\cite{wu2015core} defines the $(k,h)$-core, where $h$ accounts for the number of multiple temporal edges between two vertices of degree at least $k$, while~\cite{galimberti2018mining} introduces the concept of (maximal) span-core, i.e., a core structure assigned with clear temporal collocation.
See~\cite{bonchi18core} for a comprehensive survey.

\revision{
In this paper we adopt the definition of a multilayer core by Azimi-Tafreshi~{\em et~al.}~\cite{MultiplexCores}, 
that is a core identified by an $|L|$-dimensional integer vector $\vec{k}$ (with $|L|$ being the number of layers of the given multilayer network), where every component of $\vec{k}$ refers to the minimum-degree constraint  in the corresponding layer, i.e., $\vec{k}[i]$  states the minimum degree required for that core in the $i$-th layer, for all $i \in [1..|L|]$.
Apart from introducing the multilayer-core definition,  Azimi-Tafreshi~{\em et~al.} study the core-percolation problem from a physics standpoint, with no algorithmic contribution: they characterize cores of 2-layer Erd\H{o}s-R\'{e}nyi and  scale-free networks, and observe how this characterization fits real-world air-transportation networks.
}
To the best of our knowledge, \emph{no prior work has studied how to efficiently compute the complete core decomposition of multilayer networks}.

\spara{Densest subgraph.}
Several notions of \emph{density} exist in the literature, each of which leading to a different version of the dense-subgraph-discovery problem.
While most variants are \NPhard\ and/or inapproximable, extracting dense subgraphs according to the \emph{average-degree density}
is solvable in polynomial time~\cite{Goldberg84}.
As a result, such a density has attracted most of the research in the field, so that the subgraph maximizing the average degree is commonly referred to as the \emph{densest subgraph}.
Goldberg~\cite{Goldberg84} provides an exact algorithm for finding the densest subgraph which is based on iteratively solving ad-hoc-defined minimum-cut problem instances.
Although principled, the Goldberg's algorithm cannot scale to large graphs.
Asahiro~\emph{et~al.}~\cite{AITT00} devises a linear-time greedy algorithm that has been shown to achieve $\frac{1}{2}$-approximation guarantee by Charikar~\cite{Char00}.
Such a greedy algortihm iteratively removes the smallest-degree vertex, and,
among all the subgraphs yielded during this vertex-removal process, the densest one is ultimately output.
Note that this algorithm resembles the one used for core decomposition.
In fact, it can be proved that the inner-most core of a graph is itself a $\frac{1}{2}$-approximation of the densest subgraph.

Variants of the densest-subgraph problem with size constraints turn out to be \NPhard.
For these variants, approximation algorithms and other theoretical results have been presented~\cite{AHI02,FPK01,AKK95,AndersenC09}.
A number of works  focus on extracting a subgraph maximizing densities other than the average degree.
For instance, Tsourakakis~\emph{et~al.}~\cite{BabisKDD13} resort to the notion of quasi-clique, while Tsourakakis~\cite{Tsourakakis15a} and Wang~\emph{et~al.}~\cite{TriangleDensePVLDB10} focus on notions of density based on k-cliques and/or triangles.
The densest-subgraph problem has also been studied in different settings, such as streaming/dynamic context~\cite{DensestStreaming,BhattacharyaHNT15,EpastoEfficient2015}, and top-$k$ fashion~\cite{Balalau15,GalbrunGT16,nasir2017fully}.


\spara{Dense structures in multilayer networks.}
Several recent works have dealt with the problem of extracting dense subgraphs from a set of multiple graphs sharing the same vertex set, which is a setting equivalent to the multilayer one we study in this work.
Jethava~and~Beerenwinkel~\cite{jethava2015finding} define the \emph{densest common subgraph} problem, i.e., find a subgraph maximizing the minimum average degree over \emph{all} input graphs, and devise a linear-programming formulation and a greedy heuristic for it.
Reinthal~\emph{et~al.}~\cite{reinthal2016finding} provide a Lagrangian relaxation of the Jethava~and~Beerenwinkel's linear program, which can be solved more efficiently.
Semertzidis~\emph{et~al.}~\cite{semertzidis2016best} introduce three more variants of the problem, whose goal is to maximize the average average degree, the minimum minimum degree, and the average minimum degree, respectively. They show that the average-average variant  reduces to the traditional densest-subgraph problem, and  the minimum-minimum variant is polynomial-time solvable by a simple adaptation of the algorithm for core decomposition. They also devise heuristics for the remaining two variants.
Charikar~\emph{et~al.}~\cite{charikar2018finding} further focus on the minimum-average and average-minimum formulations, by providing several theoretical findings, including \NPhard{ness}, hardness of the approximation (for both minimum-average and average-minimum), an integrality gap for the linear-programming relaxation introduced in~\cite{jethava2015finding,reinthal2016finding} (for minimum-average), and a characterization in terms of parameterized complexity (for average-minimum).

Other contributions in this area, less directly related to our work, deal with specific cases of 2-layer networks~\cite{wu2015finding, shen2015forming} and with the \emph{community-detection} problem~\cite{berlingerio2011findingredundant, mucha2010community, papalexakis2013more,cai2005community, TagarelliEnsemble2017, tang2010community, yin2013multi}.
Boden~\emph{et~al.}~\cite{boden2012mining} study \emph{subspace clustering} for multilayer graphs, i.e., find clusters of vertices that are densely connected by edges with similar labels for all possible label sets.
Yan~\emph{et~al.}~\cite{yan2005mining} introduce the problem of mining \emph{closed relational graphs}, i.e., frequent subgraphs of a multilayer graph exhibiting large minimum cut.
Jiang~\emph{et~al.}~\cite{jiang2009mining} focus on extracting \emph{frequent cross-graph quasi-cliques}, i.e., subgraphs that are quasi-cliques in at least a fraction of layers equal to a certain minimum support and have size larger than a given threshold.
Interdonato~\emph{et~al.}~\cite{InterdonatoTISP17} are the first to study the problem of \emph{local community detection in multilayer networks}, i.e., when a seed vertex is given and the goal is to reconstruct its community by having only a limited local view of the network.
Finally, Zhu~\emph{et~al.}~\cite{zhu2018diversified} address the problem of finding the $k$ most diversified $d$-coherent cores, i.e., the $k$ subgraphs having minimum degree at least $d$ that maximize the coverage of the vertices.

In this work, in  Section~\ref{sec:densest}, we introduce a formulation of the densest-subgraph problem in multilayer networks
that trades off between high density and number of layers where the high density holds.
We apply multilayer core decomposition  to provide provable approximation guarantees for this problem.
We also show that our formulation generalizes the minimum-average densest-common-subgraph problem studied in~\cite{charikar2018finding,jethava2015finding,reinthal2016finding,semertzidis2016best}, and our method achieves  approximation guarantees for that problem too.
Furthermore, in Section~\ref{sec:quasicliques}, we show how to exploit multilayer core decomposition to speed-up the problem of finding frequent cross-graph quasi-cliques~\cite{jiang2009mining}.

\spara{Community search.}
Given a (single-layer) graph and a set of query vertices, the \emph{community search} problem aims at finding a cohesive subgraph containing the query vertices.
Community search has received a great deal of attention in the data-mining community in the last few years (see e.g., a recent tutorial \cite{HuangLX17}).  
Sozio~and~Gionis~\cite{Sozio} are the first to introduce the community-search problem, by employing the minimum degree as a cohesiveness measure.
Their formulation can be solved by a simple (linear-time) greedy algorithm, which is very similar to the one proposed in~\cite{Char00} for the densest-subgraph problem.
More recently, Cui~\emph{et~al.}~\cite{SozioLocalSIGMOD14} devise a local-search approach to improve the efficiency of the method defined in~\cite{Sozio}, but only for the special case of a single query vertex.
The minimum-degree-based problem has been further studied in~\cite{BarbieriBGG15}, by exploiting core decomposition as a preprocessing step to allow more efficient and effective solutions.

Several formulations of the community search have also been studied under different names and in slightly different settings.
Andersen~and~Lang~\cite{Andersen1} and Kloumann~and~Kleinberg~\cite{Kloumann} study \emph{seed set expansion} in social graphs, in order to find communities with small conductance or that are well-resemblant of the characteristics of the query vertices, respectively.
Other works define \emph{connectivity subgraphs} based on electricity analogues~\cite{connect}, random walks~\cite{CenterpieceKDD06}, the minimum-description-length principle~\cite{akoglu2013mining}, the Wiener index~\cite{ruchansky2015minimum},
and  network efficiency~\cite{RuchanskyBGGK17}.
Community search has been formalized for attributed~\cite{huang2017attribute,fang2017attributed} and spatial graphs~\cite{fang2017spatial} as well.

In this work, in Section~\ref{sec:communitysearch}, we formulate the community-search problem for multilayer graphs, by adapting the early definition by Sozio~and~Gionis~\cite{Sozio}, and show how our algorithms for multilayer core decomposition can be exploited to obtain optimal solutions to this problem.

\subsection{Challenges, contributions, and roadmap}\label{sec:contributions}
\revision{
Let $G = (V,E,L)$ be a multilayer graph, where $V$ is a set of vertices, $L$ is a set of layers, and $E \subseteq V \times V \times L$ is a set of edges.
Given an $|L|$-dimensional integer vector $\vec{k} = [k_{\ell}]_{\ell \in L}$, the \emph{multilayer} $\vec{k}$-\emph{core} of $G$ is  a maximal subgraph whose vertices have at least degree $k_{\ell}$ in that subgraph, for all layers $\ell \in L$~\cite{MultiplexCores}.
}
Vector $\vec{k}$ is dubbed \emph{\corenessvec} of that core.
The set of all \emph{non-empty} and \emph{distinct} multilayer cores constitutes the \emph{multilayer core decomposition} of $G$.
A major challenge of computing the core decomposition of a multilayer network is that \emph{the number of multilayer cores are exponential in the number of layers}.
This makes the problem inherently hard, as the exponential size of the output clearly precludes the existence of polynomial-time algorithms in the general case.
In fact, unlike the single-layer case where cores are all nested into each other, no total order exists among multilayer cores.
Rather, they form a \emph{core lattice} defining partial containment.
As a result, 
algorithms in the multilayer setting must be crafted carefully to handle this exponential blowup, and avoid, as much as possible, the computation of unnecessary (i.e., empty or non-distinct) cores.

A na\"{\i}ve way of computing a multilayer core decomposition consists in generating all  \corenessvec{s}, run for each vector $\vec{k}$, an algorithm that iteratively removes vertices whose degree in a layer $\ell$ is less than the $\ell$-th component of $\vec{k}$, and filter out empty and duplicated cores. 
This method has evident efficiency issues, as every core is computed  from the whole input graph, and it does not avoid generation of empty or non-distinct cores at all.
As our first contribution, we devise three more efficient algorithms that exploit effective pruning rules during the visit of the  lattice.
The first two methods are based on a \textsc{bfs} and a \textsc{dfs} strategy, respectively: the \textsc{bfs} method exploits the fact that a core is contained into the intersection of all its fathers in the lattice, while the \textsc{dfs} method iteratively performs a single-layer core decomposition to compute, one-shot, all cores along a path from a non-leaf lattice core to a leaf.
The third method adopts a \textsc{hybrid} strategy embracing the main pros of \textsc{bfs} and \textsc{dfs}, and equipped with a \emph{look-ahead} mechanism to skip non-distinct cores.

We then shift the attention to the problem of computing \emph{all and only the inner-most cores}, i.e., the cores that are not dominated by any other core in terms of their index on all the layers.
A straightforward way of approaching this problem would be to first compute the complete core decomposition, and then filter out the non-inner-most cores. However, as the inner-most cores are usually much less than the overall cores, it would be desirable to have a method that effectively exploits the maximality property and extracts the inner-most ones directly, without computing a complete decomposition. The design of an algorithm of this kind
is an interesting challenge, as it contrasts the intrinsic conceptual
properties of core decomposition, based on which a core of order $k$ (in one layer)
can be efficiently computed from the core of order $k - 1$, of which it is
a subset, thus naturally suggesting a bottom-up discovery. For this reason, at first glance, the computation of the core
of the highest order would seem as hard as computing the overall
core decomposition. In this work we show that, by means of a clever core-lattice visiting strategy, we can prune huge portions of the search space, thus achieving higher efficiency than computing the whole decomposition.

As a major application of multilayer core decomposition, we then focus on the problem of \emph{extracting the densest subgraph from a multilayer network}.
\revision{
As already discussed in Section~\ref{sec:related}, a number of works aim at extracting a subgraph that maximizes the minimum average degree over \emph{all} layers~\cite{charikar2018finding,jethava2015finding,reinthal2016finding,semertzidis2016best}.
A major limitation of that formulation is that, considering all layers, even the noisy/insignificant ones would contribute to selecting the output subgraph, which might prevent us from finding a subgraph being dense in a still large subset of layers.
Another simplistic approach at the other end of the spectrum corresponds to flattening the input multilayer graph and resorting to single-layer densest-subgraph extraction. However, this would mean disregarding the different semantics of the layers, incurring in a severe information loss.
Within this view, in this work we generalize the problem studied in~\cite{charikar2018finding,jethava2015finding,reinthal2016finding,semertzidis2016best} by introducing a formulation that accounts for a trade-off between high density  and number of layers exhibiting the high density.
}
Specifically, given a multilayer graph $G=(V,E,L)$, the average-degree density of a subset of vertices $S$ in a layer $\ell$ is defined as the number of edges induced by $S$ in $\ell$ divided by the size of $S$, i.e., $\frac{|E_{\ell}[S]|}{|S|}$.
%
%
We define the \emph{multilayer densest subgraph} as the subset of vertices $S^*$ such that the function
$$
  \max_{\hat{L} \subseteq L} \min_{\ell\in \hat{L}} \frac{\textstyle|E_{\ell}[S^*]|}{\textstyle|S^*|}{\textstyle |\hat{L}|^\beta}
$$
is maximized.
Parameter $\beta \in \mathbb{R}^+$ controls the importance of the two problem ingredients, i.e., high density and number of high-density layers.
This problem statement naturally achieves the aforementioned desired trade-off: the larger the subset $\hat{L}$ of selected layers, the smaller the minimum density $\min_{\ell\in \hat{L}} \frac{|E_{\ell}[S]|}{|S|}$ in those layers.
%
Similarly to the single-layer case where core decomposition provides a $\frac{1}{2}$-approximation of the densest subgraph, in this work we show that
computing the multilayer core decomposition of the input graph and selecting the core maximizing the proposed multilayer density function achieves a $\frac{1}{2|L|^\beta}$-approximation for the general multilayer-densest-subgraph problem formulation, and a $\frac{1}{2}$-approximation for the all-layer variant in~\cite{charikar2018finding,jethava2015finding,reinthal2016finding,semertzidis2016best}.

As a further application of  multilayer core decomposition, we show how it can speed up  \emph{frequent cross-graph quasi-clique} extraction~\cite{jiang2009mining}.
We prove that searching for frequent cross-graph quasi-cliques in restricted areas of the graph -- corresponding to multilayer cores complying with the quasi-clique condition -- is still sound and complete, while also being much more efficient.

Finally, we also provide a generalization of the \emph{community-search} problem~\cite{Sozio} to the multilayer setting, and show how to exploit multilayer core decomposition to optimally solve this problem.

\medskip

Summarizing, this work has the following contributions:
\begin{enumerate}
\item
We define the problem of \emph{core decomposition in multilayer networks},
and characterize it in terms of relation to other problems, and complexity.
We devise three algorithms that solve multilayer core decomposition efficiently (Section~\ref{sec:algorithms}).

\item We devise further algorithms to compute the \emph{inner-most cores} only (Section~\ref{sec:innermost}).

\item We study 
the problem of \emph{densest-subgraph in multilayer networks}.
We introduce a formulation that trades-off between high density and number of layers exhibiting high density, and 
exploit multilayer core decomposition to solve it with approximation guarantees (Section~\ref{sec:densest}).

\item We show how the multilayer core-decomposition tool can be exploited to speed up the extraction of \emph{frequent cross-graph quasi-cliques} (Section~\ref{sec:quasicliques}).

\item We formulate the \emph{multilayer community-search} problem and show that multilayer core decomposition provides an optimal solution to this problem (Section~\ref{sec:communitysearch}).
\end{enumerate}

We also provide extensive experiments, on numerous real datasets, to assess the performance of our proposals.
For each  aforementioned context, experiments are provided within the corresponding section.
A preliminary version of this work, covering Sections~\ref{sec:algorithms}~and~\ref{sec:densest} only, was presented in~\cite{galimberti2017core}.

\smallskip
\spara{Reproducibility.}
For the sake of reproducibility all our code and some of the datasets used in this paper are available at \url{https://github.com/egalimberti/multilayer_core_decomposition}


\section{Preliminaries and problem statements}
\label{sec:problems}

In this section we introduce the needed preliminaries and notation, we provide some fundamental properties of multilayer cores, and then formally define all the problems studied in this work.

\subsection{Multilayer core decomposition}
We are given an undirected  multilayer graph $G = (V,E,L)$, where $V$ is a set of vertices, $L$ is a set of layers, and $E \subseteq V \times V \times L$ is a set of edges.
Let $E_{\ell}$ denote the subset of edges in layer $\ell \in L$.
For a vertex $u \in V$ we denote by $\deg(u, \ell)$ and $\deg(u)$  its degree in layer $\ell$ and over all layers, respectively, i.e., $\deg(u,\ell) = |\{e = (u,v,\ell) : e \in E_{\ell}\}|$, $\deg(u) = |\{e = (u,v,\ell)  : e \in E\}| = \sum_{\ell \in L} \deg(u,\ell)$.

For a subset of vertices $S \subseteq V$ we denote by $G[S]$ the subgraph of $G$ induced by $S$, i.e., $G[S] = (S,E[S],L)$, where $E[S] = \{ e = (u,v,\ell) \mid e \in E, u \in S, v \in S \}$.
For a vertex $u \in V$ we denote by $\deg_S(u, \ell)$ and $\deg_S(u)$  its degree in subgraph $S$ considering layer $\ell$ only and all layers, respectively, i.e., $\deg_S(u,\ell) = |\{e = (u,v,\ell) : e \in E_{\ell}[S]\}|$, $\deg_S(u) = |\{e = (u,v,\ell) : e \in E[S]\}| = \sum_{\ell \in L} \deg_S(u,\ell)$.
Finally, let $\mu(\ell)$ and $\mu(\hat{L})$ denote the minimum degree of a vertex in layer $\ell$ and in a subset $\hat{L} \subseteq L$ of layers, respectively. Let also $\mu(S, \ell)$ and $\mu(S, \hat{L})$ denote the corresponding counterparts of $\mu(\ell)$ and $\mu(\hat{L})$ for a subgraph (induced by a vertex set) $S$.


A core of a multilayer graph is characterized by an $|L|$-dimensional integer vector $\vec{k} = [k_{\ell}]_{\ell \in L}$, termed \emph{\corenessvec}, whose components $k_{\ell}$ denote the minimum degree allowed in layer $\ell$.
\revision{
This corresponds to the notion of \mlkcore\ introduced by Azimi-Tafreshi~\emph{et al.}~\cite{MultiplexCores}.
As discussed in Section~\ref{sec:related}, Azimi-Tafreshi~\emph{et al.} do not study (or devise any algorithm for) the problem of computing the entire multilayer core decomposition. They study core percolation by analyzing \emph{a single} core of interest, computed with the simple iterative-peeling algorithm (Algorithm~\ref{alg:core}).
}
Formally:
\revision{
\begin{mydefinition}[multilayer core \mbox{\normalfont and} \corenessvec~\cite{MultiplexCores}]\label{def:mlkcores}
	Given a multilayer graph $G = (V, E, L)$ and an $|L|$-dimensional integer vector $\vec{k} = [k_{\ell}]_{\ell \in L}$, the multilayer $\vec{k}$\emph{-core} of $G$ is a
	\emph{maximal} subgraph $G[C] = (C \subseteq V, E[C], L)$ such that $\forall \ell \in L:
	\mu(C, \ell) \geq k_{\ell}$.
	The vector $\vec{k}$ is referred to as the \emph{\corenessvec} of $G[C]$.
\end{mydefinition}
}

Given a \corenessvec\ $\vec{k}$, we denote by $C_{\vec{k}}$ the corresponding core.
Also, as a \mlkcore\ is fully identified by the vertices belonging to it, we hereinafter refer to it by its vertex set $C_{\vec{k}}$ and the induced subgraph $G[C_{\vec{k}}]$ interchangeably.
It is important noticing that a set of vertices $C \subseteq V$ may correspond to multiple cores.
For instance, in the graph in Figure~\ref{fig:mlcores} the set $\{\mbox{\textsf{A,B,D,E}}\}$ corresponds to both $(3,0)$-core and $(3,1)$-core.
In other words, a multilayer core can be described by more than one \corenessvec. However,
as formally shown next, among such multiple \corenessvec{s} there exists one and only one that is not dominated by any other. We call this vector the \emph{maximal \corenessvec} of $C$.
In the example in Figure~\ref{fig:mlcores} the maximal \corenessvec\ of $\{\mbox{\textsf{A,B,D,E}}\}$ is $(3,1)$.

\begin{mydefinition}[maximal \corenessvec]\label{def:maximalcorenessvec}
	Let $G = (V, E, L)$ be a multilayer graph, $C \subseteq V$ be a core of $G$, and $\vec{k} = [k_{\ell}]_{\ell \in L}$ be a \corenessvec\ of $C$.
	$\vec{k}$ is said \emph{maximal} if there does not exist any \corenessvec\ $\vec{k'} = [k'_{\ell}]_{\ell \in L}$ of $C$ such that $\forall \ell \in L : k'_{\ell} \geq k_{\ell}$ and $\exists \hat{\ell} \in L : k'_{\hat{\ell}} > k_{\hat{\ell}}$.
\end{mydefinition}

\begin{figure}[t]
	\centering
	\begin{tikzpicture} [scale=0.20, every node/.style={circle, draw, scale=0.65}]
	\node (A) at (0,10) {\textsf{A}};
	\node (B) at (10,10) {B};	
	\node (C) at (20,10) {C};
	\node (D) at (0,0) {D};
	\node (E) at (10,0) {E};
	\node (F) at (20,0) {F};
	
	\draw (A) edge (B);
	\draw (A) edge (D);
	\draw (A) edge (E);
	\draw (B) edge (C);
	\draw (B) edge (D);
	\draw (B) edge (E);
	\draw (B) edge (F);
	\draw (D) edge (E);
	\draw (E) edge (F);
	
	\draw [dashed, bend left=10] (A) edge (B);
	\draw [dashed, bend left=10] (B) edge (C);
	\draw [dashed, bend right=10] (B) edge (D);
	\draw [dashed, bend left=10] (B) edge (E);
	\draw [dashed, bend left=10] (B) edge (F);
	\draw [dashed] (C) edge (E);
	\draw [dashed] (C) edge (F);
	\draw [dashed, bend right=10] (E) edge (F);
\end{tikzpicture}
	\caption{\label{fig:mlcores}{  Example 2-layer graph (solid edges refer to the first layer, while dashed edges to the second layer) with the following \mlkcores:
			$(0,0) = (1,0) = (0,1) = (1,1) = \{ \mbox{\textsf{A,B,C,D,E,F}} \},
			(2,0) = (2,1) = \{ \mbox{\textsf{A,B,D,E,F}} \},
			(3,0) = (3,1) = \{ \mbox{\textsf{A,B,D,E}} \},
			(0,2) = (1,2) = (0,3) =  (1,3) = \{ \mbox{\textsf{B,C,E,F}} \},
			(2,2) = \{ \mbox{\textsf{B,E,F}} \}$
			.}}
\end{figure}
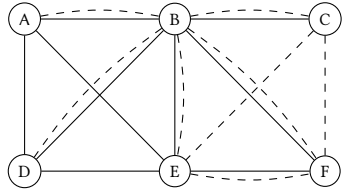

\begin{mytheorem}\label{th:maximalcorenessvecuniqueness}
	Multilayer cores have a unique maximal \corenessvec.
\end{mytheorem}
\begin{proof}
We prove the theorem by contradiction.
Assume two maximal \corenessvec{s} $\vec{k} = [k_{\ell}]_{\ell \in L} \neq \vec{k}' = [k'_{\ell}]_{\ell \in L}$ exist for a multilayer core $C$.
As $\vec{k} \neq \vec{k}'$ and they are both maximal, there exist two layers $\hat{\ell}$ and $\bar{\ell}$ such that $k_{\hat{\ell}} > k'_{\hat{\ell}}$ and $k'_{\bar{\ell}} > k_{\bar{\ell}}$.
By definition of multilayer core (Definition~\ref{def:mlkcores}), it holds that $\forall \ell \in L : \mu(C,\ell) \geq k_\ell, \mu(C,\ell) \geq k'_\ell$.
This means that the vector $\vec{k}^* = [k^*_{\ell}]_{\ell \in L}$, with $k^*_{\ell} = \max\{k_{\ell}, k'_{\ell}\}, \forall \ell \in L$, is a further \corenessvec\ of $C$.
For this vector it holds that $\forall \ell \neq \hat{\ell}, \ell \neq \bar{\ell} : k^*_{\ell} \geq k'_{\ell}$, $k^*_{\hat{\ell}} > k'_{\hat{\ell}}$, and $k^*_{\bar{\ell}} > k_{\bar{\ell}}$.
Thus, $\vec{k}^*$ dominates both $\vec{k}$ and $\vec{k}'$, which contradicts the hypothesis of maximality of $\vec{k}$ and $\vec{k}'$.
The theorem follows.
\end{proof}

The first (and main) problem we tackle in this work is the computation of the complete multilayer core decomposition, i.e., the set of all non-empty multilayer cores.

\begin{problem}[\mlcoredecomposition]\label{prob:mlcoredecomposition}
	Given a multilayer graph $G = (V, E, L)$, find the set of all \emph{non-empty} and \emph{distinct} cores of $G$, along with their corresponding maximal \corenessvec{s}.
	Such a set forms what we hereinafter refer to as the  \emph{multilayer core decomposition} of $G$.
\end{problem}

\subsection{Inner-most multilayer cores}
Cores of a single-layer graph are all nested one into another.
This makes it possible to define the notions of  $(i)$ \emph{inner-most core}, as the core of highest order, and $(ii)$ \emph{core index} (or \emph{core number}) of a vertex $u$, which is the highest order of a core containing $u$.
In the multilayer setting the picture is more complex, as multilayer cores are not all nested into each other.
As a result, the core index of a vertex is not unambiguously defined, while there can exist multiple inner-most cores:

\begin{mydefinition}[inner-most multilayer cores]\label{def:innermostcores}
	The \emph{inner-most cores} of a multilayer graph are all those cores with maximal \corenessvec\ $\vec{k} = [k_{\ell}]_{\ell \in L}$ such that there does not exist any other core with \corenessvec\  $\vec{k}' = [k'_{\ell}]_{\ell \in L}$ where $\forall \ell \in L : k'_{\ell} \geq k_{\ell}$ and $\exists \hat{\ell} \in L : k'_{\hat{\ell}} > k_{\hat{\ell}}$.
\end{mydefinition}

To this purpose, look at the example in Figure~\ref{fig:mlcores}.
It can be observed that: ($i$) cores are not nested into each other, ($ii$) $(3,1)$-core, $(1,3)$-core and $(2,2)$-core are the inner-most cores, and ($iii$) vertices \textsf{B} and \textsf{E} belong to (inner-most) cores $(3,1)$, $(1,3)$, and $(2,2)$, thus making their core index not unambiguously defined.

The second problem we tackle in this work is the development of smart algorithms  to compute all the inner-most cores, without the need of computing the complete multilayer core decomposition.

 \begin{problem}[Inner-most cores computation]\label{prob:innermost}
	Given a multilayer graph $G = (V, E, L)$, find all \emph{non-empty} and inner-most cores of $G$, along with their corresponding maximal \corenessvec{s}.
	\end{problem}

\subsection{Multilayer densest subgraph}
As anticipated in Section~\ref{sec:contributions}, the densest subgraph of a multilayer graph should provide a good trade-off between large density and the number of layers where such a large density is exhibited.
We achieve this intuition by means of the following optimization problem:

\begin{problem}[\mldensestsubgraph]\label{prob:mldensestsubgraph}
Given a multilayer graph $G=(V,E,L)$, a positive real number $\beta$, and a real-valued function $\density : 2^{V} \rightarrow \mathbb{R}^+$ defined as:
\begin{equation}\label{eq:densestfunction}
\density(S) = \max_{\hat{L} \subseteq L} \min_{\ell \in \hat{L}} \frac{|E_{\ell}[S]|}{|S|}  |\hat{L}|^{\beta},
\end{equation}
find a subset $S^* \subseteq V$ of vertices that maximizes function $\density$, i.e.,
$$
S^* = \arg\max_{S \subseteq V}  \density(S).
$$
\end{problem}

Parameter $\beta$ controls the importance of the two ingredients of the objective function $\density$, i.e., density 
and number of layers exhibiting such a density: the smaller $\beta$, the better the focus on the former aspect (density), and vice versa.
Also, as a nice side effect, solving \mldensestsubgraph\  allows for automatically finding  a set of layers of interest for the densest subgraph $S^*$. In Section~\ref{sec:densest} we will show how to exploit it to devise an algorithm with approximation guarantees for \mldensestsubgraph, thus extending to the multilayer case the intuition at the basis of the well-known $\frac{1}{2}$-approximation algorithm~\cite{AITT00,Char00} for single-layer densest subgraph.

\subsection{Frequent cross-graph quasi-cliques}
Another interesting insight into the notion of multilayer cores is about their relationship with (quasi-)cliques.
In single-layer graphs it is well-known that cores can speed-up clique finding, as a clique of size $k$ is  contained in the $(k-1)$-core.
Interestingly, a similar relationship holds in the multilayer context too.
Given a multilayer graph $G = (V,E,L)$, a layer $\ell \in L$, and a real number $\gamma \in (0,1]$, a subgraph $G[S] = (S \subseteq V, E[S], L)$ of $G$ is said to be a $\gamma$\emph{-quasi-clique} in layer $\ell$ if all its vertices have at least $\gamma (|S| -1)$ neighbors in layer $\ell$ within $S$, i.e., $\forall u \in S : deg_S(u, \ell) \geq \gamma (|S| -1)$.
Jiang~\emph{et~al.}~\cite{jiang2009mining} study the problem of extracting \emph{frequent cross-graph quasi-cliques}, defined next.

 \begin{problem}[Frequent cross-graph quasi-cliques mining~\cite{jiang2009mining}]\label{prob:cgqc}
Given a multilayer graph $G = (V, E, L)$,	a function $\Gamma : L \rightarrow (0,1]$ assigning a real value to every layer in $L$, a real number $\mbox{\emph{min\_sup}} \in (0,1]$, and an integer $\mbox{\emph{min\_size}} \geq 1$, find all maximal subgraphs $G[S]$ of $G$ of size larger than  \emph{min\_size} such that there exist at least $\mbox{\emph{min\_sup}}\times|L|$ layers $\ell$ for which $G[S]$ is a $\Gamma(\ell)$-quasi-clique.
	\end{problem}

In Section~\ref{sec:quasicliques} we will prove that a frequent cross-graph quasi-clique of size $K$ is necessarily contained into a $\vec{k}$-core described by a maximal coreness vector $\vec{k}=[k_{\ell}]_{\ell \in L}$ such that there exists a fraction of at least \emph{min\_sup} layers $\ell$ where $k_{\ell} = \lfloor \Gamma(\ell)(K - 1)\rfloor$. Based on this, exploiting multilayer core decomposition as a preprocessing step allows for speeding up any algorithm for Problem~\ref{prob:cgqc}.

\subsection{Multilayer community search}
The last application we study is the problem of  \emph{community search}.
Given a graph $G=(V,E)$ and a set $V_Q \subseteq V$ of query vertices, a very wide family of problem requires to find a connected subgraph $H$ of $G$, which contains all query vertices $V_Q$ and exhibits an adequate degree of cohesiveness, compactness, or density. This type of problem has been termed in the literature in different ways, e.g., \emph{community search}~\cite{Sozio,SozioLocalSIGMOD14,BarbieriBGG15},
\emph{seed set expansion} \cite{Andersen1,Kloumann}, \emph{connectivity subgraphs} \cite{connect,CenterpieceKDD06,ruchansky2015minimum,akoglu2013mining,RuchanskyBGGK17}, just to mention a few: see~\cite{HuangLX17} for a recent survey.
In this work we adopt the early definition by Sozio and Gionis~\cite{Sozio} which measures the cohesiveness of the resulting subgraph by means of the minimum degree inside the subgraph, and we adapt it to the multilayer setting as follows.

\begin{problem}[Multilayer Community Search]\label{prob:mlcs}
Given a multilayer graph $G=(V,E,L)$, a set  $S \subseteq V$ of vertices, and a set $\hat{L} \subseteq L$ of layers, the minimum degree in the subgraph induced by $S$~and~$\hat{L}$~is:
$$
\varphi(S,\hat{L}) = \min_{\ell \in \hat{L}} \min_{u \in S} deg_S(u,\ell).
$$

\noindent Given a positive real number $\beta$, we define a real-valued density function $\vartheta : 2^{V} \rightarrow \mathbb{R}^+$ as:
$$
\vartheta(S) = \max_{\hat{L} \subseteq L}  \varphi(S,\hat{L})  |\hat{L}|^{\beta}.
$$

Given a set $V_Q \subseteq V$ of query vertices, find a subgraph containing all the query vertices and maximizing the density function, i.e.,
\begin{equation}
S^* = \argmax_{V_Q \subseteq S \subseteq V}  \vartheta(S).
\end{equation}
\end{problem}

In Section~\ref{sec:communitysearch} we will show how to adapt multilayer core decomposition to solve Problem~\ref{prob:mlcs}.

\section{Algorithms for multilayer core decomposition}
\label{sec:algorithms}

A major challenge of  \mlcoredecomposition\ is that the number of multilayer cores may be exponential in the number of layers. Specifically, denoting by $K_{\ell}$ the maximum order of a core for layer $\ell$, the number of multilayer cores is $\mathcal{O}(\prod_{\ell \in L}K_{\ell})$. This makes \mlcoredecomposition\ intrinsically hard: \emph{in the general case, no polynomial-time algorithm can exist}.
The challenge hence lies in  handling this exponential blowup by early recognizing and skipping unnecessary portions of the search space, such as  non-distinct and/or empty cores.

Given a multilayer graph $G = (V,E,L)$ and a \corenessvec\ $\vec{k} = [k_{\ell}]_{\ell \in L}$, finding the corresponding core can easily be solved in $\mathcal{O}(|E| + |V|\times|L|)$ time by iteratively removing a vertex $u$ having $\deg_{G'}(u, \ell) < k_{\ell}$ in some layer $\ell$, where $G'$ denotes the current graph resulting from all previous vertex removals (Algorithm~\ref{alg:core}, where the set $S$ of vertices to be considered is set to $S = V$). Hence, a na\"{\i}ve way of computing a multilayer core decomposition consists of generating all possible \corenessvec{s}, run the multilayer core-detection algorithm just described for each vector, and retain only non-empty and distinct cores.
This na\"{\i}ve method requires all vectors $[k_{\ell}]_{\ell \in L}$, where each $k_{\ell}$ component is varied within the interval $[0..K_{\ell}]$.\footnote{$K_{\ell}$ values can be derived beforehand by computing a single-layer core decomposition in each layer $\ell$. This process overall takes $\mathcal{O}(|E|)$ time.}
This corresponds to a $\Theta(\prod_{\ell \in L}K_{\ell})$ number of vectors.
As a result, the overall time complexity of the method is $\mathcal{O}\big((|E| + |V|\times|L|) \times \prod_{\ell \in L}K_{\ell}\big)$.

This approach has two major weaknesses: ($i$) each core is computed starting from the whole input graph, and ($ii$) by enumerating all possible \corenessvec{s} beforehand a lot of non-distinct and/or empty (thus, unnecessary) cores may be computed.
In the following we present three methods that solve  \mlcoredecomposition\ much more efficiently.

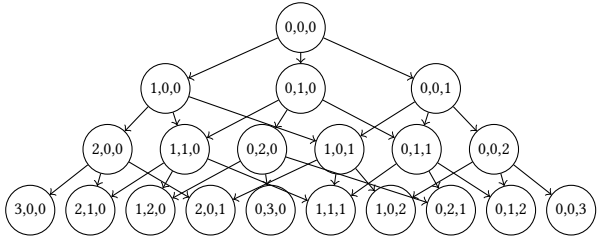
\begin{figure}[t]
	\centering
	\begin{tikzpicture} [scale=0.08, every node/.style={circle, draw, scale=0.65}]
	\node (300) at (0,0) {3,0,0};
	\node (030) at (40,0) {0,3,0};	
	\node (003) at (90,0) {0,0,3};
	\node (210) at (10,0) {2,1,0};
	\node (201) at (30,0) {2,0,1};
	\node (012) at (80,0) {0,1,2};
	\node (102) at (60,0) {1,0,2};
	\node (120) at (20,0) {1,2,0};
	\node (021) at (70,0) {0,2,1};
	\node (111) at (50,0) {1,1,1};
	
	\node (200) at (12.86,10) {2,0,0};
	\node (110) at (25.72,10) {1,1,0};
	\node (020) at (38.58,10) {0,2,0};
	\node (101) at (51.44,10) {1,0,1};
	\node (011) at (64.3,10) {0,1,1};
	\node (002) at (77.16,10) {0,0,2};
	
	\node (100) at (22.5,20) {1,0,0};
	\node (010) at (45,20) {0,1,0};
	\node (001) at (67.5,20) {0,0,1};
	
	\node (000) at (45,30) {0,0,0};
	
	\draw [->] (000) edge (100);
	\draw [->] (000) edge (010);
	\draw [->] (000) edge (001);
	
	\draw [->] (100) edge (200);
	\draw [->] (100) edge (110);
	\draw [->] (100) edge (101);
	\draw [->] (010) edge (110);
	\draw [->] (010) edge (020);
	\draw [->] (010) edge (011);
	\draw [->] (001) edge (101);
	\draw [->] (001) edge (011);
	\draw [->] (001) edge (002);
	
	\draw [->] (200) edge (300);
	\draw [->] (200) edge (210);
	\draw [->] (200) edge (201);
	\draw [->] (110) edge (210);
	\draw [->] (110) edge (120);
	\draw [->] (110) edge (111);
	\draw [->] (020) edge (120);	
	\draw [->] (020) edge (030);
	\draw [->] (020) edge (021);
	\draw [->] (101) edge (201);
	\draw [->] (101) edge (111);
	\draw [->] (101) edge (102);	
	\draw [->] (011) edge (111);
	\draw [->] (011) edge (021);
	\draw [->] (011) edge (012);
	\draw [->] (002) edge (102);
	\draw [->] (002) edge (012);
	\draw [->] (002) edge (003);
\end{tikzpicture}
	\caption{\label{fig:dagstr} Core lattice of a 3-layer graph.}
\end{figure}

\subsection{Search space}
Although multilayer cores are not all nested into each other, a notion of partial containment can still be defined.
Indeed, it can easily be observed that a $\vec{k}$-core with \corenessvec\ $\vec{k} = [k_{\ell}]_{\ell \in L}$ is contained into any $\vec{k}'$-core described by a \corenessvec\ $\vec{k}' = [k'_{\ell}]_{\ell \in L}$ whose components $k'_{\ell}$ are all no more than components $k_{\ell}$, i.e., $k'_{\ell} \leq k_{\ell}$, $\forall \ell \in L$.
This result is formalized next:

\begin{fact} \label{fact:nestedmlcores}
	Given a multilayer graph $G = (V,E,L)$ and two cores $C_{\vec{k}}$ and $C_{\vec{k'}}$ of $G$ with \corenessvec{s} $\vec{k} = [k_{\ell}]_{\ell \in L}$ and $\vec{k}' = [k'_{\ell}]_{\ell \in L}$, respectively, it holds that if $\forall \ell \in L : k'_{\ell} \leq k_{\ell}$, then $C_{\vec{k}} \subseteq C_{\vec{k'}}$.
\end{fact}
\begin{proof}
Combining the definition of multilayer core (Definition~\ref{def:mlkcores}) and the hypothesis on vectors $\vec{k}$ and $\vec{k}'$, it holds that $\forall \ell \in L : \mu(C_{\vec{k}}, \ell) \geq k_\ell \geq k'_{\ell}$.
This means that $C_{\vec{k}}$ satisfies the definition of $\vec{k}'$-core, thus implying that all vertices in $C_{\vec{k}}$ are part of $C_{\vec{k}'}$ too.
The fact follows.
\end{proof}

%

Based on Fact~\ref{fact:nestedmlcores},
the search space of our problem can be represented as a lattice defining a partial order among all cores (Figure~\ref{fig:dagstr}).
Such a lattice, which we call the \emph{core lattice}, corresponds to a \textsc{dag} where nodes represent cores,\footnote{Throughout the paper we use the term ``node'' to refer to elements of the core lattice, and ``vertex'' for the elements of the multilayer graph.}
and links represent relationships of containment between cores (a ``father'' node contains all its ``child'' nodes).
We assume the core lattice keeping track of non-empty and not necessarily distinct cores: a core is present in the lattice as many times as the number of its \corenessvec{s}.
%
Each level $i$ of the lattice represents the children of cores at lattice level $i-1$.
In particular, level $i$ contains all those cores whose \corenessvec\ results from increasing one and only one component of its fathers' \corenessvec\ by one.
Formally, a lattice level $i$ contains all \mlkcores\ with \corenessvec\ $\vec{k} = [k_{\ell}]_{\ell \in L}$ such that there exists a core at lattice level $i-1$ with \corenessvec\ $\vec{k}' = [k'_{\ell}]_{\ell \in L}$ where: $\exists \ell \in L : k_{\ell} =  k'_{\ell} + 1$, and $\forall \hat{\ell} \neq \ell : k_{\hat{\ell}} =  k'_{\hat{\ell}}$.
As a result, level 0 contains the root only, which corresponds to the whole input graph (i.e., the $[0]_{|L|}$-core), the leaves correspond to inner-most cores, and any non-leaf node has at least one and at most $\left| L \right|$ children.
Moreover, every level $i$ contains all cores whose coreness-vector components sum to $i$.


Solving the \mlcoredecomposition\ problem is hence equivalent to building the core lattice of the input graph.
The efficient methods we present next are all based on smart core-lattice building strategies that  extract cores from smaller subgraphs, while also attempting to minimize the visit/computation of unnecessary (i.e., empty/non-distinct) cores.

\begin{algorithm}[t]
	\caption{\kcorealg} \label{alg:core}
	\begin{algorithmic}[1]
		\REQUIRE A multilayer graph $G = (V,E,L)$, a set $S \subseteq V$ of vertices, an $|L|$-dimensional integer vector $\vec{k} = [k_{\ell}]_{\ell \in L}$.
		\ENSURE The $\vec{k}$-core $C_{\vec{k}}$ of $G$.
		
		\WHILE{$\exists u \in S, \exists \ell \in L : \deg_S(u, \ell) < k_{\ell}$}
		\STATE $S \gets S \setminus \{u\}$
		\ENDWHILE
		\STATE $C_{\vec{k}} = S$
		
	\end{algorithmic}
\end{algorithm}

\subsection{Breadth-first algorithm}
Two interesting corollaries can be derived from Fact \ref{fact:nestedmlcores}. First, any non-empty \mlkcore\ is necessarily contained in the intersection of all its father nodes of the core lattice.
Second, any non-empty \mlkcore\ has \emph{exactly} as many fathers as the number of non-zero components of its \corenessvec\ $\vec{k}$:

\begin{algorithm}[t]
	\caption{\bfs} \label{alg:bfs}
	\begin{algorithmic}[1]
		\REQUIRE A multilayer graph $G = (V,E,L)$.
		\ENSURE The set \coresset\ of all non-empty multilayer cores of $G$.
		
		\STATE $\coresset \!\leftarrow\! \emptyset$, $\mathbf{Q} \!\leftarrow\! \{[0]_{|L|}\}$, $\mathcal{F}([0]_{|L|}) \!\leftarrow\! \emptyset$ \COMMENT{$\mathcal{F}$ keeps track of father nodes}
		
        \WHILE {$\mathbf{Q} \neq \emptyset$}
            \STATE dequeue $\vec{k} = [k_{\ell}]_{\ell \in L}$ from $\mathbf{Q}$

            \IF [Corollary~\ref{cor:fathernumber}]{$|\{k_{\ell} : k_{\ell} > 0\}| = |\mathcal{F}(\vec{k})|$}
                	\STATE $F_{\cap} \leftarrow \bigcap_{F \in \mathcal{F}(\vec{k})} F$ \COMMENT{Corollary~\ref{cor:coreintersection}}
            		\STATE $C_{\vec{k}} \leftarrow \kcorealg(G,F_{\cap},\vec{k})$ \COMMENT{Algorithm~\ref{alg:core}}
            		
                    \IF {$C_{\vec{k}} \neq \emptyset$} \label{line:bfs:exists}
                        \STATE $\coresset \! \leftarrow \!\coresset  \cup  \{ C_{\vec{k}} \}$
                        \FORALL [enqueue child nodes]{$\ell \in L$}
                            \STATE $\vec{k}' \leftarrow [k_1, \ldots, k_{\ell} + 1, \ldots, k_{|L|}]$
                            \STATE enqueue $\vec{k}'$ into $\mathbf{Q}$
                            \STATE $\mathcal{F}(\vec{k}') \leftarrow \mathcal{F}(\vec{k}') \cup \{C_{\vec{k}}\}$
                        \ENDFOR
                    \ENDIF
            \ENDIF
        \ENDWHILE
	\end{algorithmic}
\end{algorithm}

\begin{mycorollary}\label{cor:coreintersection}
	Given a multilayer graph $G$, let $C$ be a core of $G$ and $\mathcal{F}(C)$ be the set of fathers of $C$ in the core lattice of $G$.
	It holds that $C \subseteq \bigcap_{\hat{C} \in \mathcal{F}(C)} \hat{C}$. 
\end{mycorollary}
\begin{proof}
By definition of core lattice, the \corenessvec\ of all father cores $\mathcal{F}(C)$ of $C$ is dominated by the \corenessvec\ of $C$.
Fact \ref{fact:nestedmlcores} ensures that $C \subseteq C'$, $\forall C' \in \mathcal{F}(C)$.
Assume $u \in C$, $u \notin \bigcap_{\hat{C} \in \mathcal{F}(C)} \hat{C}$ exists.
This implies existence of a father core $C' \in \mathcal{F}(C)$ s.t. $C \not\subseteq C'$, which is a contradiction.
\end{proof}

\begin{mycorollary}\label{cor:fathernumber}
	Given a multilayer graph $G$, let $C$ be a core of $G$ with \corenessvec\ $\vec{k} = [k_{\ell}]_{\ell \in L}$, and $\mathcal{F}(C)$ be the set of fathers of $C$ in the core lattice of $G$.
	It holds that $|\mathcal{F}(C)| = |\{k_{\ell} : \ell \in L, k_{\ell} > 0\}|$.
\end{mycorollary}
\begin{proof}
By definition of core lattice, a core $C$ at level $i$ has a \corenessvec\  whose components sum to $i$, while the fathers $\mathcal{F}(C)$ of $C$ have \corenessvec\ whose components sum to $i-1$. 
The \corenessvec\ of a father of $C$ can be obtained by decreasing a non-zero component of the \corenessvec\ of $C$ by one. This means that the number of fathers of $C$ is upper-bounded by the non-zero components of its \corenessvec.
More precisely, the number of fathers of $C$ is exactly equal to this number, as, according to Corollary~\ref{cor:coreintersection}, no father of $C$ can be empty, otherwise $C$ would be empty too and would not be part of the core lattice.
\end{proof}

The above corollaries pave the way to a breadth-first search building strategy of the core lattice, where cores are  generated level-by-level by properly exploiting the rules in the two corollaries (Algorithm~\ref{alg:bfs}).
Although the worst-case time complexity of  this \bfs\  method remains unchanged with respect to the na\"{\i}ve algorithm, the \textsc{bfs} method is expected to be much more efficient in practice, due to the following main features: ($i$) cores are not computed from the initial graph every time, but from a much smaller subgraph given by the intersection of all their fathers; ($ii$) in many cases, i.e., when the rule in Corollary 2 (which can be checked in constant time) arises, no overhead due to the intersection among father cores is required; ($iii$) the number of empty cores computed is limited, as no empty core may be generated from a core that has already been recognized as empty.

\subsection{Depth-first algorithm}
\label{sec:alg:dfs}
Although being much smarter than the na\"ive method, \bfs\ still has some limitations.
First, it visits every core as many times as the number of its fathers in the core lattice.
Also, as a second limitation, consider a path $\mathcal{P}$ of the  lattice connecting a non-leaf node to a leaf by varying the same $\ell$-th component of the corresponding \corenessvec{s}.
It is easy to see that the computation of all cores within $\mathcal{P}$ with \bfs\ takes $\mathcal{O}(|\mathcal{P}| \times (|E| + |V|\times|L|))$ time, as the core-decomposition process is re-started at every level of the lattice.
This process can in principle be performed more efficiently, i.e., so as to take $\mathcal{O}(|\mathcal{P}| + |E| + |V|\times|L|)$ time, as it actually corresponds to (a simple variant of) a single-layer core decomposition.

\begin{algorithm}[t]
	\caption{\dfs} \label{alg:dfs}
	\begin{algorithmic}[1]
		\REQUIRE A multilayer graph $G = (V,E,L)$.
		\ENSURE The set \coresset\ of all non-empty multilayer cores of $G$.
		
		\STATE $\coresset  \leftarrow \{V\}$, \ \ $R \leftarrow L$, \ \ $\mathbf{Q} \leftarrow \{[0]_{|L|}\}$, \ \ $\mathbf{Q}' \leftarrow \emptyset$
		
		\WHILE {$R \neq \emptyset$}
    		\STATE remove a layer from $R$
    		\FORALL {$\vec{k} \in \mathbf{Q}$}
        		\STATE $\forall \ell \!\in\! R$ s.t. $k_{\ell} = 0: \mathbf{Q}' \leftarrow \mathbf{Q}'  \cup \{\vec{k}' \mid C_{\vec{k}'} \in \kcorespathalg(G,C_{\vec{k}},\vec{k},\ell)\}$ \label{line:dfs:queue}
        		\STATE $\forall \ell \!\in\! L\setminus R$ s.t. $k_{\ell} = 0: \coresset \leftarrow \coresset \cup \kcorespathalg(G,C_{\vec{k}},\vec{k},\ell)$ \label{line:dfs:cores}
    		\ENDFOR
    		\STATE $\coresset  \leftarrow \coresset  \cup  \{C_{\vec{k}} \mid \vec{k} \in \mathbf{Q}'\}$, \ \ $\mathbf{Q} \leftarrow \mathbf{Q}'$, \ \ $\mathbf{Q}' \leftarrow \emptyset$
		\ENDWHILE
	\end{algorithmic}
\end{algorithm}

To address the two above cons, we propose a method performing a depth-first search on the core lattice.
The method, dubbed \dfs\  (Algorithm~\ref{alg:dfs}), iteratively picks a non-leaf core $\vec{k} = [k_1, \ldots, k_{\ell}, \ldots, k_{|L|}]$ and a layer $\ell$ such that $k_\ell = 0$, and computes all cores $[k_1, \ldots, k_{\ell} + 1, \ldots, k_{|L|}], \ldots, [k_1, \ldots, K_{\ell}, \ldots, k_{|L|}]$ with a run of
the $\kcorespathalg(G,C_{\vec{k}},\vec{k},\ell)$ subroutine. 
\revision{
Specifically, such a subroutine returns the cores corresponding to all \corenessvec{s} obtained by varying the $\ell$-th component of $\vec{k}$ within $[0,\ldots,K_{\ell}]$.
Also, it discards vertices violating the coreness condition specified by vector $\vec{k}$, i.e., vertices whose degree in some layer $\hat{\ell} \neq \ell$ is less than the $\hat{\ell}$-th component of $\vec{k}$.
The pseudocode of \kcorespathalg is reported as Algorithm~\ref{alg:corespath}: it closely resembles the traditional core-decomposition algorithm for single-layer graphs, except for the addition of the cycle starting at Line~\ref{line:corespath:otherlayers}, which identifies the aforementioned vertices to be discarded.
}

A side effect of this strategy is that the same core may be computed multiple times. As an example, in Figure~\ref{fig:dagstr} the $(1,2,0)$-core is computed by core decompositions initiated at both cores $(1,0,0)$ and $(0,2,0)$.
To reduce (but not eliminate) these multiple core computations, the \dfs\ method exploits the following result.


\begin{mytheorem}\label{th:dfs}
Given a multilayer graph $G = (V,E,L)$, let $[\ell_1, \ldots, \ell_{|L|}]$ be an order defined over set $L$.
Let $\mathbf{Q}_0 = \{[0]_{|L|}\}$, and, $\forall i \in [1..|L|]$,
let $\mathbf{Q}_i = \{\vec{k}' \mid C_{\vec{k}'} \in \kcorespathalg(G, C_{\vec{k}}, \vec{k}, \ell), \vec{k} \in \mathbf{Q}_{i-1}, \ell \in (\ell_i..\ell_{|L|}], k_{\ell} = 0\}$
and $\mathbf{C}_i = \{\vec{k}' \mid C_{\vec{k}'} \in \kcorespathalg(G, C_{\vec{k}}, \vec{k}, \ell), \vec{k} \in \mathbf{Q}_{i-1}, \ell \in [\ell_1..\ell_i], k_{\ell} = 0\}$.
The set $\mathbf{C} = \{C_{\vec{k}} \mid \vec{k} \in \bigcup_{i = 0}^{|L|} \mathbf{Q}_i \cup \bigcup_{i = 1}^{|L|} \mathbf{C}_i\}$ is the multilayer core decomposition of $G$.
\end{mytheorem}
\begin{proof}
The multilayer core decomposition of $G$ is formed by the union of all non-empty and distinct cores of all paths $\mathcal{P}$ of the lattice connecting a non-leaf node to a leaf by varying the same $\ell$-th component of the corresponding \corenessvec{s}.
%
%

\revision{
For any $i \in [1..|L|]$, let $\mathcal{P}_i \in \mathcal{P}$ denote the subset of paths whose \corenessvec{s} $\vec{k}'=[k'_{\ell}]_{\ell \in L}$ have a number of non-zero components equal to $i$.
By definition of $\mathbf{Q}_i$ and $\mathbf{C}_i$ it holds that all \corenessvec{s} $\vec{k}'$ of the cores along the paths in $\mathcal{P}_i$ are in $\mathbf{Q}_i \cup \mathbf{C}_i = \{\vec{k}' : |\{k'_{\ell} : \ell \in L, k'_{\ell} > 0\}| = i\}$.
Also, since some of the paths may overlap, all cores along the paths $\mathcal{P}_i$ are computed by executing single-layer core decompositions initiated at a subset of cores along the paths $\mathcal{P}_{i-1}$.
Such a subset of cores is represented by the subset of \corenessvec{s} within $\mathbf{Q}_{i-1} = \{\vec{k} :  |\{k_{\ell} : \ell \in [\ell_{2}..\ell_{|L|}], k_{\ell} > 0\}| = i-1\}$.
Moreover, note that single-layer core decompositions for the layers where $k_{\ell} \neq 0$ are discarded, as it boils down to visit cores in $\mathcal{P}_{i-1}$.
As a result, the set  $\{C_{\vec{k}} \mid \vec{k} \in \bigcup_{i = 0}^{|L|} \mathbf{Q}_i \cup \bigcup_{i = 1}^{|L|} \mathbf{C}_i\}$ correctly contains all possible \corenessvec{s} of the core lattice.
}
\end{proof}

\begin{algorithm}[t]
\caption{\revision{\kcorespathalg}} \label{alg:corespath}
\begin{algorithmic}[1]
\REQUIRE \revision{A multilayer graph $G = (V,E,L)$, a set $S \subseteq V$ of vertices, an $|L|$-dimensional integer vector $\vec{k} = [k_{\ell}]_{\ell \in L}$, and a layer $\ell \in L$.}
\ENSURE \revision{The set $\mathbf{C}_{\vec{k},\ell}$ of the multilayer cores of $G$ varying the $\ell$-th component of $\vec{k}$.}
\STATE \revision{$\mathbf{C}_{\vec{k},\ell} \leftarrow \emptyset$, $\mathcal{D} \leftarrow \emptyset$}
\FORALL{\revision{$u \in S$}}
	\STATE \revision{$\mathcal{D}(\deg_S(u, \ell)) \leftarrow \mathcal{D}(\deg_S(u, \ell)) \cup \{u\}$}
\ENDFOR
\FORALL{\revision{$k \in [0,\ldots,K_{\ell}]$}}
	\WHILE{\revision{$D(k) \neq \emptyset$}}
		\STATE \revision{remove a vertex $u$ from $\mathcal{D}(k)$}
		\STATE \revision{$S \leftarrow S \setminus \{u\}$}
		\FORALL{\revision{$v \in S: (u,v,\ell) \in E \land \deg_S(v, \ell) \geq k$}}
			\STATE \revision{$\mathcal{D}(\deg_S(v, \ell) + 1) \leftarrow \mathcal{D}(\deg_S(v, \ell) + 1) \setminus \{v\}$}
			\STATE \revision{$\mathcal{D}(\deg_S(v, \ell)) \leftarrow \mathcal{D}(\deg_S(v, \ell)) \cup \{v\}$}
		\ENDFOR
		\FORALL{\revision{$\hat{\ell} \in L \setminus \{\ell\}$}} \label{line:corespath:otherlayers}
			\FORALL{\revision{$v \in S: (u,v,\hat{\ell}) \in E \land \deg_S(v, \hat{\ell}) < k_{\hat{\ell}}$}}
				\STATE \revision{$\mathcal{D}(\deg_S(v, \ell)) \leftarrow \mathcal{D}(\deg_S(v, \ell)) \setminus \{v\}$}
				\STATE \revision{$\mathcal{D}(k) \leftarrow \mathcal{D}(k) \setminus \{v\}$}
			\ENDFOR
		\ENDFOR
	\ENDWHILE
	\STATE \revision{$\mathbf{C}_{\vec{k},\ell} \leftarrow \mathbf{C}_{\vec{k},\ell} \cup \{S\}$}
\ENDFOR
\end{algorithmic}
\end{algorithm}

Referring to  Algorithm~\ref{alg:dfs}, the result in Theorem~\ref{th:dfs} is implemented by keeping track of a subset of layers $R \subseteq L$.
At the beginning $R = L$, and, at each iteration of the main cycle, a layer $\ell$ is removed from it.
\revision{
The algorithm is guaranteed to be sound and complete regardless of the layer ordering, which may instead affect running time.
In our experiments we test several orders, such as random, or non-decreasing/non-increasing average-degree density.
All those orders resulted in comparable running times, we therefore decided to stick to the the simplest one, i.e., random.
}
Set $\mathbf{Q}$ keeps track of (the \corenessvec\ of) all lattice nodes where the current single-layer core-decomposition processes need to be run from.
$\mathbf{Q}'$ stores the (\corenessvec\ of) cores computed from each node in $\mathbf{Q}$ and for each layer within $R$, while also forming the basis of $\mathbf{Q}$ for the next iteration.

In summary, compared to \bfs, the \textsc{dfs} method reduces both the time complexity of computing all cores in a path $\mathcal{P}$ from a non-leaf node to a leaf of the core lattice (from $\mathcal{O}(|\mathcal{P}| \times (|E| + |V|\times|L|))$ to $\mathcal{O}(|\mathcal{P}| + |E| + |V|\times|L|)$), and the number of times a core is \emph{visited}, which may now be smaller than the number of its fathers.
On the other hand, \dfs\ comes with the aforementioned issue that some cores may be \emph{computed} multiple times (while in \bfs\ every core is computed only once).
Furthermore, cores are computed starting from larger subgraphs, as intersection among multiple fathers can not exploited.

\begin{algorithm}[t]
	\caption{\hybrid} \label{alg:hybrid}
	\begin{algorithmic}[1]
		\REQUIRE A multilayer graph $G = (V,E,L)$.
		\ENSURE The set \coresset\ of all non-empty multilayer cores of $G$.
		
		\STATE $\mathbf{Q} \!\leftarrow\! \{[0]_{|L|}\}$, \ $\mathcal{F}([0]_{|L|}) \!\leftarrow\! \emptyset$ \COMMENT{$\mathcal{F}$ keeps track of father nodes}
		\STATE $\mathbf{Q}' \!\leftarrow\! \bigcup_{\ell \in L} \{\vec{k} \mid C_{\vec{k}} \in \kcorespathalg(G, \!V, \![0]_{|L|},\!\ell)\}$ \COMMENT{looked-ahead cores}
		\STATE $\coresset \leftarrow \{C_{\vec{k}} \mid \vec{k} \in \mathbf{Q}'\}$
		
        \WHILE {$\mathbf{Q} \neq \emptyset$}
            \STATE dequeue $\vec{k} = [k_{\ell}]_{\ell \in L}$ from $\mathbf{Q}$ 

            \IF [Corollary~\ref{cor:fathernumber}]{$|\{k_{\ell} : k_{\ell} > 0\}| = |\mathcal{F}(\vec{k})| \wedge \vec{k} \notin \mathbf{Q}'$}
                \STATE $F_{\cap} \leftarrow \bigcap_{F \in \mathcal{F}(\vec{k})} F$ \COMMENT{Corollary~\ref{cor:coreintersection}}
        		\STATE $C_{\vec{k}} \leftarrow \kcorealg(G,F_{\cap},\vec{k})$ \COMMENT{Algorithm~\ref{alg:core}} 
        		
                \IF{$C_{\vec{k}} \neq \emptyset$} \label{line:hybrid:exists}
            		\STATE $\coresset \! \leftarrow \!\coresset  \cup  \{ C_{\vec{k}} \}$
            		\STATE $\vec{d}_{\mu}(C_{\vec{k}}) \leftarrow [\mu({C_{\vec{k}}},\ell)]_{\ell \in L}$ \COMMENT{look-ahead mechanism (Corollary~\ref{cor:mindegvector})}
            		\STATE $\mathbf{Q}' \leftarrow \mathbf{Q}' \cup \{\vec{k}' \mid \vec{k} \leq \vec{k}' \leq \vec{d}_{\mu}(C_{\vec{k}})\}$
                \ENDIF
            \ENDIF

            \IF{$\vec{k} \in \mathbf{Q}'$}
                \FORALL[enqueue child nodes]{$\ell \in L$}
            		\STATE $\vec{k}' \leftarrow [k_1, \ldots, k_{\ell} + 1, \ldots, k_{|L|}]$
            		\STATE enqueue $\vec{k}$' into $\mathbf{Q}$
            		\STATE $\mathcal{F}(\vec{k}') \leftarrow \mathcal{F}(\vec{k}') \cup \{C_{\vec{k}}\}$
                \ENDFOR
            \ENDIF
        \ENDWHILE
	\end{algorithmic}
\end{algorithm}

\subsection{Hybrid algorithm}
The output of both  \bfs\ and \dfs\ correctly corresponds to all distinct cores of the input  graph and the corresponding maximal \corenessvec{s}.\footnote{Pseudocodes in Algorithms~\ref{alg:bfs} and \ref{alg:dfs} guarantee this as cores are added to a set \coresset\ that does not allow duplicates. Any real implementation can easily take care of this by checking whether a core is already in \coresset, and update it in case the corresponding \corenessvec\ contains the previously-stored one.}
Nevertheless, none of these methods is able to skip the computation of non-distinct cores.
Indeed, both methods need to compute every core $C$ as many times as the number of its \corenessvec{s}  to guarantee completeness.
To address this limitation, we devise a further method where the main peculiarities of  \bfs\ and \dfs\ are joined into a ``hybrid''  lattice-visiting strategy. This \hybrid\ method exploits the following corollary of Theorem~\ref{th:maximalcorenessvecuniqueness}, stating that the maximal \corenessvec\ of a core $C$ is given by the vector containing the minimum degree of a vertex in $C$ for each layer:

\begin{mycorollary}\label{cor:mindegvector}
	Given a multilayer graph $G = (V,E,L)$,
	the maximal \corenessvec\ of a multilayer core $C$ of $G$ corresponds to the $|L|$-dimensional integer vector $\vec{d}_{\mu}(C) = [\mu(C,\ell)]_{\ell \in L}$.
\end{mycorollary}
\begin{proof}
By Definition~\ref{def:mlkcores}, vector $\vec{d}_{\mu}(C)$ is a \corenessvec\ of $C$.
Assume that $\vec{d}_{\mu}(C)$ is not maximal, meaning that another \corenessvec\ $\vec{k} = [k_{\ell}]_{\ell \in L}$ dominating $\vec{d}_{\mu}(C)$ exists.
This implies that $k_{\ell} \geq \mu(C,\ell)$, and $\exists \hat{\ell} \in L : k_{\hat{\ell}} > \mu(C,\hat{\ell})$.
By definition of multilayer core, all vertices in $C$ have degree larger than the minimum degree $\mu(C,\hat{\ell})$ in layer $\hat{\ell}$, which is a clear contradiction.
\end{proof}

Corollary~\ref{cor:mindegvector} gives a rule to skip the computation of non-distinct cores: given a core $C$ with \corenessvec\ $\vec{k} = [k_{\ell}]_{\ell \in L}$, all cores with \corenessvec\ $\vec{k'} = [k'_{\ell}]_{\ell \in L}$ such that $\forall \ell \in L : k_{\ell} \leq k'_{\ell} \leq \mu(C,\ell)$ are guaranteed to be equal to $C$ and do not need to be explicitly computed.
For instance, in Figure~\ref{fig:dagstr}, assume that the min-degree vector of the $(0,0,1)$-core is $(0,1,2)$. Then, cores $(0,0,2)$, $(0,1,1)$, and $(0,1,2)$ can immediately be set equal to the $(0,0,1)$-core.
The \hybrid\ algorithm we present here (Algorithm~\ref{alg:hybrid}) exploits this rule by performing a breadth-first search  equipped with a ``look-ahead'' mechanism resembling  a depth-first search.
Moreover, \hybrid\ starts with a single-layer core decomposition for each layer so as to have more fathers early-on for intersections.
Cores interested by the look-ahead rule are still \emph{visited} and stored in $\mathbf{Q}'$, as they may be needed for future core computations.
However, no further computational overhead is required for them.
\revision{
The time complexity of \hybrid\ is the same as \bfs, plus an additional $\mathcal{O}(|E|)$ time for every visited multilayer core, which is needed in the look-ahead rule to compute the min-degree vector of that core.
}

\subsection{Discussion}

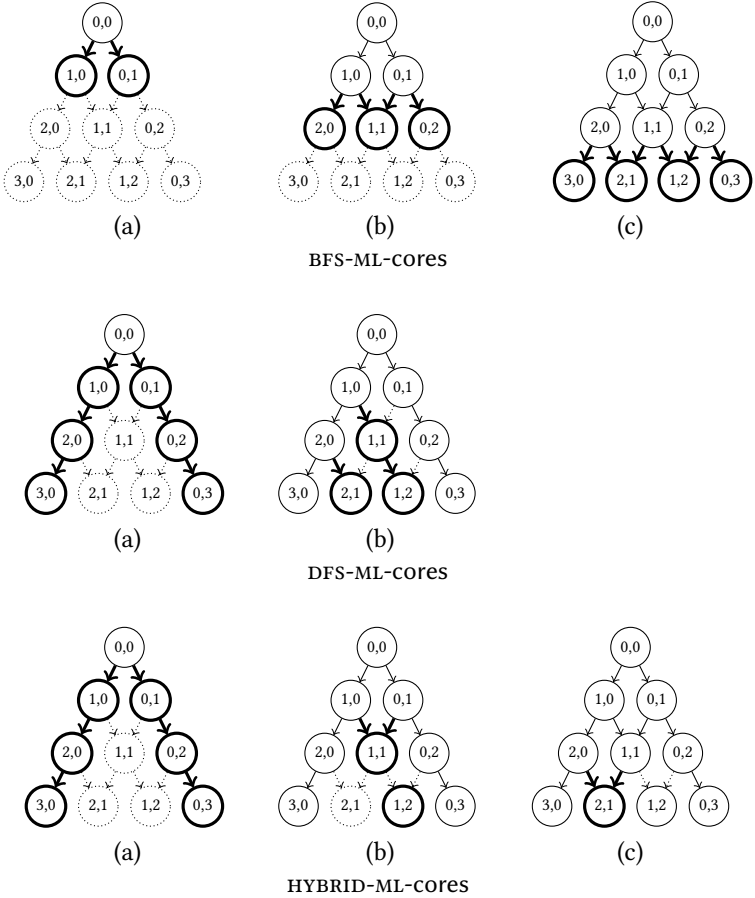
\begin{figure}[t]
\begin{tabular}{ccc}
\begin{tikzpicture} [scale=0.087, every node/.style={circle, draw, scale=0.65}]
	\node [densely dotted] (30) at (0,0) {3,0};
	\node [densely dotted] (21) at (8,0) {2,1};
	\node [densely dotted] (12) at (16,0) {1,2};
	\node [densely dotted] (03) at (24,0) {0,3};
	
	\node [densely dotted] (20) at (4,8) {2,0};
	\node [densely dotted] (11) at (12,8) {1,1};
	\node [densely dotted] (02) at (20,8) {0,2};
	
	\node [very thick] (10) at (8,16) {1,0};
	\node [very thick] (01) at (16,16) {0,1};
	
	\node (00) at (12,24) {0,0};

	\draw [densely dotted, ->] (20) edge (30);
	\draw [densely dotted, ->] (20) edge (21);
	\draw [densely dotted, ->] (11) edge (21);
	\draw [densely dotted, ->] (11) edge (12);
	\draw [densely dotted, ->] (02) edge (12);
	\draw [densely dotted, ->] (02) edge (03);

	\draw [densely dotted, ->] (10) edge (20);
	\draw [densely dotted, ->] (10) edge (11);
	\draw [densely dotted, ->] (01) edge (02);
	\draw [densely dotted, ->] (01) edge (11);

	\draw [very thick, ->] (00) edge (10);
	\draw [very thick, ->] (00) edge (01);
\end{tikzpicture} \hspace{0.5cm} & \begin{tikzpicture} [scale=0.087, every node/.style={circle, draw, scale=0.65}]
	\node [densely dotted] (30) at (0,0) {3,0};
	\node [densely dotted] (21) at (8,0) {2,1};
	\node [densely dotted] (12) at (16,0) {1,2};
	\node [densely dotted] (03) at (24,0) {0,3};
	
	\node [very thick] (20) at (4,8) {2,0};
	\node [very thick] (11) at (12,8) {1,1};
	\node [very thick] (02) at (20,8) {0,2};
	
	\node (10) at (8,16) {1,0};
	\node (01) at (16,16) {0,1};
	
	\node (00) at (12,24) {0,0};

	\draw [densely dotted, ->] (20) edge (30);
	\draw [densely dotted, ->] (20) edge (21);
	\draw [densely dotted, ->] (11) edge (21);
	\draw [densely dotted, ->] (11) edge (12);
	\draw [densely dotted, ->] (02) edge (12);
	\draw [densely dotted, ->] (02) edge (03);

	\draw [very thick, ->] (10) edge (20);
	\draw [very thick, ->] (10) edge (11);
	\draw [very thick, ->] (01) edge (02);
	\draw [very thick, ->] (01) edge (11);

	\draw [->] (00) edge (10);
	\draw [->] (00) edge (01);
\end{tikzpicture} & \hspace{0.5cm} \begin{tikzpicture} [scale=0.087, every node/.style={circle, draw, scale=0.65}]
	\node [very thick] (30) at (0,0) {3,0};
	\node [very thick] (21) at (8,0) {2,1};
	\node [very thick] (12) at (16,0) {1,2};
	\node [very thick] (03) at (24,0) {0,3};
	
	\node (20) at (4,8) {2,0};
	\node (11) at (12,8) {1,1};
	\node (02) at (20,8) {0,2};
	
	\node (10) at (8,16) {1,0};
	\node (01) at (16,16) {0,1};
	
	\node (00) at (12,24) {0,0};

	\draw [very thick, ->] (20) edge (30);
	\draw [very thick, ->] (20) edge (21);
	\draw [very thick, ->] (11) edge (21);
	\draw [very thick, ->] (11) edge (12);
	\draw [very thick, ->] (02) edge (12);
	\draw [very thick, ->] (02) edge (03);

	\draw [->] (10) edge (20);
	\draw [->] (10) edge (11);
	\draw [->] (01) edge (02);
	\draw [->] (01) edge (11);

	\draw [->] (00) edge (10);
	\draw [->] (00) edge (01);
\end{tikzpicture}\\
(a) & (b) & (c)\\
\multicolumn{3}{c}{\bfs}\vspace{0.5cm}\\
\begin{tikzpicture} [scale=0.087, every node/.style={circle, draw, scale=0.65}]
	\node [very thick] (30) at (0,0) {3,0};
	\node [densely dotted] (21) at (8,0) {2,1};
	\node [densely dotted] (12) at (16,0) {1,2};
	\node [very thick] (03) at (24,0) {0,3};
	
	\node [very thick] (20) at (4,8) {2,0};
	\node [densely dotted] (11) at (12,8) {1,1};
	\node [very thick] (02) at (20,8) {0,2};
	
	\node [very thick] (10) at (8,16) {1,0};
	\node [very thick] (01) at (16,16) {0,1};
	
	\node (00) at (12,24) {0,0};

	\draw [very thick, ->] (20) edge (30);
	\draw [densely dotted, ->] (20) edge (21);
	\draw [densely dotted, ->] (11) edge (21);
	\draw [densely dotted, ->] (11) edge (12);
	\draw [densely dotted, ->] (02) edge (12);
	\draw [very thick, ->] (02) edge (03);

	\draw [very thick, ->] (10) edge (20);
	\draw [densely dotted, ->] (10) edge (11);
	\draw [very thick, ->] (01) edge (02);
	\draw [densely dotted, ->] (01) edge (11);

	\draw [very thick, ->] (00) edge (10);
	\draw [very thick, ->] (00) edge (01);
\end{tikzpicture} & \begin{tikzpicture} [scale=0.087, every node/.style={circle, draw, scale=0.65}]
	\node (30) at (0,0) {3,0};
	\node [very thick] (21) at (8,0) {2,1};
	\node [very thick] (12) at (16,0) {1,2};
	\node (03) at (24,0) {0,3};
	
	\node (20) at (4,8) {2,0};
	\node [very thick] (11) at (12,8) {1,1};
	\node (02) at (20,8) {0,2};
	
	\node (10) at (8,16) {1,0};
	\node (01) at (16,16) {0,1};
	
	\node (00) at (12,24) {0,0};

	\draw [ ->] (20) edge (30);
	\draw [very thick, ->] (20) edge (21);
	\draw [densely dotted, ->] (11) edge (21);
	\draw [very thick, ->] (11) edge (12);
	\draw [densely dotted, ->] (02) edge (12);
	\draw [->] (02) edge (03);

	\draw [ ->] (10) edge (20);
	\draw [very thick, ->] (10) edge (11);
	\draw [->] (01) edge (02);
	\draw [densely dotted, ->] (01) edge (11);

	\draw [->] (00) edge (10);
	\draw [->] (00) edge (01);
\end{tikzpicture} & \\
(a) & (b) & \\
\multicolumn{3}{c}{\dfs}\vspace{0.5cm}\\
\begin{tikzpicture} [scale=0.087, every node/.style={circle, draw, scale=0.65}]
	\node [very thick] (30) at (0,0) {3,0};
	\node [densely dotted] (21) at (8,0) {2,1};
	\node [densely dotted] (12) at (16,0) {1,2};
	\node [very thick] (03) at (24,0) {0,3};
	
	\node [very thick] (20) at (4,8) {2,0};
	\node [densely dotted] (11) at (12,8) {1,1};
	\node [very thick] (02) at (20,8) {0,2};
	
	\node [very thick] (10) at (8,16) {1,0};
	\node [very thick] (01) at (16,16) {0,1};
	
	\node (00) at (12,24) {0,0};

	\draw [very thick, ->] (20) edge (30);
	\draw [densely dotted, ->] (20) edge (21);
	\draw [densely dotted, ->] (11) edge (21);
	\draw [densely dotted, ->] (11) edge (12);
	\draw [densely dotted, ->] (02) edge (12);
	\draw [very thick, ->] (02) edge (03);

	\draw [very thick, ->] (10) edge (20);
	\draw [densely dotted, ->] (10) edge (11);
	\draw [very thick, ->] (01) edge (02);
	\draw [densely dotted, ->] (01) edge (11);

	\draw [very thick, ->] (00) edge (10);
	\draw [very thick, ->] (00) edge (01);
\end{tikzpicture} & \begin{tikzpicture} [scale=0.087, every node/.style={circle, draw, scale=0.65}]
	\node (30) at (0,0) {3,0};
	\node [densely dotted] (21) at (8,0) {2,1};
	\node [very thick] (12) at (16,0) {1,2};
	\node (03) at (24,0) {0,3};
	
	\node (20) at (4,8) {2,0};
	\node [very thick] (11) at (12,8) {1,1};
	\node (02) at (20,8) {0,2};
	
	\node (10) at (8,16) {1,0};
	\node (01) at (16,16) {0,1};
	
	\node (00) at (12,24) {0,0};

	\draw [ ->] (20) edge (30);
	\draw [densely dotted, ->] (20) edge (21);
	\draw [densely dotted, ->] (11) edge (21);
	\draw [densely dotted, ->] (11) edge (12);
	\draw [densely dotted, ->] (02) edge (12);
	\draw [->] (02) edge (03);

	\draw [ ->] (10) edge (20);
	\draw [very thick, ->] (10) edge (11);
	\draw [->] (01) edge (02);
	\draw [very thick, ->] (01) edge (11);

	\draw [->] (00) edge (10);
	\draw [->] (00) edge (01);
\end{tikzpicture} & \begin{tikzpicture} [scale=0.087, every node/.style={circle, draw, scale=0.65}]
	\node (30) at (0,0) {3,0};
	\node [very thick] (21) at (8,0) {2,1};
	\node (12) at (16,0) {1,2};
	\node (03) at (24,0) {0,3};
	
	\node (20) at (4,8) {2,0};
	\node (11) at (12,8) {1,1};
	\node (02) at (20,8) {0,2};
	
	\node (10) at (8,16) {1,0};
	\node (01) at (16,16) {0,1};
	
	\node (00) at (12,24) {0,0};

	\draw [ ->] (20) edge (30);
	\draw [very thick, ->] (20) edge (21);
	\draw [very thick, ->] (11) edge (21);
	\draw [densely dotted, ->] (11) edge (12);
	\draw [densely dotted, ->] (02) edge (12);
	\draw [->] (02) edge (03);

	\draw [ ->] (10) edge (20);
	\draw [->] (10) edge (11);
	\draw [->] (01) edge (02);
	\draw [->] (01) edge (11);

	\draw [->] (00) edge (10);
	\draw [->] (00) edge (01);
\end{tikzpicture}\\
(a) & (b) & (c)\\
\multicolumn{3}{c}{\hybrid}\\
\end{tabular}
\caption{\label{fig:running}\revision{Running example of our algorithms for multilayer core decomposition over a core lattice of a 2-layer graph.
Nodes and links depicted by solid lines have been visited in previous steps of the algorithm, those in thick lines are visited during the current step, while the remaining in dotted lines have not been visited yet.}}
\end{figure}

\revision{
To have a concrete comparison of the characteristics of the proposed algorithms for multilayer core decomposition, we report in Figure~\ref{fig:running} a running example over a core lattice of a simple 2-layer graph.
All the algorithms start by visiting the root of the core lattice, which corresponds to the whole input multilayer graph (this preliminary step is left out from Figure~\ref{fig:running} since it is shared by all the methods).
\bfs\ visits the core lattice level by level, and exploits every containment relationship.
The execution pattern of \dfs\ is instead much different: it starts by finding those multilayer cores having a single component of the \corenessvec\ other than zero and, in a later step, visits the rest of the core lattice.
In both steps (a) and (b) \dfs\ visits cores following straight paths in the search space, i.e., from a core to a leaf.
As a result, not all the containment relationships are exploited.
For instance, the computation of the $(2,1)$-core exploits the containment from the $(2,0)$-core, but not from the $(1,1)$-core.
\hybrid\ is, as expected, a mix of the two other methods.
The first step is identical to \dfs.
At step (b), \hybrid\ starts to visit the remaining cores by a breadth-first-search strategy, while also exploiting the look-ahead mechanism.
In particular, the minimum degree vector of the $(1,1)$-core is found to be equal to $(1,2)$;
therefore, the $(1,2)$-core is not computed directly, but set equal to the $(1,1)$-core.
In the final step \hybrid\ visits the remaining core by going on with the breadth-first search.
}

We already discussed (in the respective paragraphs) the strengths and weaknesses of \bfs\ and \dfs: the best among the two is determined by the peculiarities of the specific input graph.
On the other hand, \hybrid\ profitably exploits the main nice features of both \bfs\ and \dfs, thus is expected to outperform both methods in most cases.
However, in those graphs where the number of non-distinct cores is limited, the overhead due to the look-ahead mechanism can make  the performance of \hybrid\ degrade.

In terms of space requirements, \bfs\ needs to keep in memory all those cores having at least a child in the queue, i.e., at most two levels of the lattice (the space taken by a multilayer core is $\mathcal{O}(|V|)$).
The same applies to \hybrid\, with the addition of the cores computed through single-layer core decomposition and look-ahead, until all their children have been processed.
\dfs\ instead requires to store all cores where the single-layer core-decomposition process should be started from, both in the current iteration and the next one.
Thus, we expect \dfs\ to take more space than \bfs\ and \hybrid,
as in practice the number of cores to be stored should be more than the cores belonging to two lattice levels.


\subsection{Experimental results}
\label{sec:experiments}
In this subsection we present experiments to $(i)$ compare the proposed algorithms in terms of runtime, memory consumption, and search-space exploration; $(ii)$ characterize the output core decompositions, also by comparing total number of cores and number of inner-most cores.

\spara{Datasets.}
We select publicly-available real-world multilayer networks, whose main characteristics are summarized in Table~\ref{tab:datasets}.
\textsf{Homo}\footnote{\url{http://deim.urv.cat/~manlio.dedomenico/data.php}\label{foot:deim}} and \textsf{SacchCere}$^{\ref{foot:deim}}$ are networks describing different types of genetic interactions between genes in Homo Sapiens and Saccharomyces Cerevisiae, respectively.
\textsf{ObamaInIsrael}$^{\ref{foot:deim}}$ represents different types of social interaction (e.g., \emph{re-tweeting}, \emph{mentioning}, and \emph{replying}) among Twitter users, focusing on Barack Obama's visit to Israel in 2013. Similarly, \textsf{Higgs}$^{\ref{foot:deim}}$ is built by tracking the spread of news about the discovery of the Higgs boson on Twitter, with the additional layer for the \emph{following} relation.
\textsf{Friendfeed}\footnote{\url{http://multilayer.it.uu.se/datasets.html}\label{foot:uu}} contains public interactions among users of Friendfeed collected over two months (e.g., \emph{commenting}, \emph{liking}, and \emph{following}).
\textsf{FriendfeedTwitter}$^{\ref{foot:uu}}$ is a multi-platform social network, where layers represent interactions within Friendfeed and Twitter between users registered to both platforms~\cite{DickisonMagnaniRossi2016}.
\textsf{Amazon}\footnote{\url{https://snap.stanford.edu/data/}} is a co-purchasing \emph{temporal network}, containing four snapshots between March and June 2003.
Finally, \textsf{DBLP}\footnote{\url{http://dblp.uni-trier.de/xml/}} is derived following the methodology in~\cite{bonchi2015chromatic}.
For each co-authorship relation (edge), the bag of words resulting from the titles of all papers co-authored by the two authors is collected.
Then \emph{LDA} topic modeling~\citep{blei2003latent} is applied to automatically identify a hundred topics.
Among these, ten topics that are recognized as the most relevant to the data-mining area have been hand-picked.
Every selected topic corresponds to a layer. An edge between two co-authors in a certain layer exists if the relation between those co-authors is labeled with the topic corresponding to that layer.

\begin{table}[t!]
\centering
\caption{\revision{Characteristics of the real-world datasets: number of vertices ($|V|$), number of overall edges ($|E|$), number of layers ($|L|$),
minimum, average, and maximum number of edges in a layer (min $|E_\ell|$, avg $|E_\ell|$, max $|E_\ell|$), and application domain.}}
\label{tab:datasets}
\revision{
\begin{tabular}{c|ccccccccc}
\multicolumn{1}{c}{dataset} & $|V|$ & $|E|$ & $|L|$ & min $|E_\ell|$ & avg $|E_\ell|$ & max $|E_\ell|$ & domain\\
\hline
\textsf{Homo} & $18$k & $153$k & $7$ & $256$ & $21$k & $83$k & genetic\\
\textsf{SacchCere} & $6.5$k & $247$k & $7$ & $1.3$k & $35$k & $91$k & genetic\\
\textsf{DBLP} & $513$k & $1.0$M & $10$ &  $96$k & $101$k & $113$k & co-authorship\\
\textsf{ObamaInIsrael} & $2.2$M & $3.8$M & $3$ & $557$k & $1.2$M & $1.8$M & social\\
\textsf{Amazon} & $410$k & $8.1$M & $4$ & $899$k & $2.0$M & $2.4$M & co-purchasing\\
\textsf{FriendfeedTwitter} & $155$k & $13$M & $2$ & $5.2$M & $6.8$M & $8.3$M & social\\
\textsf{Higgs} & $456$k & $13$M & $4$ & $28$k & $3.4$M & $12$M & social\\
\textsf{Friendfeed} & $510$k & $18$M & $3$ & $226$k & $6.2$M & $18$M & social\\
\hline
\end{tabular}
}
\end{table}

\spara{Implementation.}
All methods are implemented in Python (v. 2.7.12) and compiled by Cython: all our code is available at \href{https://github.com/egalimberti/multilayer_core_decomposition}{github.com/egalimberti/multilayer\_core\_decomposition}.
All experiments are run on a machine equipped with Intel Xeon CPU at 2.1GHz and 128GB RAM except for Figure~\ref{fig:dblplayers}, whose results are obtained on Intel Xeon CPU at 2.7GHz with 128GB RAM.

\spara{Comparative evaluation.}
We compare the na\"{\i}ve baseline (for short \textsc{n}) and the three proposed methods \bfs\ (for short \textsc{bfs}), \dfs\ (\textsc{dfs}), \hybrid\ (\textsc{h})
in terms of running time, memory usage, and number of computed cores (as a measure of the explored search-space portion).
The results of this comparison are shown in Table~\ref{tab:results}.
As expected, \textsc{n} is the least efficient method: it is outperformed by our algorithms by 1--4 orders of magnitude.
Due to its excessive requirements, we could not run it in reasonable time (i.e., 30 days) on the \textsf{Friendfeed} dataset.

\revision{
Among the proposed methods, \textsc{h} is recognized as the best method (in absolute or with performance comparable to the best one) in the first five (out of a total of eight) datasets.
In the remaining three datasets the best method is \textsc{dfs}.
This is mainly motivated by the fact that those three datasets have a relatively small number of layers, an aspect which  \textsc{dfs} takes particular advantage from (as also better testified by the experiment with varying the number of layers discussed below).
In some cases \textsc{h} is also comparable to  \textsc{bfs}, thus confirming the fact that in datasets where the number of non-distinct cores is not so large the performance of the two methods gets closer.
A similar reasoning holds between \textsc{bfs} and \textsc{dfs} (at least with a small/moderate number of the layers, see next): \textsc{bfs} is faster in most cases, but, due to the respective pros and cons discussed in Section~\ref{sec:algorithms}, it is not surprising that the two methods achieve comparable performance in a number of other cases.
}

\begin{table}[t!]
\centering
\caption{Comparative evaluation: proposed methods and baseline. Runtime differs from~\cite{galimberti2017core} since a different server was employed. For each dataset, best performances are bolded.\label{tab:results}}

\begin{tabular}{c|c|c|ccc}
\multicolumn{1}{c}{dataset} & \multicolumn{1}{c}{\#output cores} &  \multicolumn{1}{c}{method} &  \multicolumn{1}{c}{runtime (s)} &\multicolumn{1}{c}{memory (MB)}  &\multicolumn{1}{c}{\#computed cores} \\
\hline
\textsf{Homo} & $1\,845$ & \textsc{n} & $1\,145$ & $27$ & $12\,112$\\
 & & \textsc{bfs} & $13$ & $26$ & $3\,043$\\
 & & \textsc{dfs} & $27$ & $27$ & $6\,937$\\
 & & \textsc{h} & $\mathbf{12}$ & $\mathbf{25}$ & $\mathbf{2\,364}$\\
\hline
\textsf{SacchCere} & $74\,426$ & \textsc{n} & $24\,469$ & $55$ & $278\,402$\\
 & & \textsc{bfs} & $\mathbf{1\,134}$ & $\mathbf{34}$ & $89\,883$\\
 & & \textsc{dfs} & $2\,627$ & $57$ & $223\,643$\\
 & & \textsc{h} & $1\,146$ & $35$ & $\mathbf{83\,978}$\\
\hline
\textsf{DBLP} & $3\,346$ & \textsc{n} & $103\,231$ & $608$ & $34\,572$\\
 & & \textsc{bfs} & $68$ & $612$ & $6\,184$\\
 & & \textsc{dfs} & $282$ & $627$ & $38\,887$\\
 & & \textsc{h} & $\mathbf{29}$ & $\mathbf{521}$ & $\mathbf{5\,037}$\\
\hline
\textsf{Obama} & $2\,573$ &\textsc{n} & $37\,554$ & $1\,286$ & $3\,882$\\
\textsf{InIsrael} & & \textsc{bfs} & $226$ & $1\,299$ & $3\,313$\\
 & & \textsc{dfs} & $\mathbf{150}$ & $1\,384$ & $3\,596$\\
 & & \textsc{h} & $177$ & $\mathbf{1\,147}$ & $\mathbf{2\,716}$\\
\hline
\textsf{Amazon} & $1\,164$ & \textsc{n} & $11\,990$ & $\mathbf{425}$ & $1\,823$\\
 & & \textsc{bfs} & $3\,981$ & $534$ & $1\,354$\\
 & & \textsc{dfs} & $5\,278$ & $619$ & $2\,459$\\
 & & \textsc{h} & $\mathbf{3\,913}$ & $536$ & $\mathbf{1\,334}$\\
\hline
\textsf{Friendfeed} & $76\,194$ & \textsc{n} & $409\,489$ & $220$ & $80\,954$\\
\textsf{Twitter} & & \textsc{bfs} & $61\,113$ & $\mathbf{215}$ & $80\,664$\\
 & & \textsc{dfs} & $\mathbf{1\,973}$ & $267$ & $80\,745$\\
 & & \textsc{h} & $59\,520$ & $268$ & $\mathbf{76\,419}$\\
\hline
\textsf{Higgs} & $8\,077$ & \textsc{n} & $163\,398$ & $474$ & $22\,478$\\
 & & \textsc{bfs} & $2\,480$ & $\mathbf{465}$ & $12\,773$\\
 & & \textsc{dfs} & $\mathbf{640}$ & $490$ & $14\,119$\\
 & & \textsc{h} & $2\,169$ & $493$ & $\mathbf{9\,389}$\\
\hline
\textsf{Friendfeed} & $365\,666$ & \textsc{bfs} & $58\,278$ & $\mathbf{465}$ & $546\,631$\\
 & &\textsc{dfs} & $\mathbf{13\,356}$ & $591$ & $568\,107$\\
 & & \textsc{h} & $47\,179$ & $490$ & $\mathbf{389\,323}$\\
\hline
\end{tabular}

\end{table}

To test the behavior with varying the number of layers, Figure~\ref{fig:dblplayers} shows the running times of the proposed methods on different versions of the \textsf{DBLP} dataset, obtained by selecting a variable number of layers, from $2$ to $10$.
While the performance of the three methods is comparable up to six layers, beyond this threshold the execution time of \textsc{dfs} grows much faster than \textsc{bfs} and \textsc{h}.
This attests that the pruning rules of \textsc{bfs} and \textsc{h} are more effective as the layers increase.
To summarize, \textsc{dfs} is expected to have runtime comparable to (or better than) \textsc{bfs} and \textsc{h} when the number of layers is small, while \textsc{h} is faster than \textsc{bfs} when the number of non-distinct cores is large.

\begin{figure}[t!]
\centering
\includegraphics[width=0.6\columnwidth]{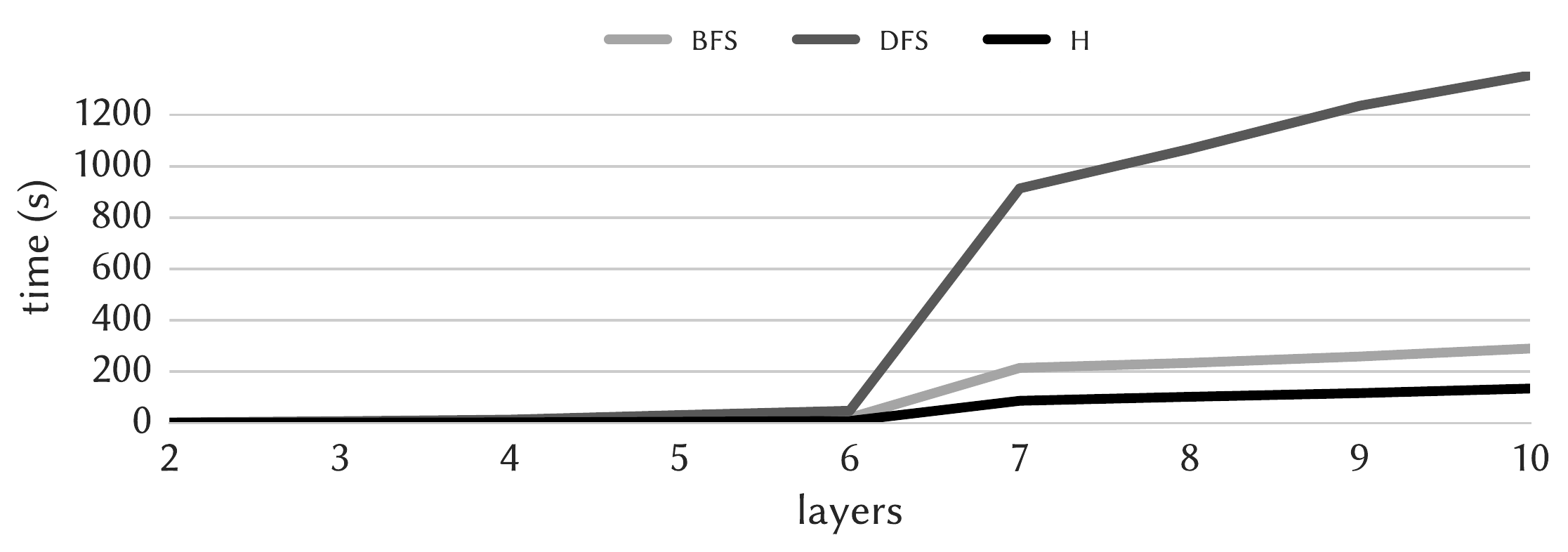}
\caption{\label{fig:dblplayers}   Runtime of the proposed methods with varying the number of layers (\textsf{DBLP} dataset).}
\end{figure}

The number of computed cores is always larger than the output cores as all methods might compute empty cores or, in the case of  \textsc{dfs}, the same core multiple times.
Table~\ref{tab:results} shows that \textsc{dfs} computes more cores than \textsc{bfs} and \textsc{h}, which conforms to its design principles.

Finally, all methods turn out to be memory-efficient, taking no more than $1.5$GB of memory.

\spara{Core-decomposition characterization.}
Figure~\ref{fig:levcores} reports the distribution of number of cores, core size, and average-degree density (i.e., number of edges divided by number of vertices) of the subgraph corresponding to a core.
Distributions are shown by level of the lattice\footnote{Recall that the lattice level has been defined in Section 3.1: level $i$ contains all cores whose coreness-vector components sum to $i$.} for the \textsf{SacchCere} and \textsf{Friendfeed} datasets.
Although the two datasets have very different scales, the distributions exhibit similar trends.
Being limited by the number of layers, the number of cores in the first levels of the lattice is very small, but then it exponentially grows until reaching its maximum within the first $25-30\%$ visited levels.
The average size of the cores is close to the number of vertices in the first lattice level, when cores' degree conditions are not very strict.
Then it decreases as the number of cores gets larger, with a maximum reached when very small cores stop ``propagating'' in the lower lattice levels.
Finally, the average (average-degree) density tends to increase for higher lattice level.
However, there are a couple of exceptions: it decreases ($i$) in the first few levels of \textsf{SacchCere}'s lattice, and ($ii$) in the last levels of both \textsf{SacchCere} and \textsf{Friendfeed}, where the core size
starts getting smaller, thus implying small average-degree values.

\begin{figure}[t!]
\begin{tabular}{ccc}
\includegraphics[width=0.3\columnwidth]{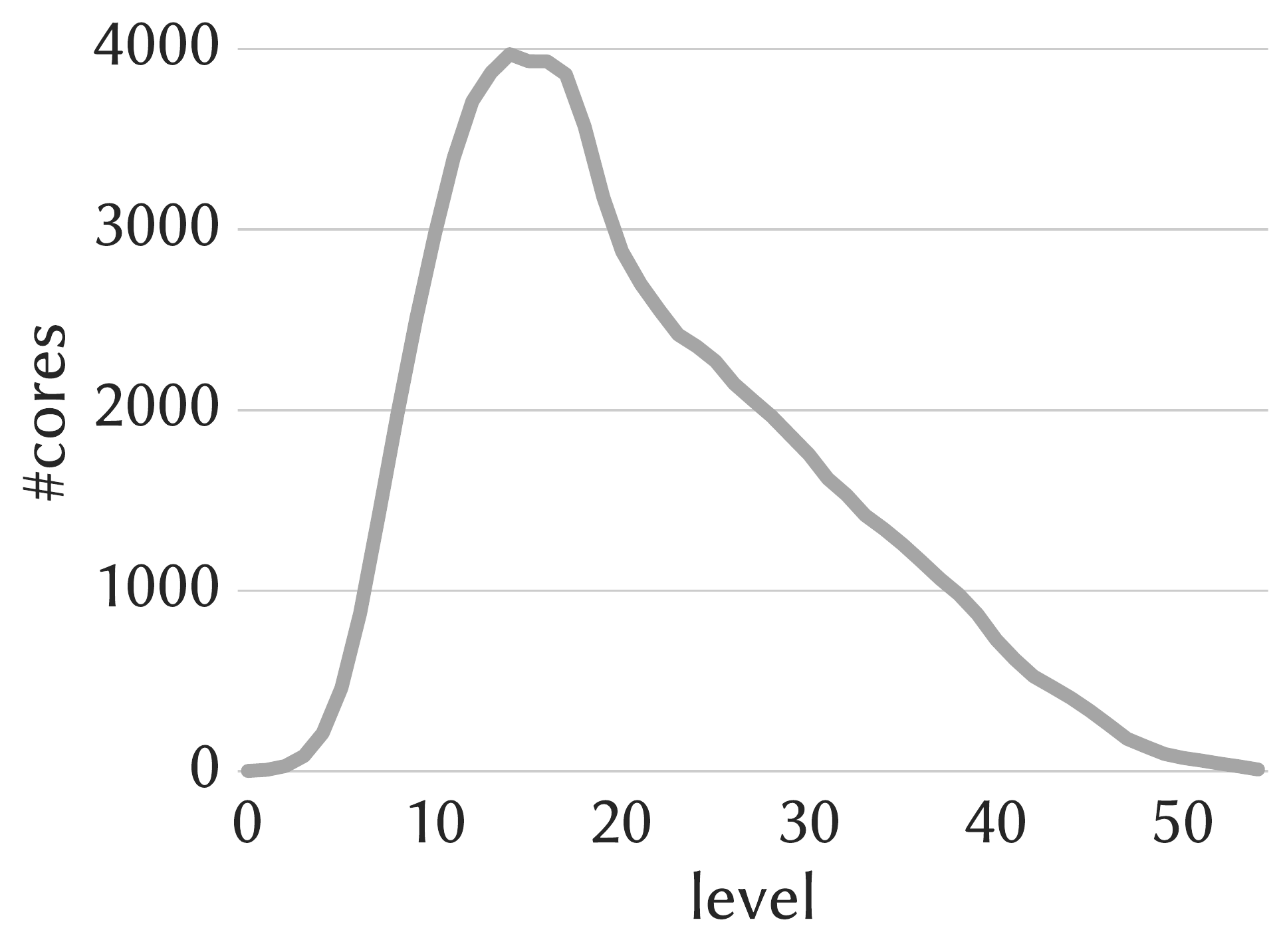} & \includegraphics[width=0.3\columnwidth]{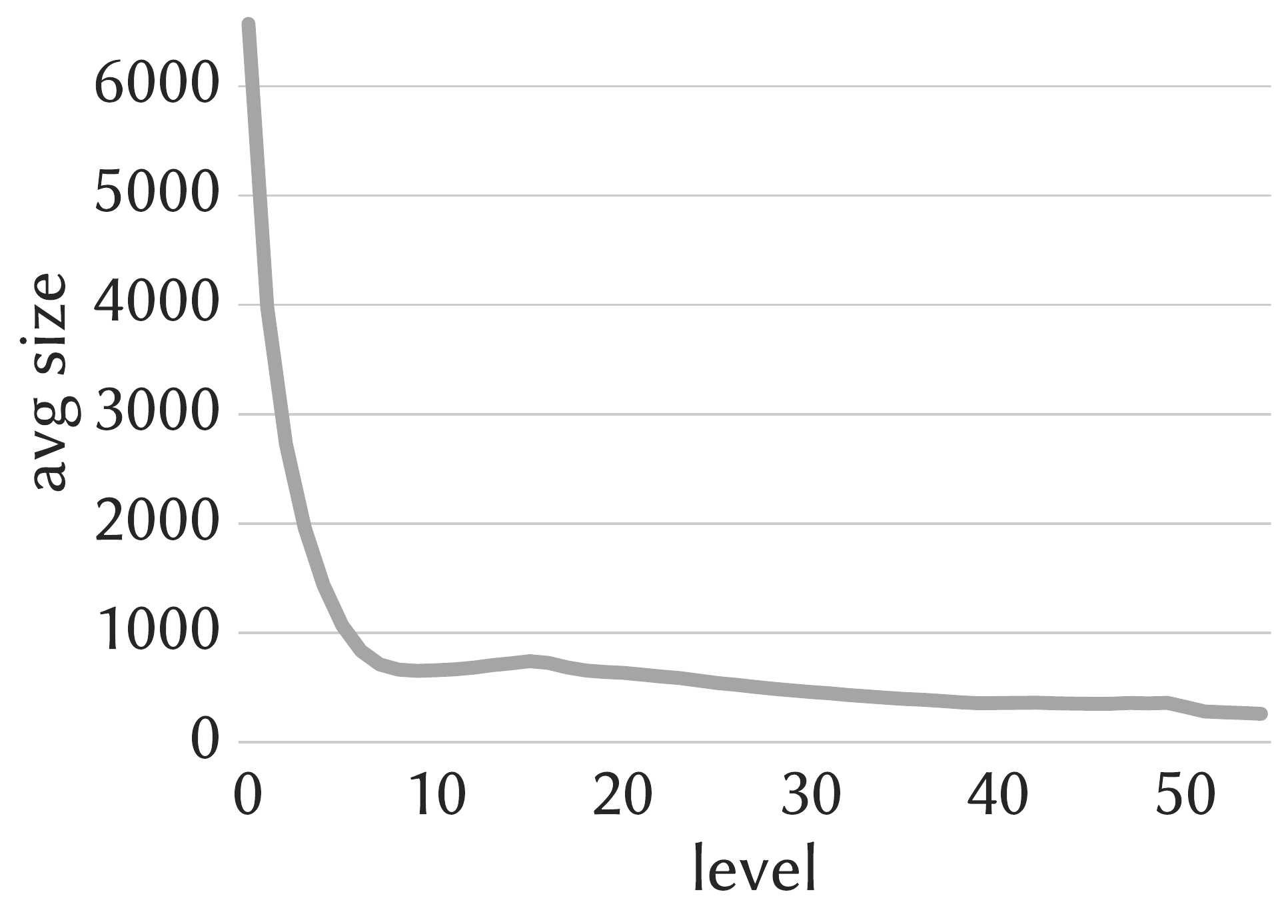}
& \includegraphics[width=0.3\columnwidth]{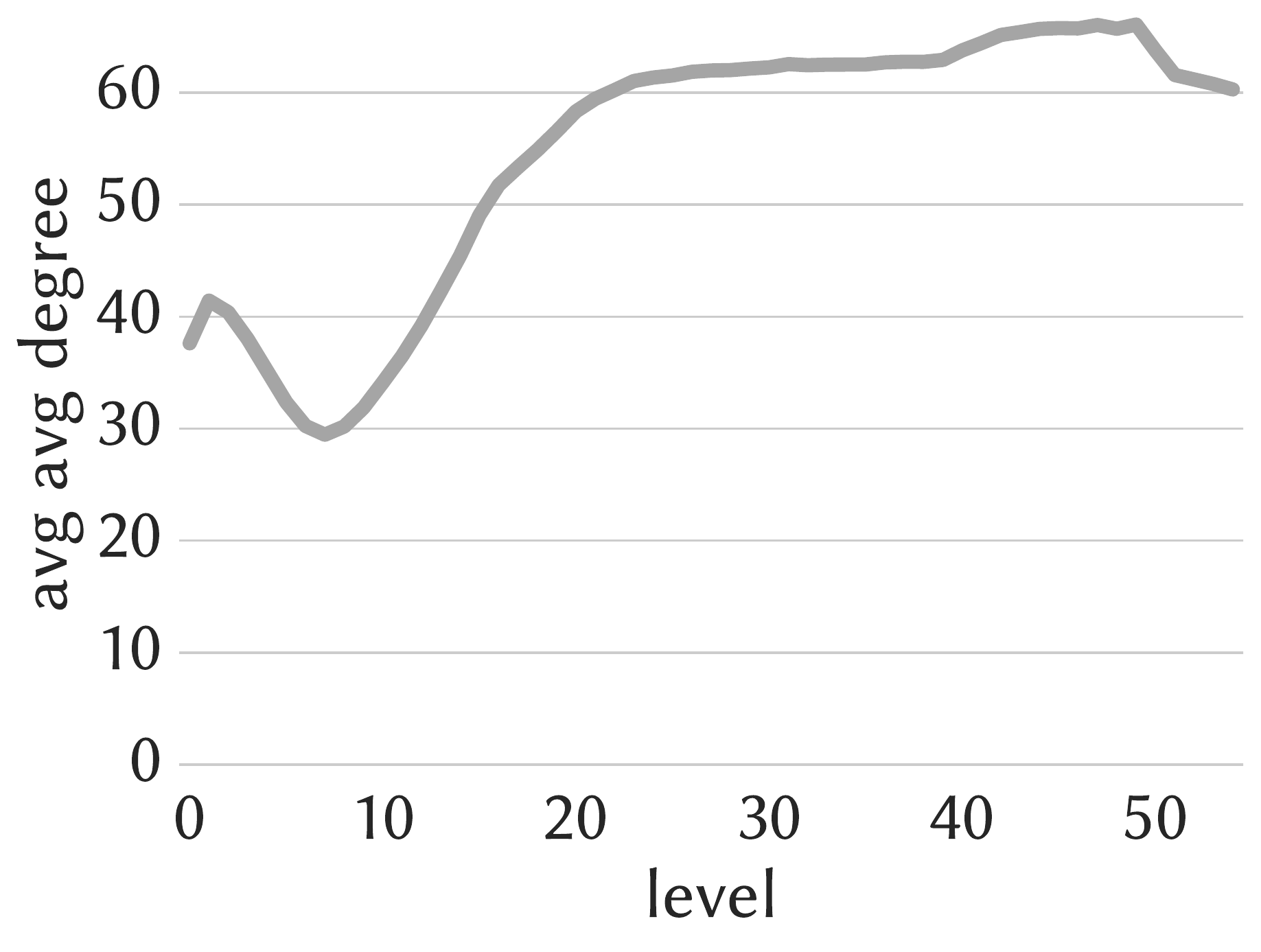}\\
\multicolumn{3}{c}{\textsf{SacchCere}}\\
\includegraphics[width=0.3\columnwidth]{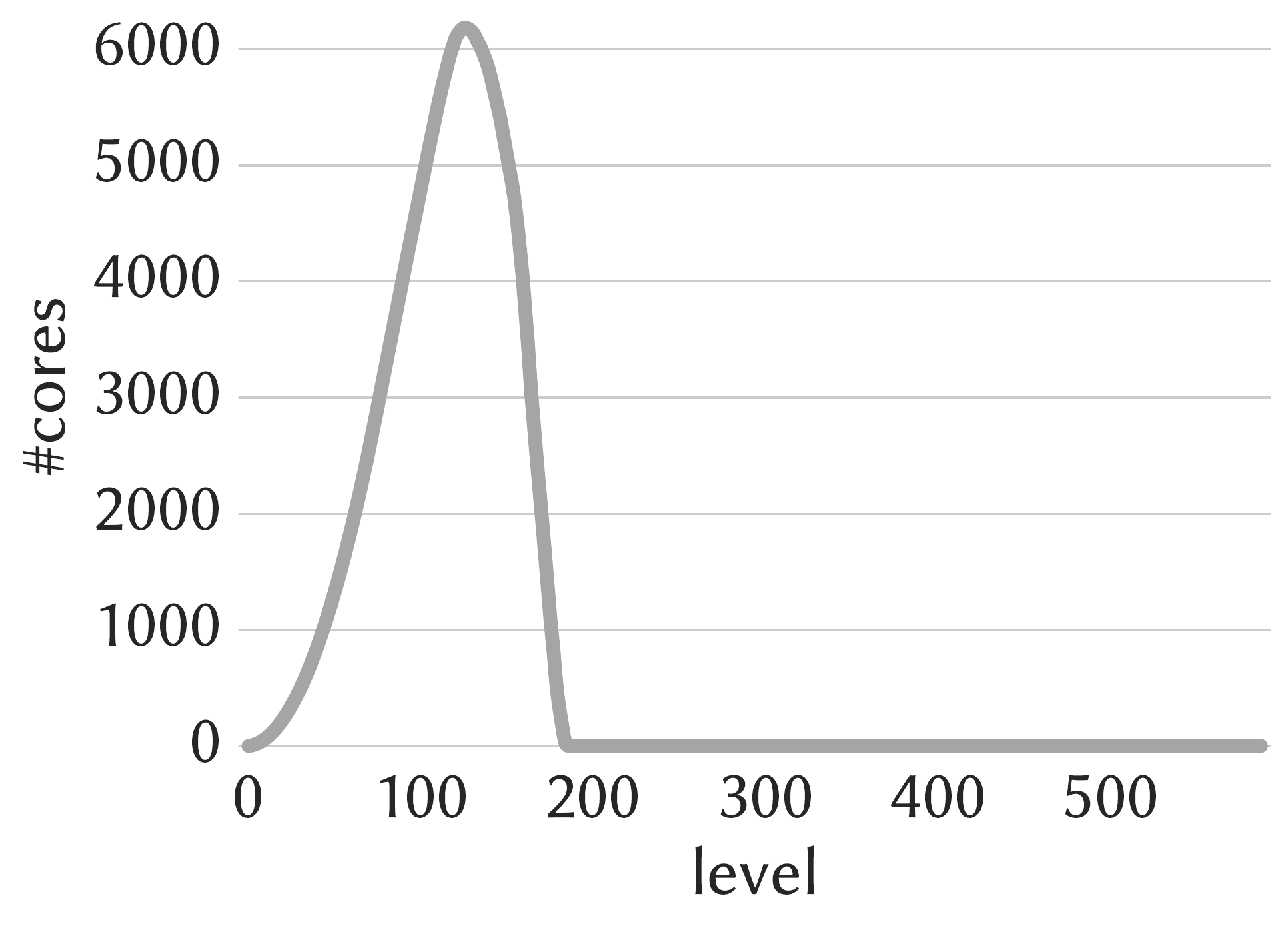} & \includegraphics[width=0.3\columnwidth]{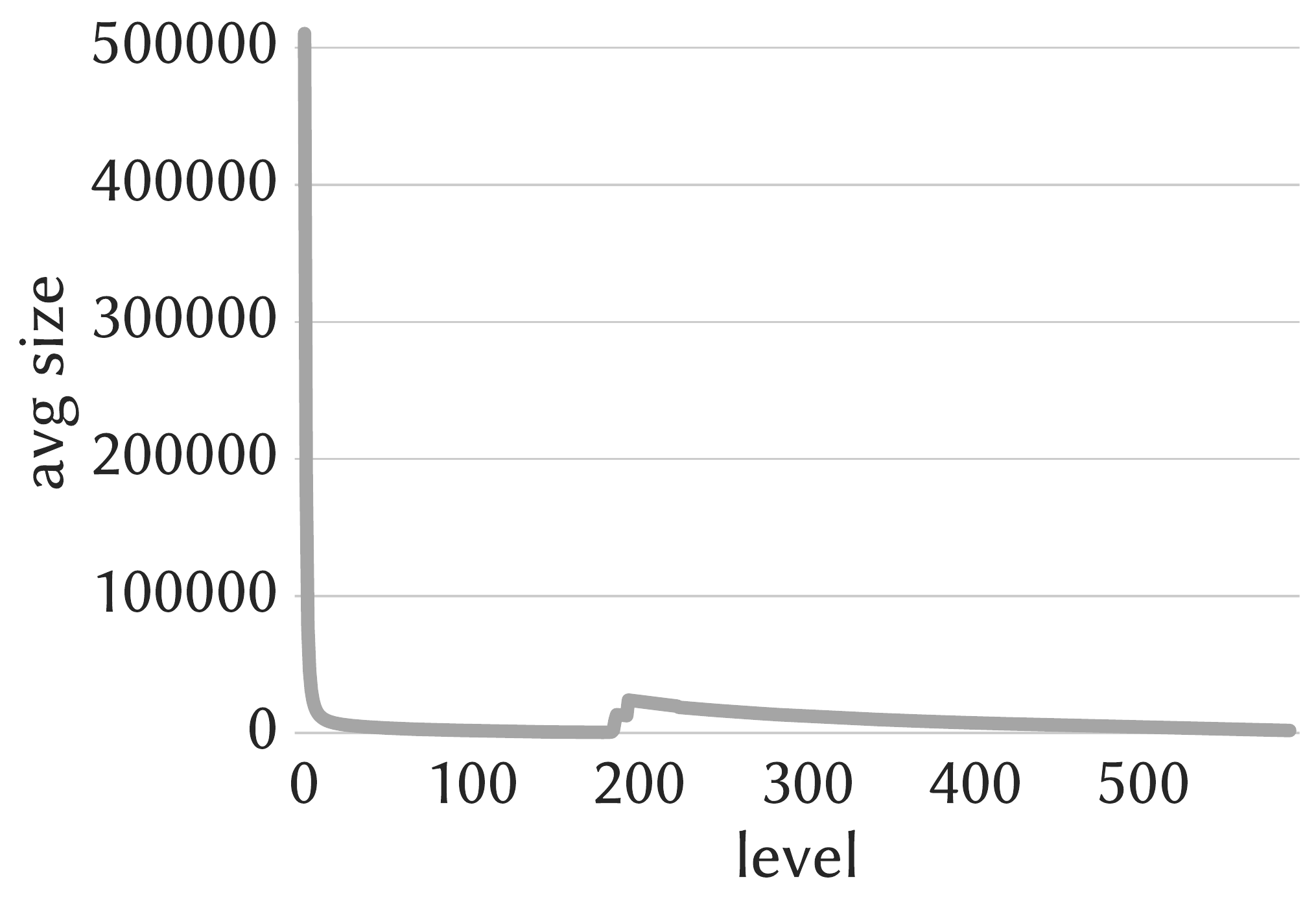}& \includegraphics[width=0.3\columnwidth]{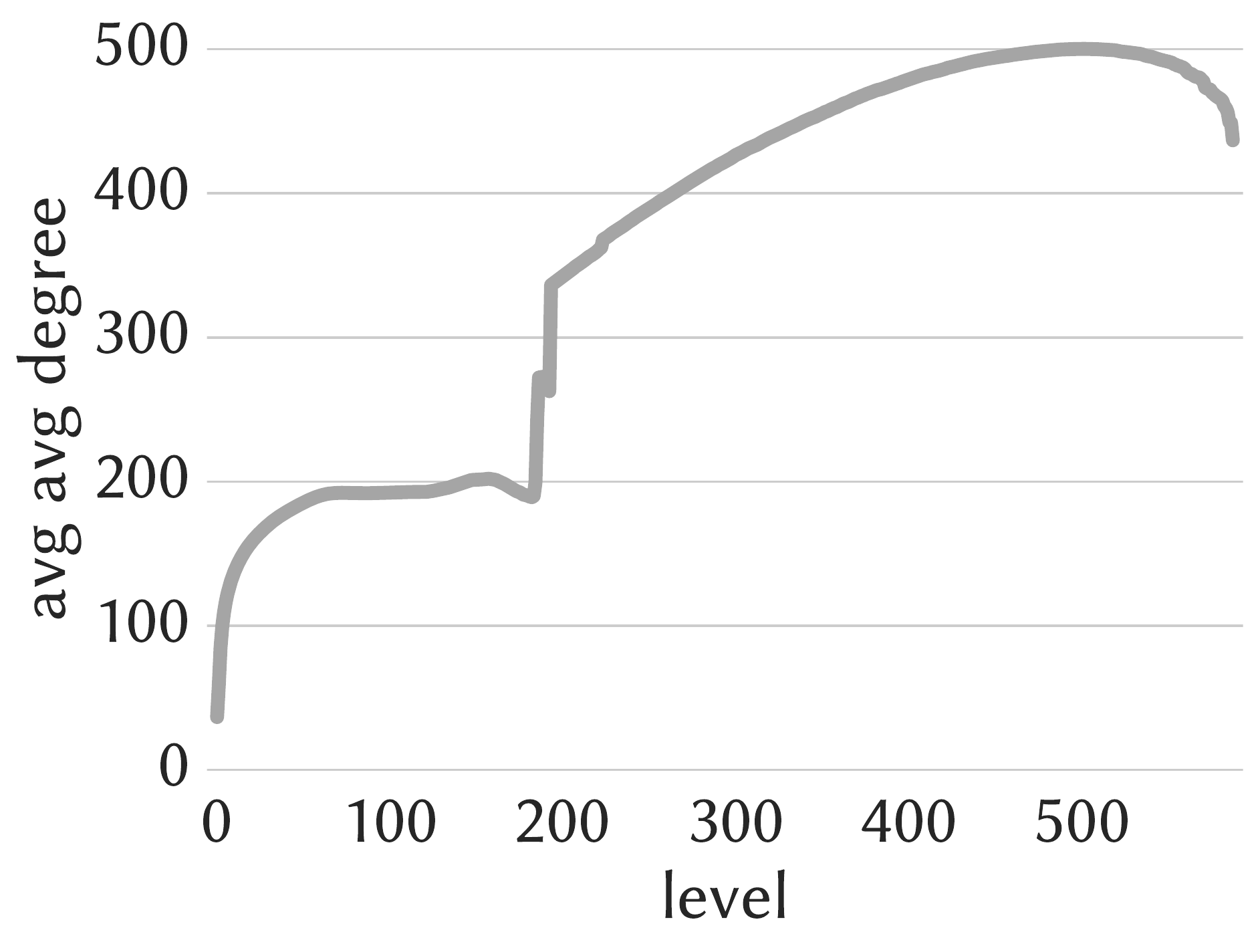}\\
 \multicolumn{3}{c}{\textsf{Friendfeed}}\\
\end{tabular}
\caption{\label{fig:levcores}  Distribution of number of cores (left), average core size (center), and average average-degree density of a core (right) to the core-lattice level, for datasets \textsf{SacchCere} (top) and \textsf{Friendfeed} (bottom).}
\end{figure}

\begin{figure}[t!]
\includegraphics[width=0.6\columnwidth]{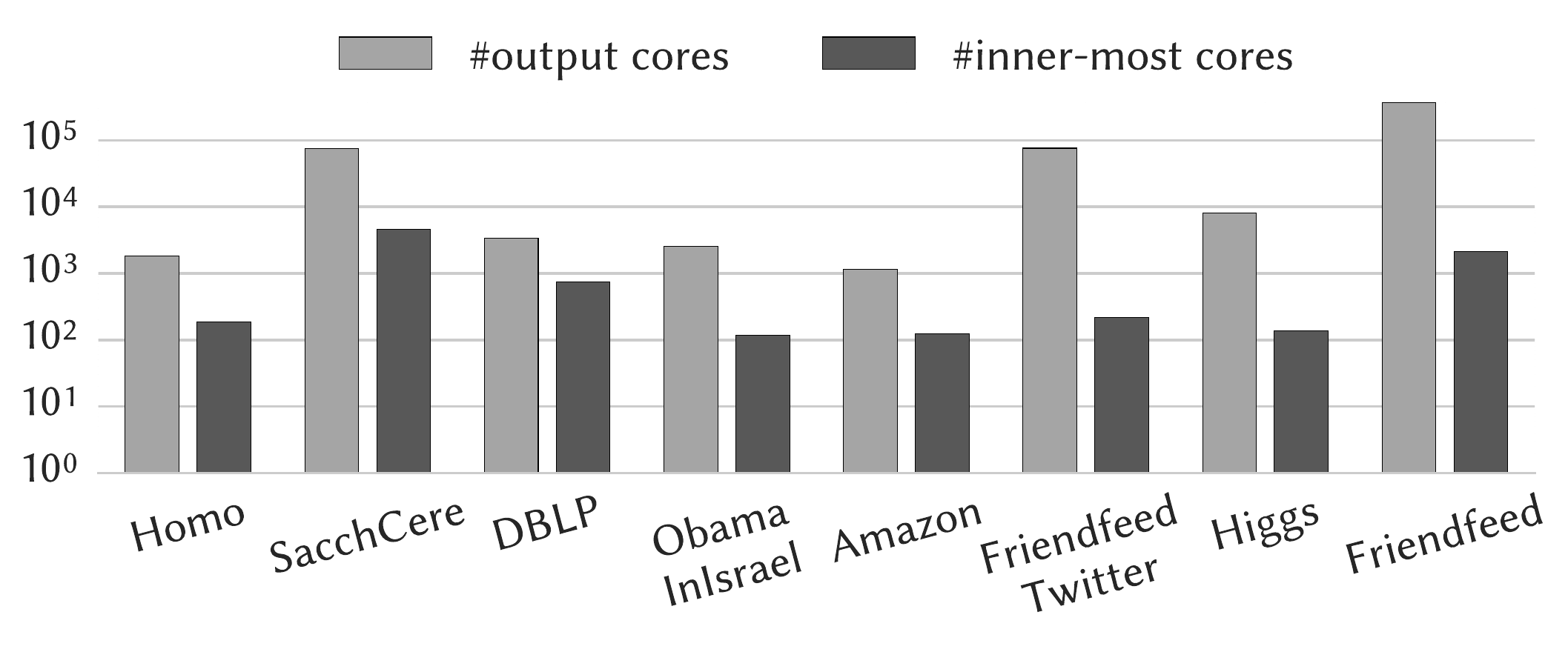}
\caption{\label{fig:numcores}   Number of output cores (total and inner-most).}
\vspace{3mm}
\end{figure}

In Figure~\ref{fig:numcores} we show the comparison between the number of all cores and inner-most cores for all the datasets.
The number of cores differs quite a lot from dataset to dataset, depending on dataset size, number of layers, and density. The fraction of inner-most cores exhibits a non-decreasing trend as the layers increase, ranging from $0.3\%$ of the total number of output cores (\textsf{FriendfeedTwitter}) to $22\%$ (\textsf{DBLP}).

\medskip

Given that the inner-most cores are \emph{per-se} interesting and typically one or more orders of magnitude fewer in number than the total cores, it would be desirable to have a method that effectively exploits the maximality property and extracts the inner-most ones directly, without computing a complete decomposition. This is presented in the next section.

\section{Algorithms for Inner-most multilayer cores}
\label{sec:innermost}

\revision{
In this section we focus on the problem of finding the non-empty inner-most multilayer cores of a multilayer graph (Problem~\ref{prob:innermost}).
Specifically, the main goal here is to devise a method that is  more efficient than a na\"ive one that computes the whole multilayer core decomposition and then a-posteriori filters non-inner-most cores out.
}
%
%
To this end, we devise a recursive algorithm, which is termed \rmax\  and whose outline is shown as Algorithm~\ref{alg:im} (and Algorithm~\ref{alg:rightim}).
We provide the details of the algorithm next.
In the remainder of this section we assume the layer set $L$ of the input multilayer graph $G=(V,E,L)$ to be an ordered list $[\ell_1, \ldots, \ell_{|L|}]$.
The specific ordering we adopt in this work is by non-decreasing average-degree density, as, among the various orderings tested, this is the one that provides the best experimental results.

The proposed  \rmax\ algorithm is based on the notion of \emph{$\ell_r$-right-inner-most multilayer cores} of a core $C_{\vec{k}}$, i.e., all those cores having \corenessvec\ $\vec{k'}$ equal to $\vec{k}$ up to layer $\ell_{r-1}$, and for which the inner-most condition holds for layers from $\ell_r$ to $\ell_{|L|}$.

\begin{mydefinition}[$\ell_r$-right-inner-most multilayer cores]
Given a multilayer graph $G=(V,E,L)$ and a layer $\ell_r \in L$, the \emph{$\ell_r$-right-inner-most multilayer cores} of a core $C_{\vec{k}}$ of $G$, where $\vec{k} = [k_{\ell}]_{\ell \in L}$, correspond to all the cores of $G$ with \corenessvec\ $\vec{k'} = [k'_{\ell}]_{\ell \in L}$ such that $\forall \ell \in [\ell_1,\ell_r): k'_{\ell} = k_{\ell}$, and there does not exist any other core with \corenessvec\  $\vec{k}'' = [k''_{\ell}]_{\ell \in L}$ such that
$\forall \ell \in [\ell_1,\ell_r): k''_{\ell} = k_{\ell}$,
$\forall \ell \in [\ell_r, \ell_{|L|}]: k''_{\ell} \geq k'_{\ell}$, and $\exists \hat{\ell} \in [\ell_r, \ell_{|L|}]: k''_{\hat{\ell}} > k'_{\hat{\ell}}$.
\end{mydefinition}

\revision{
Let $C_{[0]_{|L|}}$ be the root of the core lattice.
$C_{[0]_{|L|}}$ has a \corenessvec\ composed of zero components.
Therefore, according to the above definition, it is easy to observe that the $\ell_1$-right-inner-most multilayer cores of $C_{[0]_{|L|}}$ correspond to the desired ultimate output, i.e., to all inner-most multilayer cores of the input multilayer graph.
}

\begin{fact} \label{fact:imzero}
Given a multilayer graph $G=(V,E,L)$, let $\imcoresset_{\ell_1}$ be the set of all $\ell_1$-right-inner-most multilayer cores of core $C_{[0]_{|L|}}$.
$\imcoresset_{\ell_1}$ corresponds to all inner-most multilayer cores of $G$.
\end{fact}

\revision{
The proposed \rmax\ algorithm recursively computes $\ell_r$-right-inner-most multilayer cores, starting from the root of the core lattice (Algorithm~\ref{alg:im}).
The goal is to exploit Fact~\ref{fact:imzero} and ultimately have the $\ell_1$-right-inner-most multilayer cores of core $C_{[0]_{|L|}}$ computed.
}
The algorithm makes use of a data structure $\mathcal{M}$ which consists of a sequence of nested maps, one for each layer but the last one (i.e., $\ell_{|L|}$).
\revision{
For every layer $\ell_r$  that has been so far processed by the recursive procedure, $\mathcal{M}$ keeps track of the minimum-degree that a core should have in layer $\ell_r$ to be recognized as an ineer-most one.
}
Specifically, given a \corenessvec\ $\vec{k}$ and a layer $\ell_r$, the instruction $\mathcal{M}(\vec{k},\ell_r)$ iteratively accesses the nested maps using  the elements of $\vec{k}$ up to layer $\ell_r$ as keys.
As an example, consider a  \corenessvec\ $\vec{k} = [k_{\ell}]_{\ell \in L}$, with $|L| = 3$.
$\mathcal{M}(\vec{k},\ell_{|L|-1})$ first queries the outer-most map with key $k_{\ell_1}$, and obtains a further map.
Then, this second map is queried with key $k_{\ell_2}$, to finally get the ultimate desired numerical value.
\revision{
Note that, if  $\ell_r < \ell_{|L|-1}$, then $\mathcal{M}(\vec{k},\ell_r)$ returns a further map.
Conversely, if  $\ell_r = \ell_{|L|-1}$, then $\mathcal{M}(\vec{k},\ell_r)$ returns a numerical value.
}
If $\vec{k}$ does not correspond to a sequence of valid keys for $\mathcal{M}$, we assume that $0$ is returned as a default value.
$\mathcal{M}$ is initialized as empty, and populated during the various recursive iterations.

\begin{algorithm}[t!]
	\caption{\rmax} \label{alg:im}
	\begin{algorithmic}[1]
		
		\REQUIRE A multilayer graph $G = (V,E,L)$.
		\ENSURE The set \imcoresset\ of all inner-most multilayer cores of $G$.

        \STATE sort $L$ by non-decreasing average-degree density
        \STATE $\mathcal{M} \leftarrow \emptyset$
        \STATE $\imcoresset \leftarrow$ \rmaxsub$(G,V,[0]_{|L|},\ell_1,\mathcal{M})$
	\end{algorithmic}
\end{algorithm}
\begin{algorithm}[t!]
	\caption{\rmaxsub} \label{alg:rightim}
	\begin{algorithmic}[1]
		
		\REQUIRE A multilayer graph $G = (V,E,L)$, a set $S \subseteq V$of vertices, a \corenessvec\ $\vec{k} = [k_{\ell}]_{\ell \in L}$, a layer $\ell_r \in L$, and a data structure $\mathcal{M}$.
		\ENSURE The set \rightimcoresset\ of all right-inner-most multilayer cores of $C_\vec{k}$ given $\ell_r$.

        \STATE $\rightimcoresset \leftarrow \emptyset$

        \IF {$\ell_r \neq \ell_{|L|}$} \label{line:rightim:mainif}

            \STATE $\mathbf{Q} \leftarrow \{\vec{k}' \mid C_{\vec{k}'} \in \kcorespathalg(G,S,\vec{k},\ell_r)\} \cup \{\vec{k}\}$ \label{line:rightim:cd}
            \STATE $\coresset \leftarrow \kcorespathalg(G,S,\vec{k},\ell_r) \cup \{S\}$ \label{line:rightim:nodesets}

            \FORALL {$\vec{k'} \in \mathbf{Q}$ in decreasing order of $k'_{\ell_r}$} \label{line:rightim:rec}
                \STATE $\mathcal{M}(\vec{k'}, \ell_r) \leftarrow \emptyset$
                \STATE $\rightimcoresset \leftarrow \rightimcoresset \ \cup$ \rmaxsub$(G,C_{\vec{k'}},\vec{k'},\ell_{r+1}, \mathcal{M})$
            \ENDFOR \label{line:rightim:endfirstblock}

        \ELSE
            \STATE $k_\mathcal{M} \leftarrow 0$ \label{line:rightim:lbstart}
            \FORALL {$\ell \in [\ell_1,\ell_{|L|})$}
                \STATE $\vec{k}^\ell = [k_{\ell_1},\dots,k_{\ell}+1,\dots,k_{\ell_{|L|}}]$
                \STATE $k_\mathcal{M} \leftarrow \max\{k_\mathcal{M},\mathcal{M}(\vec{k}^\ell,\ell_{|L|-1})\}$
            \ENDFOR

            \STATE $\vec{k'} \leftarrow [k_{\ell_1},\dots,k_{\ell_{|L|-1}},k_\mathcal{M}]$ \label{line:rightim:lbend}
            \STATE $\vec{k}^I \leftarrow $ \rmaxcore$(G,S,\vec{k'},\ell_{|L|})$ \label{line:rightim:im}
            \IF {$\vec{k}^I \neq \textsc{null}$}
                \STATE $\rightimcoresset \leftarrow \rightimcoresset \cup \vec{k}^I$
                \STATE $\mathcal{M}(\vec{k}^I,\ell_{|L|-1}) \leftarrow k^I_{\ell_{|L|}}+1$
            \ELSE
                \STATE $\mathcal{M}(\vec{k'},\ell_{|L|-1}) \leftarrow k'_{\ell_{|L|}}$
            \ENDIF \label{line:rightim:endsecondblock}
        \ENDIF
	\end{algorithmic}
\end{algorithm}

\begin{figure}[t]
\centering
\tikzstyle{every node}=[draw=black,fill=none,right,align=center]
\tikzstyle{selected}=[draw=red,fill=red!30]
\tikzstyle{optional}=[dashed,fill=gray!50]
\tikzstyle{ann} = [draw=none,fill=none,left]

\tiny
\begin{tikzpicture}[%
grow via three points={one child at (0.5,-0.975) and two children at (0.5,-0.975) and (0.5,-1.95)},
edge from parent path={(\tikzparentnode.south) |- (\tikzchildnode.west)}]
\node[label={left:1}] {\underline{\textbf{\rmaxsub}$\mathbf{(V,(0,0,0),\boldsymbol\ell_1)}$}\\
$\kcorespathalg(V,(0,0,0),\ell_1)$\\
$\mathbf{Q} \leftarrow \{(2,0,0),(1,0,0),(0,0,0)\}$}
    child {node[label={[label distance=0.5cm]left:1.1}] {\underline{\textbf{\rmaxsub}$\mathbf{(C_{(2,0,0)},(2,0,0),\boldsymbol\ell_2)}$}\\
    $\kcorespathalg(C_{(2,0,0)},(2,0,0),\ell_2)$\\
    $\mathbf{Q} \leftarrow \{(2,0,0)\}$}
        child {node[label={[label distance=0.5cm]left:1.1.1}] {\underline{\textbf{\rmaxsub}$\mathbf{(C_{(2,0,0)},(2,0,0),\boldsymbol\ell_3)}$}\\
        \rmaxcore$(C_{(2,0,0)},(2,0,0),\ell_3) \rightarrow (2,0,3)$}}
    }
    child [missing] {}
    child {node[label={[label distance=0.5cm]left:1.2}] {\underline{\textbf{\rmaxsub}$\mathbf{(C_{(1,0,0)},(1,0,0),\boldsymbol\ell_2)}$}\\
    $\kcorespathalg(C_{(1,0,0)},(1,0,0),\ell_2)$\\
    $\mathbf{Q} \leftarrow \{(1,2,0),(1,1,0),(1,0,0)\}$}
        child {node[label={[label distance=0.5cm]left:1.2.1}] {\underline{\textbf{\rmaxsub}$\mathbf{(C_{(1,2,0)},(1,2,0),\boldsymbol\ell_3)}$}\\
        \rmaxcore$(C_{(1,2,0)},(1,2,0),\ell_3) \rightarrow (1,2,4)$}}
        child {node[label={[label distance=0.5cm]left:1.2.2}] {\underline{\textbf{\rmaxsub}$\mathbf{(C_{(1,1,0)},(1,1,0),\boldsymbol\ell_3)}$}\\
        \rmaxcore$(C_{(1,1,0)},(1,1,5),\ell_3) \rightarrow \textsc{null}$}}
        child {node[label={[label distance=0.5cm]left:1.2.3}] {\underline{\textbf{\rmaxsub}$\mathbf{(C_{(1,0,0)},(1,0,0),\boldsymbol\ell_3)}$}\\
        \rmaxcore$(C_{(1,0,0)},(1,0,5),\ell_3) \rightarrow (1,0,7)$}}
    }
    child [missing] {}
    child [missing] {}
    child [missing] {}
    child {node[label={[label distance=0.5cm]left:1.3}] {\underline{\textbf{\rmaxsub}$\mathbf{(C_{(0,0,0)},(0,0,0),\boldsymbol\ell_2)}$}\\
    $\kcorespathalg(C_{(0,0,0)},(0,0,0),\ell_2)$\\
    $\mathbf{Q} \leftarrow \{(0,3,0),(0,2,0),(0,1,0),(0,0,0)\}$}
        child {node[label={[label distance=0.5cm]left:1.3.1}] {\underline{\textbf{\rmaxsub}$\mathbf{(C_{(0,3,0)},(0,3,0),\boldsymbol\ell_3)}$}\\
        \rmaxcore$(C_{(0,3,0)},(0,3,0),\ell_3) \rightarrow (0,3,1)$}}
        child {node[label={[label distance=0.5cm]left:1.3.2}] {\underline{\textbf{\rmaxsub}$\mathbf{(C_{(0,2,0)},(0,2,0),\boldsymbol\ell_3)}$}\\
        \rmaxcore$(C_{(0,2,0)},(0,2,5),\ell_3) \rightarrow \textsc{null}$}}
        child {node[label={[label distance=0.5cm]left:1.3.3}] {\underline{\textbf{\rmaxsub}$\mathbf{(C_{(0,1,0)},(0,1,0),\boldsymbol\ell_3)}$}\\
        \rmaxcore$(C_{(0,1,0)},(0,1,5),\ell_3) \rightarrow (0,1,5)$}}
        child {node[label={[label distance=0.5cm]left:1.3.4}] {\underline{\textbf{\rmaxsub}$\mathbf{(C_{(0,0,0)},(0,0,0),\boldsymbol\ell_3)}$}\\
        \rmaxcore$(C_{(0,0,0)},(0,0,8),\ell_3) \rightarrow \textsc{null}$}}
    };
\end{tikzpicture}
\caption{Execution of the \rmax\ algorithm (Algorithm~\ref{alg:im}) on a toy $3$-layer graph.\label{fig:maxexec} }
\end{figure}
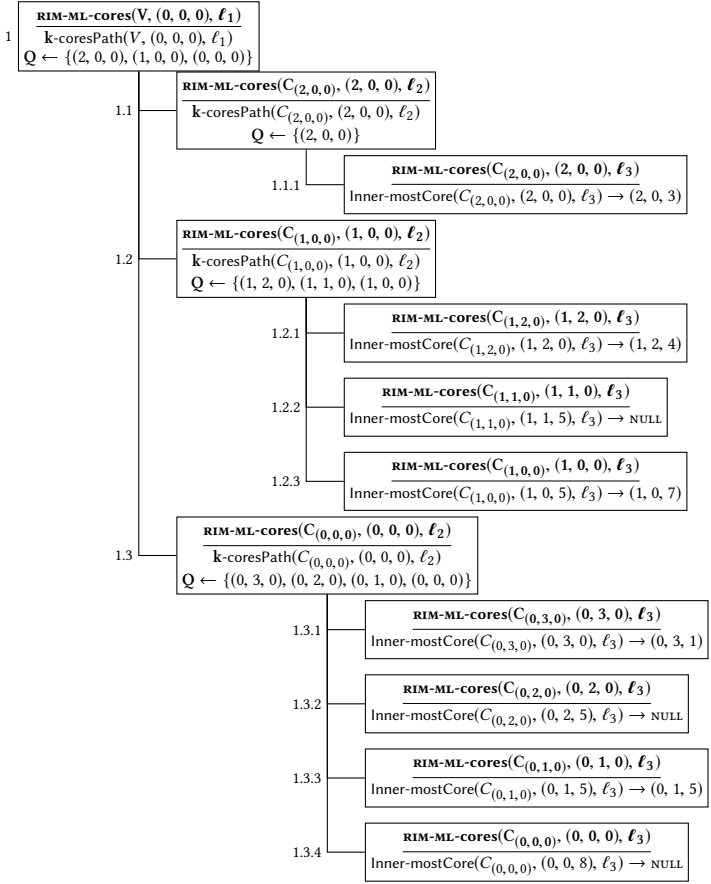

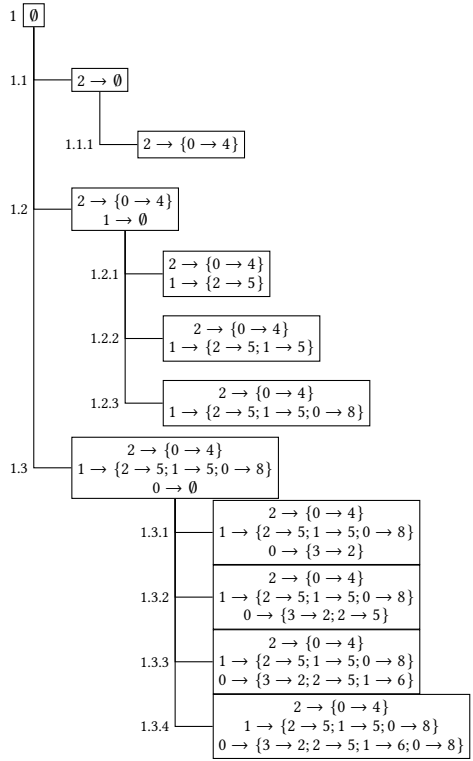
\begin{figure}[t]
\centering
\tikzstyle{every node}=[draw=black,fill=none,right,align=center]
\tikzstyle{selected}=[draw=red,fill=red!30]
\tikzstyle{optional}=[dashed,fill=gray!50]
\tikzstyle{ann} = [draw=none,fill=none,left]

\tiny
\begin{tikzpicture}[%
grow via three points={one child at (0.5,-0.85) and two children at (0.5,-0.85) and (0.5,-1.7)},
edge from parent path={(\tikzparentnode.south) |- (\tikzchildnode.west)}]
\node[label={left:$1$}] {$\emptyset$}
    child {node[label={[label distance=0.5cm]left:1.1}] {$2 \rightarrow \emptyset$}
        child {node[label={[label distance=0.5cm]left:1.1.1}] {$2 \rightarrow \{0 \rightarrow 4\}$}}
    }
    child [missing] {}
    child {node[label={[label distance=0.5cm]left:1.2}] {$2 \rightarrow \{0 \rightarrow 4\}$\\
    $1 \rightarrow \emptyset$}
        child {node[label={[label distance=0.5cm]left:1.2.1}] {$2 \rightarrow \{0 \rightarrow 4\}$\\
                                                                                         $1 \rightarrow \{2 \rightarrow 5\}$}}
        child {node[label={[label distance=0.5cm]left:1.2.2}] {$2 \rightarrow \{0 \rightarrow 4\}$\\
                                                                                         $1 \rightarrow \{2 \rightarrow 5; 1 \rightarrow 5\}$}}
        child {node[label={[label distance=0.5cm]left:1.2.3}] {$2 \rightarrow \{0 \rightarrow 4\}$\\
                                                                                         $1 \rightarrow \{2 \rightarrow 5; 1 \rightarrow 5; 0 \rightarrow 8\}$}}
    }
    child [missing] {}
    child [missing] {}
    child [missing] {}
    child {node[label={[label distance=0.5cm]left:1.3}] {$2 \rightarrow \{0 \rightarrow 4\}$\\
                                                                                   $1 \rightarrow \{2 \rightarrow 5; 1 \rightarrow 5; 0 \rightarrow 8\}$\\
                                                                                   $0 \rightarrow \emptyset$}
        child {node[label={[label distance=0.5cm]left:1.3.1}] {$2 \rightarrow \{0 \rightarrow 4\}$\\
                                                                                         $1 \rightarrow \{2 \rightarrow 5; 1 \rightarrow 5; 0 \rightarrow 8\}$\\
                                                                                         $0 \rightarrow \{3 \rightarrow 2\}$}}
        child {node[label={[label distance=0.5cm]left:1.3.2}] {$2 \rightarrow \{0 \rightarrow 4\}$\\
                                                                                         $1 \rightarrow \{2 \rightarrow 5; 1 \rightarrow 5; 0 \rightarrow 8\}$\\
                                                                                         $0 \rightarrow \{3 \rightarrow 2; 2 \rightarrow 5\}$}}
        child {node[label={[label distance=0.5cm]left:1.3.3}] {$2 \rightarrow \{0 \rightarrow 4\}$\\
                                                                                         $1 \rightarrow \{2 \rightarrow 5; 1 \rightarrow 5; 0 \rightarrow 8\}$\\
                                                                                         $0 \rightarrow \{3 \rightarrow 2; 2 \rightarrow 5; 1 \rightarrow 6\}$}}
        child {node[label={[label distance=0.5cm]left:1.3.4}] {$2 \rightarrow \{0 \rightarrow 4\}$\\
                                                                                         $1 \rightarrow \{2 \rightarrow 5; 1 \rightarrow 5; 0 \rightarrow 8\}$\\
                                                                                         $0 \rightarrow \{3 \rightarrow 2; 2 \rightarrow 5; 1 \rightarrow 6; 0 \rightarrow 8\}$}}
    };
\end{tikzpicture}
\caption{Content of the $\mathcal{M}$ data structure during the execution of the \rmax\ algorithm as per the example shown in Fig.~\ref{fig:maxexec}.\label{fig:maxstruct}}
\end{figure}

\revision{
Algorithm~\ref{alg:rightim} may be logically split into two main blocks:  the first one (Lines~\ref{line:rightim:cd}~--~\ref{line:rightim:endfirstblock}) taking care of the  recursion, and the second one (Lines~\ref{line:rightim:lbstart}~--~\ref{line:rightim:endsecondblock}) computing the $\ell_r$-right-inner-most cores.
}
The first block of the algorithm is executed when the current $\ell_r$ layer is not the last one.
\revision{
In that block the \kcorespathalg subroutine (already used in Algorithm~\ref{alg:dfs} and described in Section~\ref{sec:alg:dfs}) is run on set $S$ of vertices, layer $\ell_r$, and taking into account the constraints in vector $\vec{k}$ (Lines~\ref{line:rightim:cd}~and~\ref{line:rightim:nodesets}).
Then, for each \corenessvec\ $\vec{k'}$ that has been found, a recursive call is made, where the layer of interest becomes the next layer $\ell_{r+1}$, and the data structure $\mathcal{M}$ is augmented by adding a further (empty) nested map (this new map will be populated within the upcoming recursive executions).
}
The \corenessvec s are processed in decreasing order of $k'_{\ell_r}$.
\revision{
This processing order ensures the correctness of the following: once a multilayer core has been identified as $\ell_r$-right-inner-most, it permanently becomes part of the ultimate output cores (no further recursive call will remove it from the output).
}
Note also that, for each $\vec{k'}$, 
\revision{
\rmaxsub\ can be run on $C_{\vec{k'}}$ only, i.e., the core of \corenessvec\ $\vec{k'}$.
}
This guarantees better efficiency, without affecting correctness.

\revision{
The second block of the algorithm (Lines~\ref{line:rightim:lbstart}~--~\ref{line:rightim:endsecondblock}) works as follows.
}
When the last layer has been reached, i.e., $\ell_r = \ell_{|L|}$, the current recursion ends, and an $\ell_r$-right-inner-most multilayer core is returned (if any).
First of all, the algorithm computes a  \corenessvec\ $\vec{k'}$ which is potentially $\ell_r$-right-inner-most (Lines~\ref{line:rightim:lbstart}~--~\ref{line:rightim:lbend}).
In this regard, note that the $k_\mathcal{M}$ value is derived from the information that has been stored in $\mathcal{M}$ in the earlier recursive iterations.
Finally, the algorithm computes the inner-most core in $\ell_{|L|}$ constrained by $\vec{k'}$, by means of the \rmaxcore\ subroutine.
\revision{
Such a subroutine, similarly to the \kcorespathalg one, takes as input a multilayer graph $G$, a subset $S$ of vertices, a \corenessvec\ $\vec{k}$, and a layer $\ell$. It returns the multilayer core having \corenessvec\ of highest $\ell$-th component of the vertices in $S$, considering the constraints specified in $\vec{k}$.
}
If the \rmaxcore\ procedure actually returns a multilayer core, then it is guaranteed that such a core is $\ell_r$-right-inner-most, and is therefore added to the solution (and $\mathcal{M}$ is updated accordingly).

In Figure~\ref{fig:maxexec} we show an example of the execution of the proposed \rmax\ algorithm for a simple 3-layer graph, while Figure~\ref{fig:maxstruct} reports the content of the $\mathcal{M}$ data structure for this example.
Every box corresponds to a call of Algorithm~\ref{alg:rightim}, for which we specify ($i$) the input parameters ($G$ and $\mathcal{M}$ are omitted for the sake of brevity), 
\revision{
($ii$) the calls to the \kcorespathalg or \rmaxcore\ subroutines, and ($iii$) the content of $\mathbf{Q}$ (when it is instantiated).
}
For instance, the \corenessvec\ given as input to \rmaxcore\ at box $1.3.4$ has the last element equal to the maximum between what is stored in $\mathcal{M}$ at the end of the paths $1 \rightarrow 0$ and $\ 0 \rightarrow 1$, i.e., $8$ and $5$, that have been set at boxes $1.2.3$ and $1.3.3$, respectively.

\subsection{Experimental results}
\spara{Running times.}
We asses the efficiency of \rmax\ (for short \textsc{im}) by comparing it  to the aforementioned na\"ive approach for computing inner-most multilayer cores, which consists in firstly computing all multilayer cores (by means of one of the three algorithms presented in Section~\ref{sec:algorithms}) and filtering out the non-inner-most ones.
The results of this experiment are reported in Table~\ref{tab:r_results}.
First of all, it can be observed that the a-posteriori filtering of the inner-most multilayer cores does not consistently affect the runtime of the algorithms for multilayer core decomposition: this means that most of the time is spent for computing the overall core decomposition.
The main outcome of this experiment is that the running time of the proposed  \textsc{im} method is smaller than the time required by \textsc{bfs}, \textsc{dfs}, or \textsc{h} summed up to the time spent in the a-posteriori filtering, with considerable speed-up  from $1.3$ to an order of magnitude on the larger datasets, e.g., \textsf{FriendfeedTwitter} and \textsf{Friendfeed}.
The only exception is on the \textsf{DBLP} dataset where \textsc{bfs} and \textsc{h} run slightly faster, probably due to fact that its edges are (almost) equally distributed among the layers, which makes the effectiveness of the ordering vanish.

\begin{table}[t]
\centering
\caption{Runtime (in seconds) of the methods for multilayer core decomposition, the a-posteriori filtering of the inner-most multilayer cores, and the proposed \rmax\ method for directly computing inner-most multilayer cores.\label{tab:r_results}}

\begin{tabular}{c|ccc|c|c}
\multicolumn{1}{c}{dataset} & \multicolumn{1}{c}{\textsc{bfs}} & \multicolumn{1}{c}{\textsc{dfs}} & \multicolumn{1}{c}{\textsc{h}} & \multicolumn{1}{c}{filtering} & \multicolumn{1}{c}{\textsc{im}}\\
\hline
\textsf{Homo} & $13$ & $27$ & $12$ & $0.5$ & $5$ \\
\textsf{SacchCere} & $1\,134$ & $2\,627$ & $1\,146$ & $24$ & $336$ \\
\textsf{DBLP} & $68$ & $282$ & $29$ & $1$ & $148$ \\
\textsf{ObamaInIsrael} & $226$ & $150$ & $177$ & $7$ & $120$ \\
\textsf{Amazon} & $3\,981$ & $5\,278$ & $3\,913$ & $129$ & $2\,530$ \\
\textsf{FriendfeedTwitter} & $61\,113$ & $1\,973$ & $59\,520$ & $276$ & $1\,583$ \\
\textsf{Higgs} & $2\,480$ & $640$ & $2\,169$ & $33$ & $356$ \\
\textsf{Friendfeed} & $58\,278$ & $13\,356$ & $47\,179$ & $394$ & $2\,640$ \\
\hline
\end{tabular}

\end{table}

\spara{Characterization.}
We also show the characteristics of the inner-most multilayer cores.
Figure~\ref{fig:levcores_im} reports  the distribution of number, size, and average-degree density of all cores and inner-most cores only.
Distributions are shown in a way similar to what previously done in Figure~\ref{fig:levcores}, i.e., by level of the core lattice, and for the \textsf{SacchCere} and \textsf{Amazon} datasets.

For both datasets, there are no inner-most cores in the first levels of the lattice.
As expected, the number of inner-most cores considerably increases when the number of all cores decreases.
This is due to the fact that some cores stop propagating throughout the lattice, hence they are recognized as inner-most.
In general, inner-most cores are on average smaller than all multilayer cores.
Nonetheless, for the levels 12 and 13 of the \textsf{Amazon} dataset, inner-most cores have greater size than all cores.
This behavior is consistent with our definitions: inner-most cores are cores without descendants, thus they are expected to be the smallest-sized ones, but they do not necessarily have to.
Finally, the distribution of the average-degree density exhibits a similar trend to the distribution of the size: this is expected as the two measures depend on each other.

\begin{figure}[t!]
\begin{tabular}{ccc}
& \includegraphics[width=0.3\columnwidth]{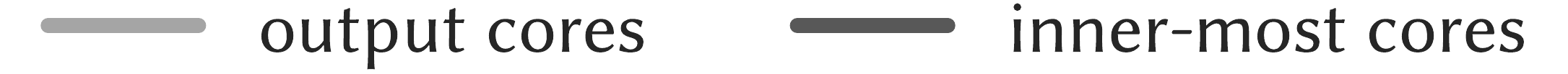} & \\
\includegraphics[width=0.3\columnwidth]{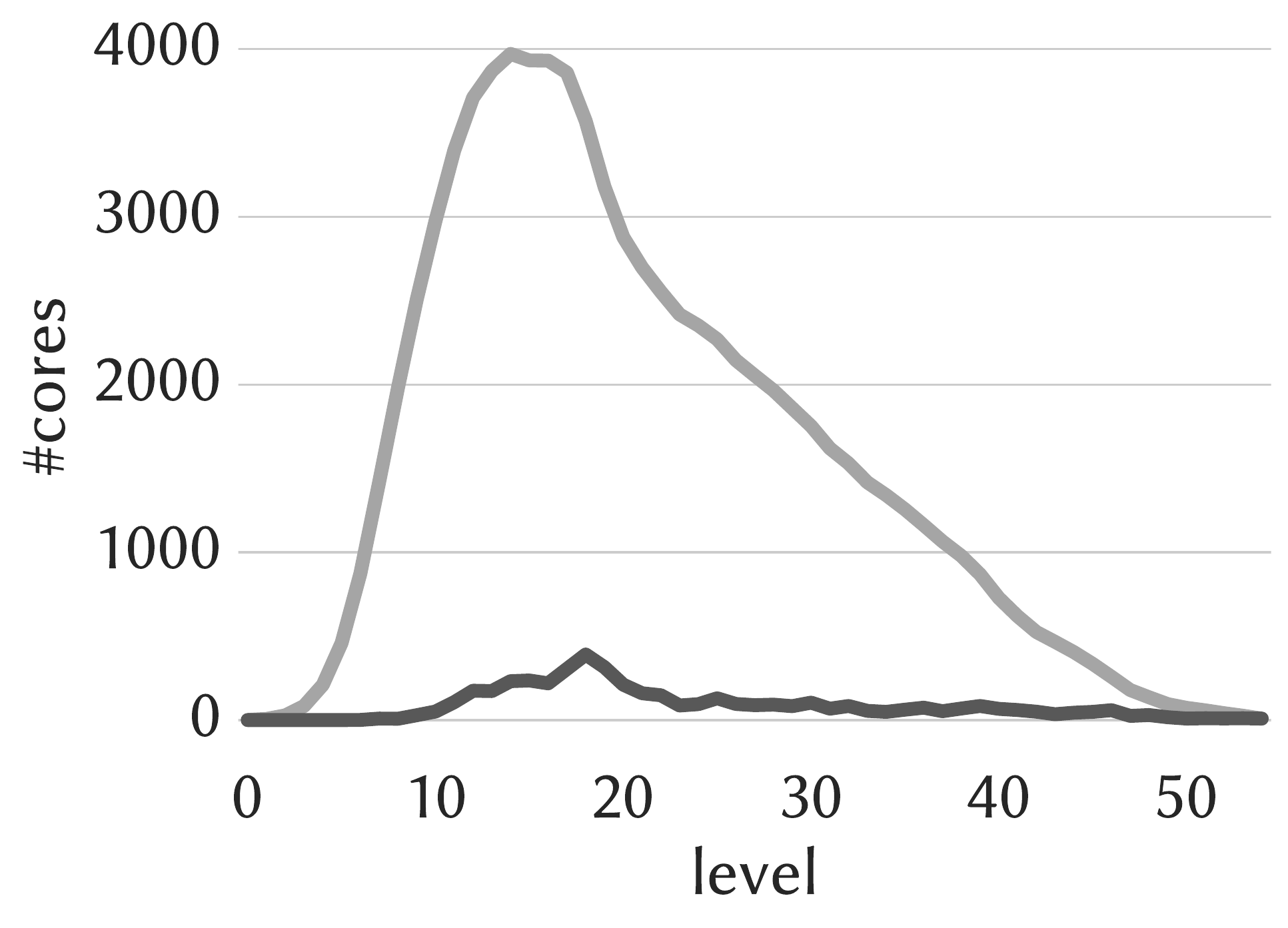} & \includegraphics[width=0.3\columnwidth]{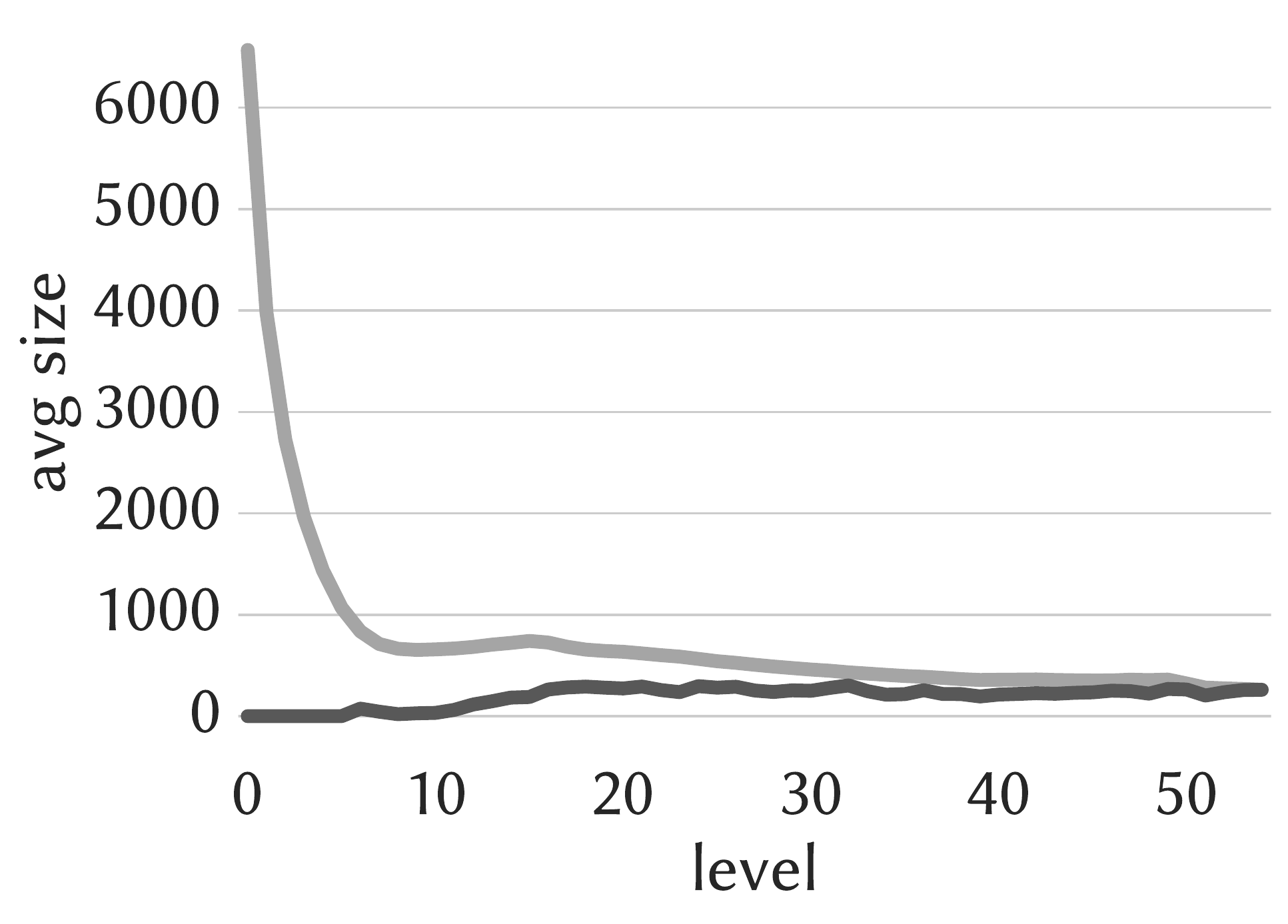}
& \includegraphics[width=0.3\columnwidth]{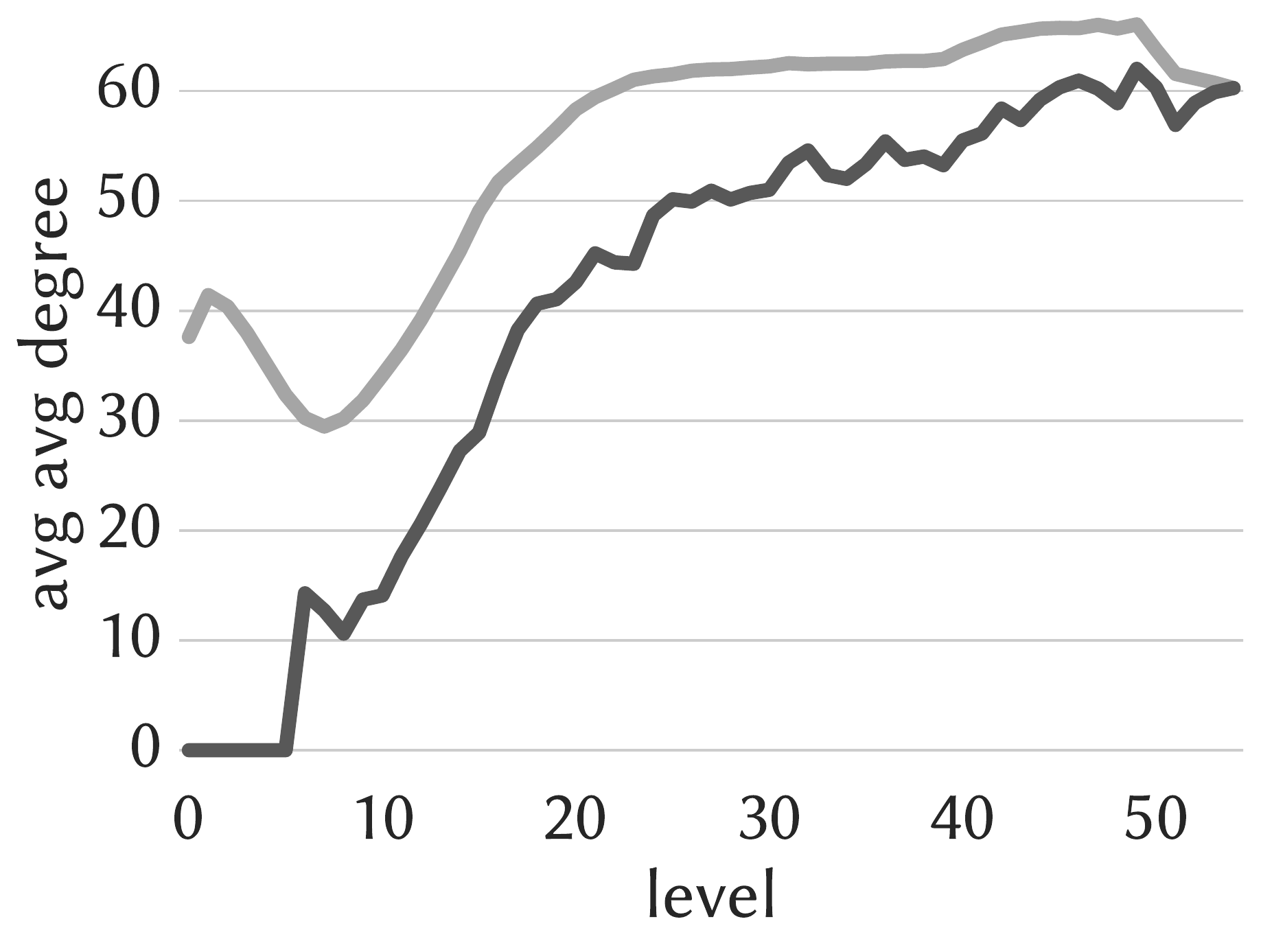}\\
\multicolumn{3}{c}{\textsf{SacchCere}}\\
& \includegraphics[width=0.3\columnwidth]{figures/im_legenda.pdf} & \\
\includegraphics[width=0.3\columnwidth]{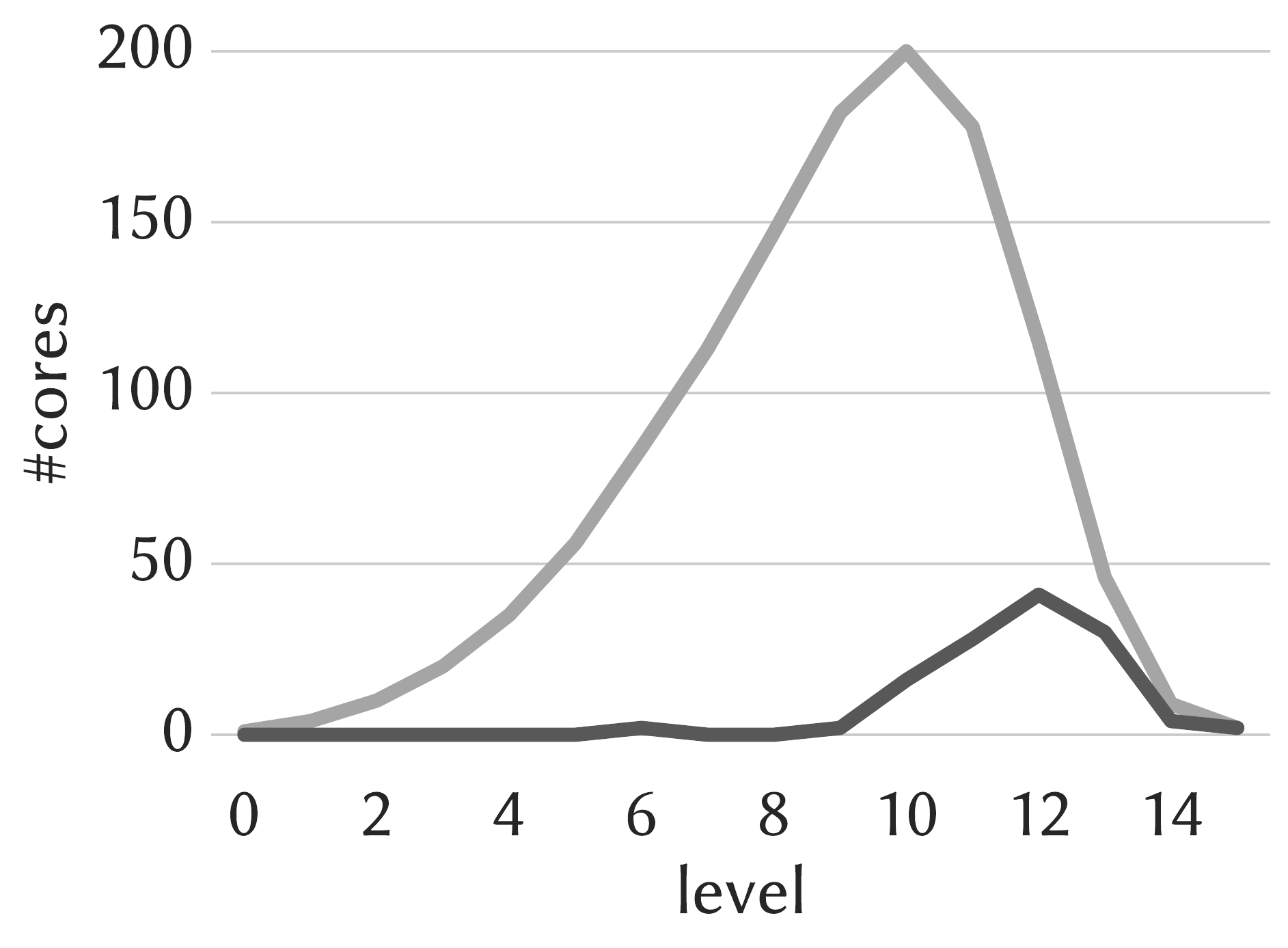} & \includegraphics[width=0.3\columnwidth]{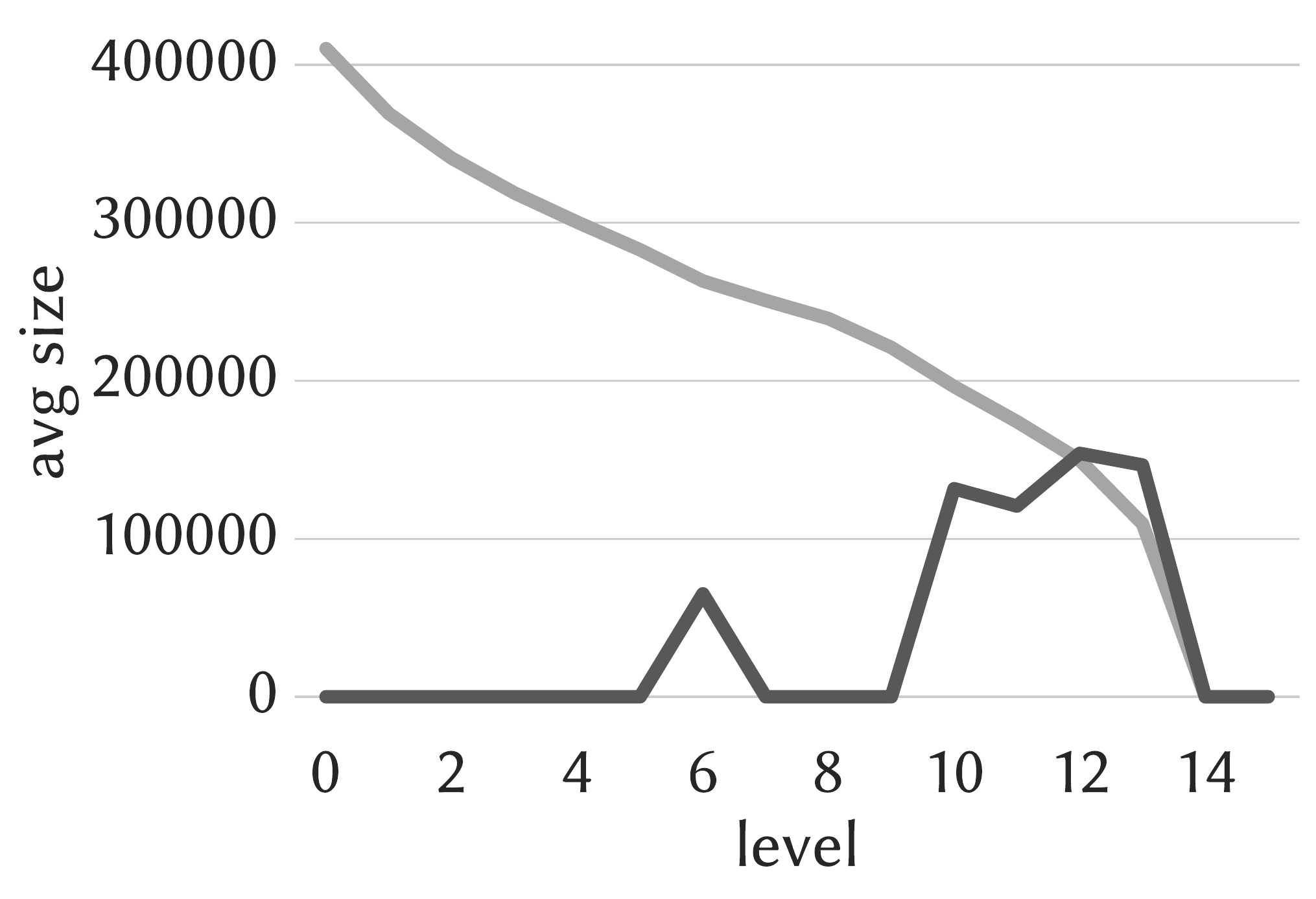}& \includegraphics[width=0.3\columnwidth]{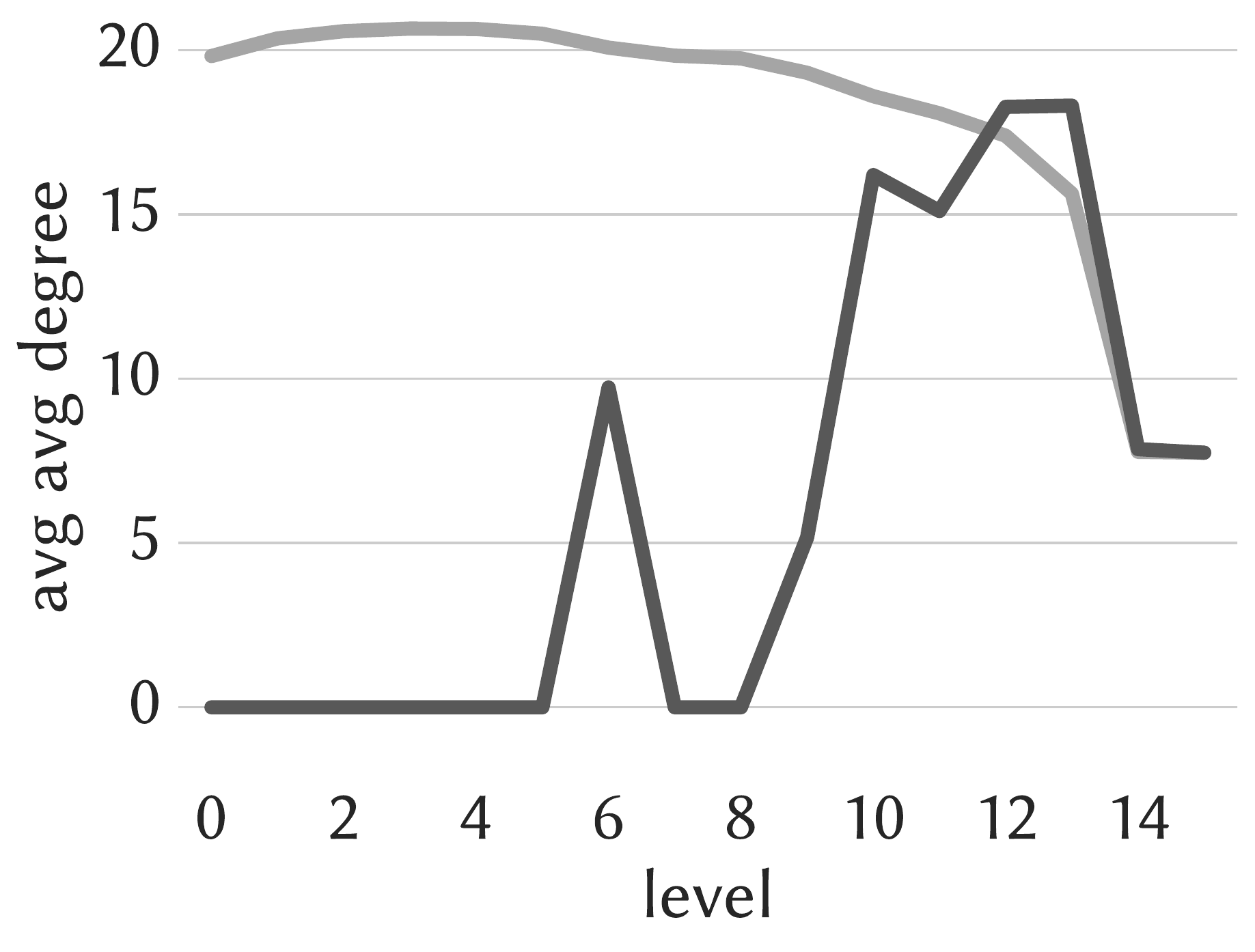}\\
 \multicolumn{3}{c}{\textsf{Amazon}}\\
\end{tabular}
\caption{\label{fig:levcores_im} Comparison of the distributions, to the core-lattice level, of number (left), average size (center), and average average-degree density (right) of multilayer cores and inner-most multilayer cores, for datasets \textsf{SacchCere} (top) and \textsf{Amazon} (bottom).}
\end{figure}

\section{Multilayer densest subgraph}
\label{sec:densest}
In this section we showcase the usefulness of multilayer core-decomposition in the context of multilayer densest-subgraph discovery.
Particularly, we show how to exploit the multilayer core-decomposition to devise an algorithm with approximation guarantees for the \mldensestsubgraph\ problem introduced in Section~\ref{sec:problems} (Problem~\ref{prob:mldensestsubgraph}), thus extending to the multilayer setting the intuition at the basis of the well-known $\frac{1}{2}$-approximation algorithm~\cite{AITT00,Char00} for single-layer densest-subgraph extraction.

%
%

\subsection{Hardness}
We start by formally showing that the \mldensestsubgraph\ problem (Problem~\ref{prob:mldensestsubgraph}) is \NPhard.

\begin{mytheorem}\label{th:densesthard}
Problem~\ref{prob:mldensestsubgraph} is \NPhard.
\end{mytheorem}
To prove the theorem, we introduce two variants of Problem~\ref{prob:mldensestsubgraph}'s objective function, i.e., $\density_{\textsc{all}}(\cdot)$, which considers all layers in $L$, and $\density_{\textsc{$\neg$all}}(\cdot)$, which considers all subsets of layers but the whole layer set $L$.
Specifically, for any given multilayer graph $G = (V, E, L)$ and vertex subset $S \subseteq V$, the two functions are defined as:  
\begin{equation}\label{eq:fodensestAL}
\density_{\textsc{all}}(S) = \min_{\ell \in L} \frac{|E_\ell[S]|}{|S|}|L|^{\beta},
\end{equation}
\begin{equation}\label{eq:fodensestNAL}
\density_{\textsc{$\neg$all}}(S) = \max_{\hat{L} \in 2^L \setminus\{L\}}\min_{\ell \in \hat{L}} \frac{|E_\ell[S]|}{|S|}|\hat{L}|^{\beta}.
\end{equation}
We also define $\dmax$ as the maximum degree of a vertex in a layer:
\begin{equation}\label{eq:degmax}
\dmax = \max_{\ell \in L} \max_{u \in V}deg(u, \ell),
\end{equation}
and introduce the following three auxiliary lemmas.
\begin{mylemma}\label{lemma:densesthard1}
$\density_{\textsc{all}}(S) \geq \frac{1}{|V|}|L|^\beta$, for all $S \subseteq V$ such that $\forall \ell \in L:|E_\ell[S]| > 0$.
\end{mylemma}
\begin{proof}
For a vertex set $S$ spanning at least one edge in every layer, it holds that $\min_{\ell \in L}\frac{|E_\ell[S]|}{|S|} \geq \frac{1}{|V|}$, and, therefore, $\density_{\textsc{all}}(S) = \min_{\ell \in L}\frac{|E_\ell[S]|}{|S|}|L|^\beta \geq \frac{1}{|V|}|L|^\beta$.
\end{proof}

\begin{mylemma}\label{lemma:densesthard2}
$\density_{\textsc{$\neg$all}}(S) \leq \frac{\dmax}{2}(|L|-1)^\beta$, for all $S \subseteq V$.
\end{mylemma}
\begin{proof}
The maximum density of a vertex set $S$ in a layer can be at most equal to the density of the maximum clique, i.e., at most $\frac{(\dmax+1)~\dmax}{2~(\dmax+1)} = \frac{\dmax}{2}$.
At the same time, the size of a layer set $\hat{L}$ in the function $\density_{\textsc{$\neg$all}}(\cdot)$ can be at most $|L| - 1$ (as the whole layer set $L$ is not considered in $\density_{\textsc{$\neg$all}}(\cdot)$). 
This means that $\density_{\textsc{$\neg$all}}(S) = \max_{\hat{L} \in 2^L \setminus\{L\}}\min_{\ell \in \hat{L}} \frac{|E_\ell[S]|}{|S|}|\hat{L}|^{\beta} \leq \frac{\dmax}{2}(|L|-1)^\beta$.
\end{proof}

\begin{mylemma}\label{lemma:densesthard3}
$$
\beta > \frac{\log_{|L|-1}\left(\frac{|V|}{2}\dmax\right) \times \log_{|L|}(|L| -1)}{1 - \log_{|L|}(|L| - 1)} \ \ \Leftrightarrow \ \ \frac{1}{|V|}|L|^\beta  >  \frac{\dmax}{2}(|L| - 1)^\beta.
$$
\end{mylemma}
\begin{proof}
\begin{eqnarray*}
\lefteqn{\beta \ > \ \frac{\log_{|L|-1}\left(\frac{|V|}{2}\dmax\right) \times \log_{|L|}(|L| -1)}{1 - \log_{|L|}(|L| - 1)}}\\
 & \Leftrightarrow & \left(1 - \log_{|L|}(|L| - 1)\right)\beta \ > \ \log_{|L|-1}\left({\textstyle\frac{|V|}{2}}\dmax\right) \times \log_{|L|}(|L| -1)\\
 & \Leftrightarrow & \beta \ > \ \log_{|L|-1}\left({\textstyle\frac{|V|}{2}}\dmax\right) \times \log_{|L|}(|L| -1) + \beta \log_{|L|}(|L| - 1)\\
 & \Leftrightarrow & \frac{\beta}{\log_{|L|}(|L| - 1)} \ > \ \log_{|L|-1}\left({\textstyle\frac{|V|}{2}}\dmax\right) + \beta\\
 & \Leftrightarrow & \frac{\log_{|L|}|L|^\beta}{\log_{|L|}(|L| - 1)} \ > \ \log_{|L|-1}\left({\textstyle\frac{|V|}{2}}\dmax\right) + \log_{|L| - 1}(|L| - 1)^\beta\\
 & \Leftrightarrow & \log_{|L|-1}|L|^\beta \ > \ \log_{|L| - 1}\left({\textstyle\frac{|V|}{2}}\dmax(|L| - 1)^\beta\right)\\
 & \Leftrightarrow & |L|^\beta \ > \ {\textstyle\frac{|V|}{2}}\dmax(|L| - 1)^\beta\\
 & \Leftrightarrow & \frac{1}{|V|}|L|^\beta \ > \ \frac{\dmax}{2}(|L| - 1)^\beta.
\end{eqnarray*}
\end{proof}

\noindent
With Lemmas~\ref{lemma:densesthard1}--\ref{lemma:densesthard3} in place, we are now ready to provide the ultimate proof of Theorem~\ref{th:densesthard}.

\begin{proof}
We reduce from the \textsc{Min-Avg Densest Common Subgraph} (DCS-MA) problem~\cite{jethava2015finding}, which aims at finding a subset of vertices $S \subseteq V$ from a multilayer graph $G = (V,L,S)$ maximizing $\min_{\ell \in L}\frac{E_\ell[S]}{|S|}$, and has been recently shown to be \NPhard\ in~\cite{charikar2018finding}.
We distinguish two cases. The first (trivial) one is when $G$ has a layer with no edges. In this case any vertex subset would be an optimal solution for DCS-MA (with overall objective function equal to zero), including the optimal solution to our \mldensestsubgraph\ problem run on the same $G$ (no matter which $\beta$ is used).
In the second case $G$ has at least one edge in every layer.
In this case solving our \mldensestsubgraph\ problem on $G$, with $\beta$ set to any value $> \frac{\log_{|L|-1}\left(\frac{|V|}{2}\dmax\right) \times \log_{|L|}(|L| -1)}{1 - \log_{|L|}(|L| - 1)}$, gives a solution that is optimal for DCS-MA as well.
Indeed, it can be observed that, for all $S \subseteq V$ such that $\forall \ell \in L:|E_\ell[S]| > 0$:
\begin{eqnarray*}
\density_{\textsc{all}}(S) & \geq & \frac{1}{|V|}|L|^\beta\eqncomment{\qquad\qquad\qquad}{Lemma~\ref{lemma:densesthard1}}\\
& > & \frac{\dmax}{2}(|L|-1)^\beta\eqncomment{~\qquad}{Lemma~\ref{lemma:densesthard3}}\\
& \geq & \density_{\textsc{$\neg$all}}(S).\eqncomment{\qquad\qquad\qquad}{Lemma~\ref{lemma:densesthard2}}
\end{eqnarray*}
This means that, for that particular value of $\beta$, the optimal solution of \mldensestsubgraph\ on input $G$ is given by maximizing the $\density_{\textsc{all}}(\cdot)$ function, which considers all layers and is, as such, equivalent to the objective function underlying the DCS-MA problem. This completes the proof.
\end{proof}

\subsection{Algorithms}
%
The approximation algorithm we devise for the \mldensestsubgraph\ problem is very simple: it computes the multilayer core decomposition of the input graph, and, among all  cores, takes the one maximizing the objective function $\density$ as the output densest subgraph (Algorithm~\ref{alg:densest}).
Despite its simplicity, the algorithm achieves provable approximation guarantees proportional to the number of layers of the input graph, precisely equal to $\frac{1}{2|L|^{\beta}}$.
We next formally prove this result.

\begin{algorithm}[t]
\caption{\densestalg} \label{alg:densest}
\begin{algorithmic}[1]

\REQUIRE A multilayer graph $G = (V,E,L)$ and a real number $\beta \in \mathbb{R}^+$.
\ENSURE $C^* \subseteq V$.

\STATE $\mathbf{C} \leftarrow \mbox{\textsf{MultiLayerCoreDecomposition}}(G)$\hfill\COMMENT{\revision{Any of Algorithms~\ref{alg:bfs}, \ref{alg:dfs}, \ref{alg:hybrid} can be used}}
\STATE $C^* \leftarrow \arg\max_{C \in \mathbf{C}} \density(C)$\hfill\COMMENT{Equation~(\ref{eq:densestfunction})}
\end{algorithmic}
\end{algorithm}

Let $\mathbf{C}$ be the core decomposition of the input multilayer graph $G = (V,E,L)$ and $C^*$ denote the core in $\mathbf{C}$ maximizing the density function $\density$, i.e., $C^* = \arg\max_{C \in \mathbf{C}} \delta(C)$.
Then, $C^*$ corresponds to the subgraph output by the proposed \densestalg\ algorithm.
Let also $\SLinnermost$ denote the subgraph maximizing the minimum degree in a single layer, i.e., $\SLinnermost = \arg\max_{S \subseteq V} f(S)$, where $f(S) = \max_{\ell \in L} \mu(S, \ell)$, while $\ell^{(\mu)} = \arg\max_{\ell \in L} \mu(\SLinnermost, \ell)$.  It is easy to see that $\SLinnermost \in \mathbf{C}$.
Finally, let $\SLdensest$ be the densest subgraph among all single-layer densest subgraphs, i.e., $\SLdensest = \arg\max_{S \subseteq V} g(S)$, where $g(S) = \max_{\ell \in L} \frac{|E_{\ell}[S]|}{|S|}$, and $\ell^*$ be the layer where $\SLdensest$ exhibits its largest density, i.e.,  $\ell^* = \arg\max_{\ell \in L} \frac{|E_{\ell}[\SLdensest]|}{|\SLdensest|}$. We start by introducing the following two lemmas that can straightforwardly be derived from the definitions of  $C^*$, $\SLinnermost$, $\SLdensest$, $\ell^{(\mu)}$, and $\ell^*$:

\begin{mylemma}\label{lemma:SLinnermost}
$\density(C^*) \geq \density(\SLinnermost)$.
\end{mylemma}
\begin{proof}
By definition, $\SLinnermost$ is a multilayer core described by (among others) the 	\corenessvec\  $\vec{k} = [k_\ell]_{\ell \in L}$ with $k_{\ell^{(\mu)}} = \max_{\ell \in L} \mu(\SLinnermost, \ell)$, and $k_{\ell} = 0$, $\forall \ell \neq \ell^{(\mu)}$.
Then  $\SLinnermost \in \mathcal{C}$.
As $C^* = \arg\max_{C \in \mathcal{C}} \density(C)$, it holds that $\density(C^*) \geq \density(\SLinnermost)$.
\end{proof}

\begin{mylemma}\label{lemma:SLdensest}
$\density(S^*) \leq \frac{\textstyle |E_{\ell^*} [\SLdensest  ]  |}{\textstyle |\SLdensest |}  |L|^{\beta}$.
\end{mylemma}
\begin{proof}
\begin{eqnarray*}
	\!\!\!\!\density(S^*) & \!\!=\!\! & \max_{\hat{L} \subseteq L} \min_{\ell \in \hat{L}} \frac{|E_{\ell}[S^*]|}{|S^*|}  |\hat{L}|^{\beta} \ \leq \  \max_{\ell \in L} \frac{|E_{\ell}[S^*]|}{|S^*|}  |L|^{\beta} \ \leq  \ \frac{ |E_{\ell^*} [\SLdensest  ]  |}{ |\SLdensest |}  |L|^{\beta}.
\end{eqnarray*}
\end{proof}

The following further lemma shows a lower bound on the minimum degree of a vertex in $\SLdensest$:

\begin{mylemma}\label{lemma:LBdensestmindeg}
$\mu(\SLdensest, \ell^*) \geq \frac{\textstyle |E_{\ell^*} [\SLdensest ] |}{\textstyle |\SLdensest  |}$.
\end{mylemma}
\begin{proof}
As $\SLdensest$ is the subgraph maximizing the density in layer $\ell^*$, removing the minimum-degree node from $\SLdensest$ cannot increase that density.
Thus, it holds that:
\begin{eqnarray*}
\lefteqn{\frac{|E_{\ell^*}[\SLdensest]|}{|S^*|} \geq \frac{|E_{\ell^*}[\SLdensest]| - \mu(\SLdensest, \ell^*)}{|\SLdensest| - 1}}\\
\qquad & \Leftrightarrow & \mu(\SLdensest, \ell^*) \geq |E_{\ell^*}[\SLdensest]| \frac{|\SLdensest| - 1}{|\SLdensest|} - |E_{\ell^*}[\SLdensest]|\qquad\\
 & \Leftrightarrow & \mu(\SLdensest,\ell^*) \geq \frac{|E_{\ell^*}[\SLdensest]|}{|\SLdensest|}.\qquad\qquad\qquad\qquad\qquad
\end{eqnarray*}
\end{proof}

The approximation factor of the proposed \densestalg\ algorithm is ultimately stated in the next theorem:

\begin{mytheorem}\label{th:densestapproximationfactor}
$\density(C^*) \geq \frac{\textstyle 1}{\textstyle 2|L|^{\beta}}\density(S^*)$.
\end{mytheorem}
\begin{proof}
\begin{eqnarray*}
\!\density(C^*) \!\!\!\!& \geq & \density(\SLinnermost) \eqncomment{\qquad\qquad\qquad\qquad\qquad\qquad\qquad \quad}{Lemma~\ref{lemma:SLinnermost}}\\
& \geq &  \max_{\ell \in L} \frac{|E_{\ell}[\SLinnermost]|}{|\SLinnermost|}  1^\beta = \max_{\ell \in L} \frac{|E_{\ell}[\SLinnermost]|}{|\SLinnermost|} \eqncomment{\qquad}{Equation~(\ref{eq:densestfunction})}\\
& \geq & \frac{1}{2} \max_{\ell \in L} \mu(\SLinnermost,\ell) \eqncomment{\qquad\qquad}{as avg degree $\geq$ min degree}\\
& = & \frac{1}{2} \mu(\SLinnermost, \ell^{(\mu)}) \eqncomment{\qquad\qquad\qquad\qquad}{by definition of \SLinnermost}\\
& \geq & \frac{1}{2} \mu(\SLdensest, \ell^*)\eqncomment{\qquad\quad}{optimality of \SLinnermost\ w.r.t. min degree}\\
& \geq & \frac{1}{2}\frac{|E_{\ell^*}[\SLdensest]|}{|\SLdensest|} \eqncomment{\qquad\qquad\qquad\qquad\qquad\qquad\quad \ }{Lemma~\ref{lemma:LBdensestmindeg}}\\
& \geq & \frac{1}{2|L|^\beta} \density(S^*). \eqncomment{\qquad\qquad\qquad\qquad\qquad\qquad\quad \ }{Lemma~\ref{lemma:SLdensest}}\\
& & \qquad\qquad\qquad\qquad\qquad\qquad\qquad\qquad\qquad\qquad\quad
\end{eqnarray*}
\end{proof}

The following corollary shows that the theoretical approximation guarantee stated in Theorem~\ref{th:densestapproximationfactor}  remains the same even if only the inner-most cores are considered (although, clearly, considering the whole core decomposition may lead to better accuracy in practice).

\begin{mycorollary}\label{cor:densestapproximationfactor1}
Given a multilayer graph $G = (V,E,L)$, let $\mathbf{C}_{\textsc{im}}$ be the set of all inner-most multilayer cores of $G$, and let $C^*_{\textsc{im}} = \arg\max_{C \in \mathbf{C}_{\textsc{im}}}\delta(C)$.
It holds that  $\density(C_{\textsc{im}}^*) \geq \frac{1}{2|L|^{\beta}}\density(S^*)$.
\end{mycorollary}
\begin{proof}
Let $\SLinnermost_{\textsc{im}} \in \mathbf{C}_{\textsc{im}}$ be an inner-most core of $G$ whose coreness vector has a component equal to $\ell^{(\mu)}$.
It is easy to see that the result in Lemma~\ref{lemma:SLinnermost} holds for $C^*_{\textsc{im}}$ and $\SLinnermost_{\textsc{im}}$ too,  i.e., becoming $\density(C^*_{\textsc{im}}) \geq \density(\SLinnermost_{\textsc{im}})$, while the proof of Theorem~\ref{th:densestapproximationfactor} holds as is, by simply replacing $C^*$ with $C^*_{\textsc{im}}$ and $\SLinnermost$ with $\SLinnermost_{\textsc{im}}$.
\end{proof}


Finally, we observe that the result in Theorem~\ref{th:densestapproximationfactor} carries over to the \textsc{Min-Avg Densest Common Subgraph} (DCS-MA) problem studied in~\cite{charikar2018finding,jethava2015finding,reinthal2016finding,semertzidis2016best} as well, as that problem can be reduced to our \mldensestsubgraph\ problem (as shown in Theorem~\ref{th:densesthard}).

\begin{figure}[t!]
\centering
\begin{tabular}{cc}
\includegraphics[width=0.3\columnwidth]{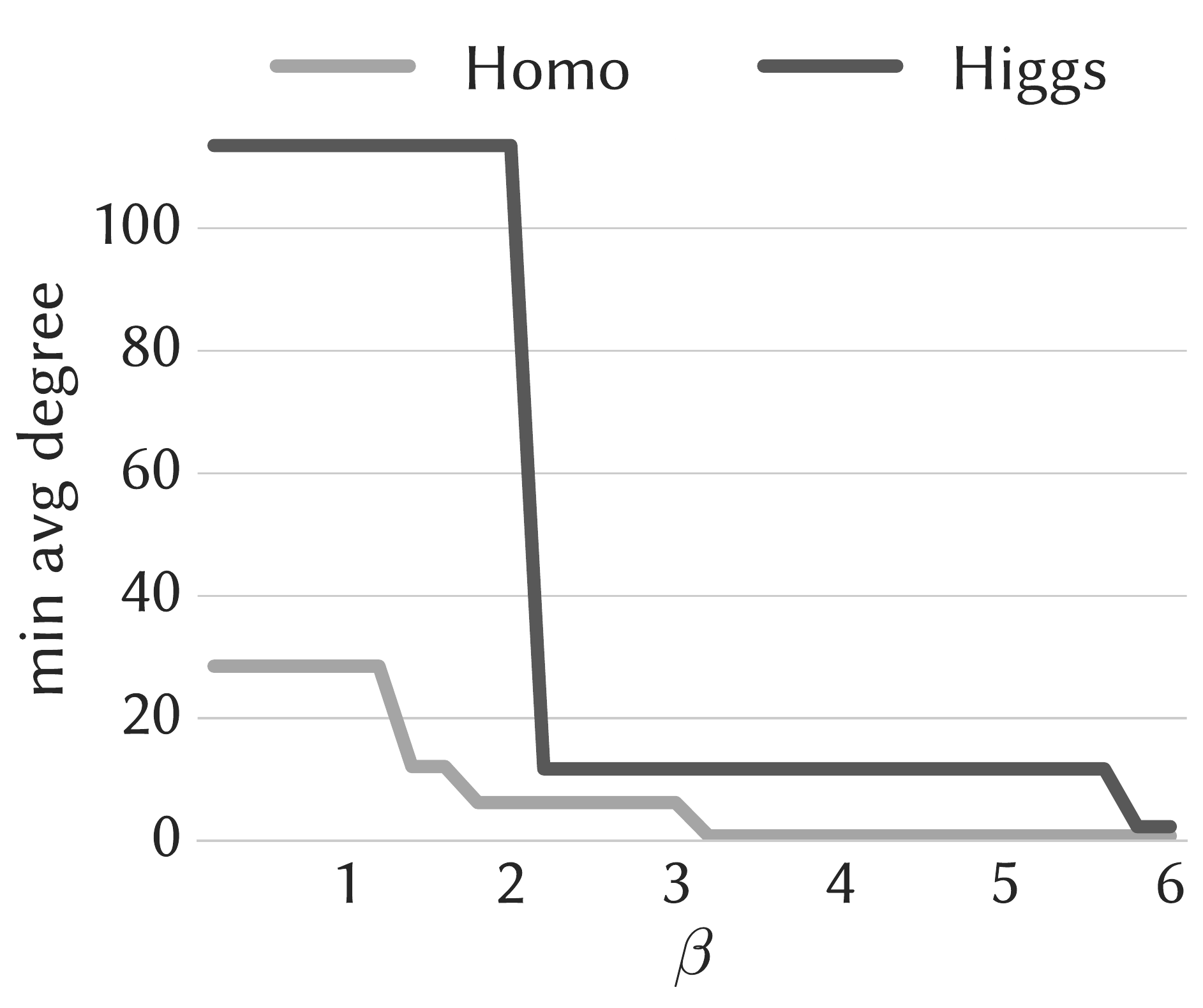} & \includegraphics[width=0.3\columnwidth]{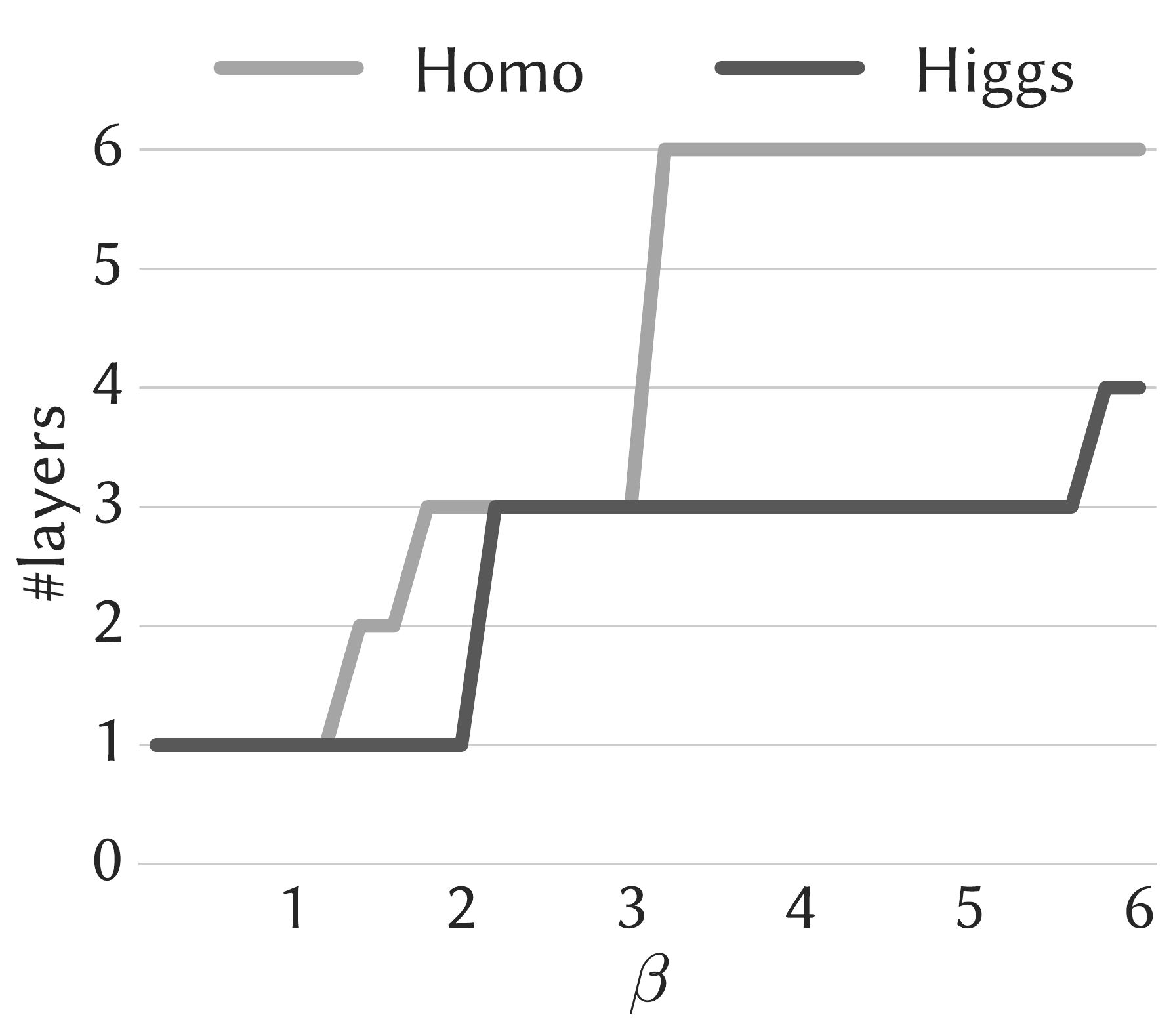}\\
\includegraphics[width=0.3\columnwidth]{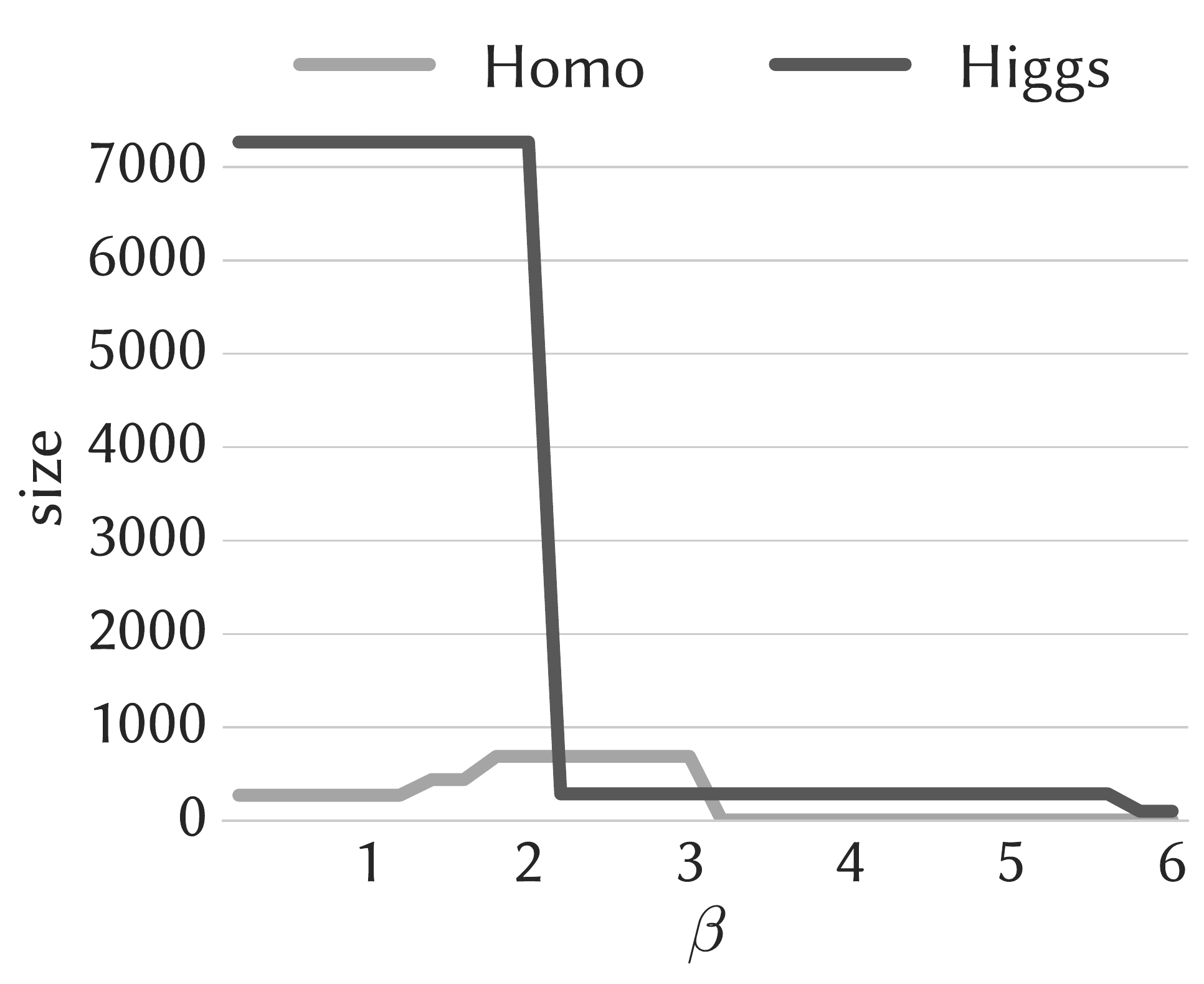} & \includegraphics[width=0.3\columnwidth]{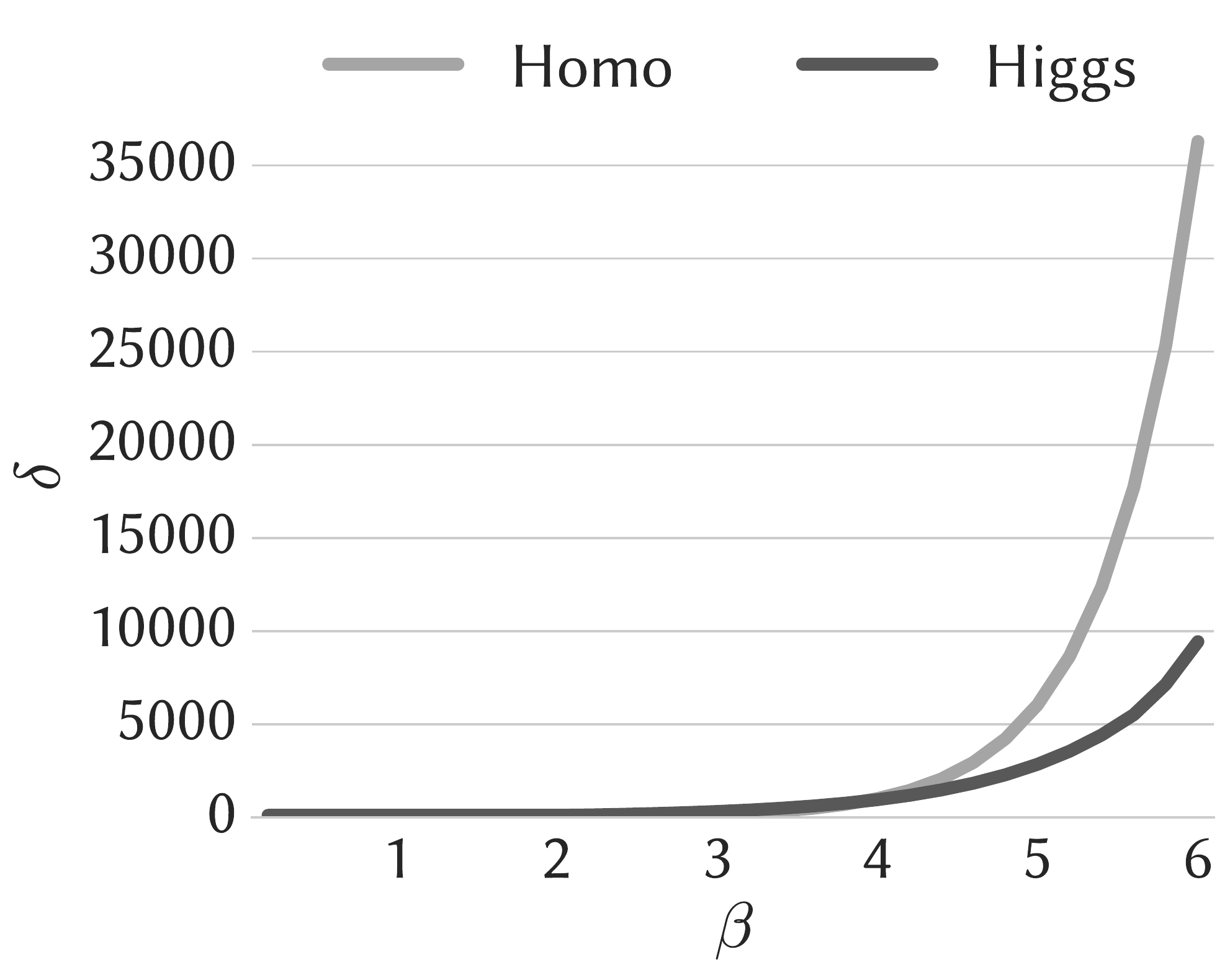}\\
\end{tabular}
\caption{\label{fig:densest}   Multilayer densest-subgraph extraction (\textsf{Homo} and \textsf{Higgs} datasets): minimum average-degree density in a layer, number of selected layers, size, and objective-function value $\density$ of the output densest subgraphs with varying $\beta$.}
\end{figure}

\subsection{Experimental results}
We experimentally evaluate our \densestalg\ algorithm (Algorithm~\ref{alg:densest}) on the datasets in Table~\ref{tab:datasets}.
Figure~\ref{fig:densest} reports the results  -- minimum average-degree density in a layer, number of selected layers, size,  objective-function value $\delta$ --
on the \textsf{Homo} and \textsf{Higgs} datasets, with varying $\beta$.
The remaining datasets, which we omit due to space constraints, exhibit similar trends on all measures.

The trends observed in the figure conform to what expected: the smaller $\beta$, the more the objective function privileges solutions with large average-degree density in a few layers (or even just one layer, for $\beta$ close to zero).
The situation is overturned with larger values of $\beta$, where the minimum average-degree density drops significantly, while the number of selected layers stands at $6$ for \textsf{Homo} and $4$ for \textsf{Higgs}.
In-between $\beta$ values lead to a balancing of the two terms of the objective function, thus giving more interesting solutions.
Also, by definition, $\density$ as a function of $\beta$ draws exponential curves.


Finally, as anecdotal evidence of the output of Algorithm~\ref{alg:densest}, in Figure~\ref{fig:dblpcase} we report the densest subgraph extracted from \textsf{DBLP}.
The subgraph contains 10 vertices and $5$ layers automatically selected by the objective function $\density$.
The minimum average-degree density is encountered on the layers corresponding to topics ``\emph{graph}'' and ``\emph{algorithm}'' (green and yellow layers in the figure), and is equal to $1.2$. The objective-function value is $\density = 41.39$.
Note that the subgraph is composed of two connected components.
In fact, like the single-layer case, multilayer cores are not necessarily connected.

\begin{figure}[t!]
\centering
\includegraphics[width=0.6\columnwidth]{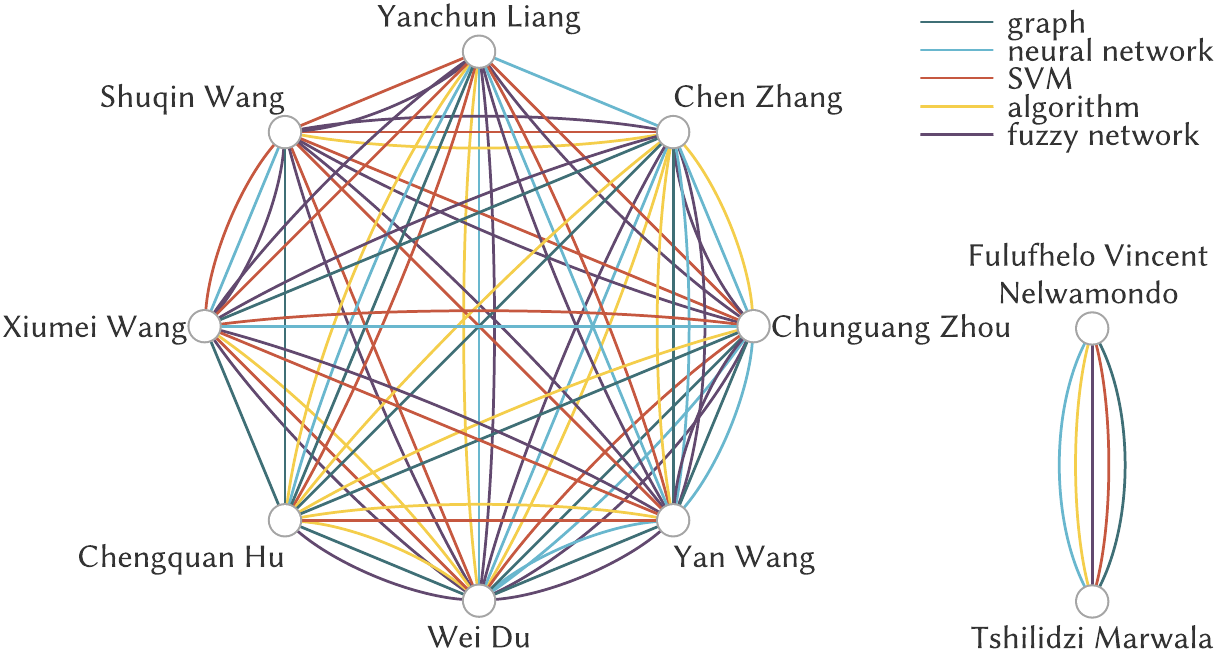}
\caption{\label{fig:dblpcase}   Multilayer densest subgraph extracted by Algorithm~\ref{alg:densest} from the \textsf{DBLP} dataset ($\beta = 2.2$).}
\end{figure}

\section{Multilayer quasi-cliques}
\label{sec:quasicliques}
Another interesting insight into the notion of multilayer cores is about their relationship with (quasi-)cliques.
In single-layer graphs it is well-known that cores can be exploited to speed-up the problem of finding cliques, as a clique of size $k$ is guaranteed to be contained into the $(k-1)$-core.
Interestingly, a similar relationship holds in the multilayer context too.
Given a multilayer graph $G = (V,E,L)$, a layer $\ell \in L$, and a real number $\gamma \in (0,1]$, a subgraph $G[S] = (S \subseteq V, E[S], L)$ of $G$ is said to be a $\gamma$\emph{-quasi-clique} in layer $\ell$ if all its vertices have at least $\gamma (|S| -1)$ neighbors in layer $\ell$ within $S$, i.e., $\forall u \in S : deg_S(u, \ell) \geq \gamma (|S| -1)$.
Jiang~\emph{et~al.}~\cite{jiang2009mining} study the problem of extracting \emph{frequent cross-graph quasi-cliques}:\footnote{The input in~\cite{jiang2009mining} has the form of a set of graphs sharing the same vertex set, which is clearly fully equivalent to the notion of multilayer graph considered in this work.} given a multilayer graph $G = (V,E,L)$,
a function $\Gamma : L \rightarrow (0,1]$ assigning a real value to every layer in $L$, a real number $\mbox{\emph{min\_sup}} \in (0,1]$, and an integer $\mbox{\emph{min\_size}} > 1$, find all maximal subgraphs $G[S]$ of $G$ of size larger than  \emph{min\_size} such that there exist at least $\mbox{\emph{min\_sup}}\times|L|$ layers $\ell$ for which $G[S]$ is a $\Gamma(\ell)$-quasi-clique.

The following theorem shows that a frequent cross-graph quasi-clique of size~$\geq\!\mbox{\emph{min\_size}}$  is necessarily contained into a $\vec{k}$-core described by a coreness vector $\vec{k}=[k_{\ell}]_{\ell \in L}$ such that there exists a fraction of \emph{min\_sup} layers $\ell$ where $k_{\ell} = \lceil \Gamma(\ell)(\mbox{\emph{min\_size}} - 1) \rceil$.
\begin{mytheorem}\label{th:cliques}
Given a multilayer graph  $G = (V,E,L)$, a real-valued function $\Gamma : L \rightarrow (0,1]$, a real number $\mbox{\emph{min\_sup}} \in (0,1]$, and an integer $\mbox{\emph{min\_size}} > 1$, a frequent cross-graph quasi-clique of $G$ complying with parameters $\Gamma$,  \emph{min\_sup}, and \emph{min\_size} is contained into a $\vec{k}$-core with coreness vector $\vec{k}=[k_{\ell}]_{\ell \in L}$ such that $|\{\ell \in L : k_{\ell} = \lceil \Gamma(\ell)(\mbox{\emph{min\_size}} - 1)\rceil \}| = \lceil \mbox{\emph{min\_sup}} \times |L| \rceil$.
\end{mytheorem}
\begin{proof}
Assume that  a cross-graph quasi-clique $S$ of $G$ complying with parameters $\Gamma$,  \emph{min\_sup}, and \emph{min\_size}
is not contained into any $\vec{k}$-core with coreness vector $\vec{k}=[k_{\ell}]_{\ell \in L}$ such that $|\{\ell \in L : k_{\ell} = \lceil \Gamma(\ell)(\mbox{\emph{min\_size}} - 1)\rceil \}| = \lceil \mbox{\emph{min\_sup}} \times |L| \rceil$.
This means that $S$ contains a vertex $u$ such that
$|\{\ell \in L : deg_S(u, \ell) \geq \Gamma(\ell)(\mbox{\emph{min\_size}} -1)\}| < \mbox{\emph{min\_sup}} \times |L|$, which means that 
$|\{\ell \in L : deg_S(u, \ell) \geq \Gamma(\ell)(|S|-1)\}| < \mbox{\emph{min\_sup}} \times |L|$ as well, since $|S| \geq \mbox{\emph{min\_size}}$.
This violates the definition of frequent cross-graph quasi-clique.
\end{proof}

As a simple corollary, the computation of frequent cross-graph quasi-cliques can therefore be circumstantiated to the subgraph given by the union of all multilayer cores complying with the condition stated in Theorem~\ref{th:cliques}.

\begin{mycorollary}\label{cor:cliques}
Given a multilayer graph  $G = (V,E,L)$, a real-valued function $\Gamma : L \rightarrow (0,1]$, a real number $\mbox{\emph{min\_sup}} \in (0,1]$, and an integer $\mbox{\emph{min\_size}} > 1$, let $G' = (V',E',L)$ the subgraph of $G$ given by the union of all multilayer cores of $G$ complying with Theorem~\ref{th:cliques}.
It holds that all cross-graph quasi-cliques of $G$ complying with parameters $\Gamma$,  \emph{min\_sup}, and \emph{min\_size} are contained into $G'$.
\end{mycorollary}

The finding in Corollary~\ref{cor:cliques} can profitably be exploited to have a more efficient extraction of frequent cross-graph quasi-cliques.
Specifically, the idea is to 
($i$) compute \emph{all}  multilayer cores of the input graph $G$ (including the non-distinct ones, as the condition stated in Theorem~\ref{th:cliques} refers to not necessarily maximal coreness vectors); 
($ii$) process all multilayer cores of $G$ one by one, retain only the ones complying with Theorem~\ref{th:cliques}, and compute the subgraph $G'$ induced by the union of all such cores;
($iii$) run any algorithm for frequent cross-graph quasi-cliques on $G'$.
Based on the above theoretical results, such a procedure is guaranteed to be sound and complete, and it is expected to provide a significant speed-up, as $G'$ is expected to be much smaller than the original graph $G$. 


\begin{table}
\centering
\caption{Comparison of the runtime of the efficient extraction of frequent cross-graph quasi-cliques by Corollary~\ref{cor:cliques} and of the original algorithm~\cite{jiang2009mining}, for the \textsf{SacchCere} dataset.
The evaluation is proposed varying one of the parameters, i.e., $\Gamma$, $\mbox{\emph{min\_sup}}$, and $\mbox{\emph{min\_size}}$, at a time.
The number of solution quasi-cliques and the number of vertices $|V'|$ of the subgraph $G'$ are also reported. For each run of the experiment, the smallest runtime is bolded.
\label{tab:fcgqc_sacchcere}}

\centerline{
\setlength{\tabcolsep}{2.25pt}
\begin{tabular}{p{6pt}p{6pt}p{6pt}p{6pt}p{6pt}p{6pt}p{8pt}|c|c|c|c|c|c}
\multicolumn{7}{c}{} & \multicolumn{1}{c}{} & \multicolumn{1}{c}{} & \multicolumn{1}{c}{\# solution} & \multicolumn{1}{c}{} & \multicolumn{2}{c}{runtime (s)}\\\cline{12-13}
\multicolumn{7}{c}{$\Gamma$} & \multicolumn{1}{c}{$\mbox{\emph{min\_sup}}$} & \multicolumn{1}{c}{$\mbox{\emph{min\_size}}$} & \multicolumn{1}{c}{quasi-cliques} & \multicolumn{1}{c|}{$|V'|$} & \multicolumn{1}{c}{Corollary~\ref{cor:cliques}} & \multicolumn{1}{c}{\cite{jiang2009mining}}\\
\cline{1-13}
$1$ & $1$ & $1$ & $1$ & $.2$ & $.2$ & $1$ & $0.5$ & $6$ & $2$ & $371$ & $\mathbf{3}$ & $169$\\\cline{1-7}\cline{10-13}
$.9$ & $.9$ & $.9$ & $.9$ & $.2$ & $.2$ & $.9$ &  \multicolumn{1}{c|}{} & \multicolumn{1}{c|}{} & $2$ & $371$ & $\mathbf{25}$ & $17\,561$\\\cline{1-7}\cline{10-13}
$.8$ & $.8$ & $.8$ & $.8$ & $.2$ & $.2$ & $.8$ & \multicolumn{1}{c|}{} & \multicolumn{1}{c|}{} & $6$ & $1\,196$ & $\mathbf{734}$ & $22\,932$\\\cline{1-7}\cline{10-13}
$.7$ & $.7$ & $.7$ & $.7$ & $.2$ & $.2$ & $.7$ & \multicolumn{1}{c|}{} & \multicolumn{1}{c|}{} & $6$ & $1\,196$ & $\mathbf{728}$ & $23\,376$\\\cline{1-7}\cline{10-13}
$.6$ & $.6$ & $.6$ & $.6$ & $.2$ & $.2$ & $.6$ & \multicolumn{1}{c|}{} & \multicolumn{1}{c|}{} & $59$ & $2\,300$ & $\mathbf{5\,200}$ & $28\,948$\\\cline{1-7}\cline{10-13}
$.5$ & $.5$ & $.5$ & $.5$ & $.2$ & $.2$ & $.5$ & \multicolumn{1}{c|}{} & \multicolumn{1}{c|}{} & $59$ & $2\,300$ & $\mathbf{5\,123}$ & $29\,677$\\
\cline{1-13}
\end{tabular}
}
\vspace{0.5cm}

\centerline{
\setlength{\tabcolsep}{2.25pt}
\begin{tabular}{p{6pt}p{6pt}p{6pt}p{6pt}p{6pt}p{6pt}p{8pt}|c|c|c|c|c|c}
\multicolumn{7}{c}{} & \multicolumn{1}{c}{} & \multicolumn{1}{c}{} & \multicolumn{1}{c}{\# solution} & \multicolumn{1}{c}{} & \multicolumn{2}{c}{runtime (s)}\\\cline{12-13}
\multicolumn{7}{c}{$\Gamma$} & \multicolumn{1}{c}{$\mbox{\emph{min\_sup}}$} & \multicolumn{1}{c}{$\mbox{\emph{min\_size}}$} & \multicolumn{1}{c}{quasi-cliques} & \multicolumn{1}{c|}{$|V'|$} & \multicolumn{1}{c}{Corollary~\ref{cor:cliques}} & \multicolumn{1}{c}{\cite{jiang2009mining}}\\
\cline{1-13}
$.5$ & $.5$ & $.5$ & $.5$ & $.2$ & $.2$ & $.5$ & $1$ & $3$ & $2$ & $152$ & $\mathbf{2}$ & $281$\\\cline{8-8}\cline{10-13}
\multicolumn{7}{c|}{} & $0.9$ & \multicolumn{1}{c|}{} & $2$ & $152$ & $\mathbf{2}$ & $282$\\\cline{8-8}\cline{10-13}
\multicolumn{7}{c|}{} & $0.8$ & \multicolumn{1}{c|}{} & $28$ & $940$ & $\mathbf{23}$ & $292$\\\cline{8-8}\cline{10-13}
\multicolumn{7}{c|}{} & $0.7$ & \multicolumn{1}{c|}{} & $323$ & $3\,271$ & $\mathbf{205}$ & $411$\\\cline{8-8}\cline{10-13}
\multicolumn{7}{c|}{} & $0.6$ & \multicolumn{1}{c|}{} & $323$ & $3\,271$ & $\mathbf{203}$ & $414$\\\cline{8-8}\cline{10-13}
\multicolumn{7}{c|}{} & $0.5$ & \multicolumn{1}{c|}{} & $1\,630$ & $4\,581$ & $\mathbf{2\,569}$ & $3\,075$\\
\cline{1-13}
\end{tabular}
}
\vspace{0.5cm}

\centerline{
\setlength{\tabcolsep}{2.25pt}
\begin{tabular}{p{6pt}p{6pt}p{6pt}p{6pt}p{6pt}p{6pt}p{8pt}|c|c|c|c|c|c}
\multicolumn{7}{c}{} & \multicolumn{1}{c}{} & \multicolumn{1}{c}{} & \multicolumn{1}{c}{\# solution} & \multicolumn{1}{c}{} & \multicolumn{2}{c}{runtime (s)}\\\cline{12-13}
\multicolumn{7}{c}{$\Gamma$} & \multicolumn{1}{c}{$\mbox{\emph{min\_sup}}$} & \multicolumn{1}{c}{$\mbox{\emph{min\_size}}$} & \multicolumn{1}{c}{quasi-cliques} & \multicolumn{1}{c|}{$|V'|$} & \multicolumn{1}{c}{Corollary~\ref{cor:cliques}} & \multicolumn{1}{c}{\cite{jiang2009mining}}\\
\cline{1-13}
$.5$ & $.5$ & $.5$ & $.5$ & $.2$ & $.2$ & $.5$ & $0.5$ & $7$ & $27$ & $2\,254$ & $\mathbf{5\,606}$ & $34\,904$\\\cline{9-13}
\multicolumn{7}{c|}{} & \multicolumn{1}{c|}{} & $6$ & $59$ & $2\,300$ & $\mathbf{5\,123}$ & $29\,677$\\\cline{9-13}
\multicolumn{7}{c|}{} & \multicolumn{1}{c|}{} & $5$ & $357$ & $3\,363$ & $\mathbf{4\,493}$ & $21\,206$\\\cline{9-13}
\multicolumn{7}{c|}{} & \multicolumn{1}{c|}{} & $4$ & $378$ & $3\,363$ & $\mathbf{3\,704}$ & $15\,465$\\\cline{9-13}
\multicolumn{7}{c|}{} & \multicolumn{1}{c|}{} & $3$ & $1\,630$ & $4\,581$ & $\mathbf{2\,569}$ & $3\,075$\\
\cline{1-13}
\end{tabular}
}

\end{table}

\begin{table}
\centering
\caption{Comparison of the runtime of the efficient extraction of frequent cross-graph quasi-cliques by Corollary~\ref{cor:cliques} and of the original algorithm~\cite{jiang2009mining}, for the \textsf{DBLP} dataset.
The evaluation is proposed varying one of the parameters, i.e., $\Gamma$, $\mbox{\emph{min\_sup}}$, and $\mbox{\emph{min\_size}}$, at a time.
The number of solution quasi-cliques and the number of vertices $|V'|$ of the subgraph $G'$ are also reported.
$++$ indicates runtime longer than $259\,200$ seconds (i.e., $3$ days).
For each run of the experiment, the smallest runtime is bolded.
\label{tab:fcgqc_dblp}}

\centerline{
\setlength{\tabcolsep}{2.25pt}
\begin{tabular}{p{6pt}p{6pt}p{6pt}p{6pt}p{6pt}p{6pt}p{6pt}p{6pt}p{6pt}p{8pt}|c|c|c|c|c|c}
\multicolumn{10}{c}{} & \multicolumn{1}{c}{} & \multicolumn{1}{c}{} & \multicolumn{1}{c}{\# solution} & \multicolumn{1}{c}{} & \multicolumn{2}{c}{runtime (s)}\\\cline{15-16}
\multicolumn{10}{c}{$\Gamma$} & \multicolumn{1}{c}{$\mbox{\emph{min\_sup}}$} & \multicolumn{1}{c}{$\mbox{\emph{min\_size}}$} & \multicolumn{1}{c}{quasi-cliques} & \multicolumn{1}{c|}{$|V'|$} & \multicolumn{1}{c}{Corollary~\ref{cor:cliques}} & \multicolumn{1}{c}{\cite{jiang2009mining}}\\
\cline{1-16}
$1$ & $1$ & $1$ & $1$ & $1$ & $1$ & $1$ & $1$ & $1$ & $1$ & $0.2$ & $8$ & $2$ & $18$ & $\mathbf{0.2}$ & $26\,496$\\\cline{1-10}\cline{13-16}
$.9$ & $.9$ & $.9$ & $.9$ & $.9$ & $.9$ & $.9$ & $.9$ & $.9$ & $.9$ & \multicolumn{1}{c|}{} & \multicolumn{1}{c|}{} & $2$ & $18$ & $\mathbf{0.2}$ & $26\,112$\\\cline{1-10}\cline{13-16}
$.8$ & $.8$ & $.8$ & $.8$ & $.8$ & $.8$ & $.8$ & $.8$ & $.8$ & $.8$ & \multicolumn{1}{c|}{} & \multicolumn{1}{c|}{} & $13$ & $75$ & $\mathbf{0.3}$ & $26\,867$\\\cline{1-10}\cline{13-16}
$.7$ & $.7$ & $.7$ & $.7$ & $.7$ & $.7$ & $.7$ & $.7$ & $.7$ & $.7$ & \multicolumn{1}{c|}{} & \multicolumn{1}{c|}{} & $18$ & $196$ & $\mathbf{1}$ & $27\,387$\\\cline{1-10}\cline{13-16}
$.6$ & $.6$ & $.6$ & $.6$ & $.6$ & $.6$ & $.6$ & $.6$ & $.6$ & $.6$ & \multicolumn{1}{c|}{} & \multicolumn{1}{c|}{} & $18$ & $196$ & $\mathbf{1}$ & $27\,084$\\\cline{1-10}\cline{13-16}
$.5$ & $.5$ & $.5$ & $.5$ & $.5$ & $.5$ & $.5$ & $.5$ & $.5$ & $.5$ & \multicolumn{1}{c|}{} & \multicolumn{1}{c|}{} & $121$ & $801$ & $\mathbf{18}$ & $31\,508$\\
\cline{1-16}
\end{tabular}
}
\vspace{1.5cm}

\centerline{
\setlength{\tabcolsep}{2.25pt}
\begin{tabular}{p{6pt}p{6pt}p{6pt}p{6pt}p{6pt}p{6pt}p{6pt}p{6pt}p{6pt}p{8pt}|c|c|c|c|c|c}
\multicolumn{10}{c}{} & \multicolumn{1}{c}{} & \multicolumn{1}{c}{} & \multicolumn{1}{c}{\# solution} & \multicolumn{1}{c}{} & \multicolumn{2}{c}{runtime (s)}\\\cline{15-16}
\multicolumn{10}{c}{$\Gamma$} & \multicolumn{1}{c}{$\mbox{\emph{min\_sup}}$} & \multicolumn{1}{c}{$\mbox{\emph{min\_size}}$} & \multicolumn{1}{c}{quasi-cliques} & \multicolumn{1}{c|}{$|V'|$} & \multicolumn{1}{c}{Corollary~\ref{cor:cliques}} & \multicolumn{1}{c}{\cite{jiang2009mining}}\\
\cline{1-16}
$.5$ & $.5$ & $.5$ & $.5$ & $.5$ & $.5$ & $.5$ & $.5$ & $.5$ & $.5$ & $0.5$ & $3$ & $8$ & $182$ & $\mathbf{0.2}$ & $26\,969$\\\cline{11-11}\cline{13-16}
\multicolumn{10}{c|}{} & $0.4$ & \multicolumn{1}{c|}{} & $195$ & $2\,375$ & $\mathbf{1}$ & $26\,964$\\\cline{11-11}\cline{13-16}
\multicolumn{10}{c|}{} & $0.3$ & \multicolumn{1}{c|}{} & $3\,394$ & $22\,659$ & $\mathbf{210}$ & $32\,981$\\\cline{11-11}\cline{13-16}
\cline{1-16}
\end{tabular}
}
\vspace{1.5cm}

\centerline{
\setlength{\tabcolsep}{2.25pt}
\begin{tabular}{p{6pt}p{6pt}p{6pt}p{6pt}p{6pt}p{6pt}p{6pt}p{6pt}p{6pt}p{8pt}|c|c|c|c|c|c}
\multicolumn{10}{c}{} & \multicolumn{1}{c}{} & \multicolumn{1}{c}{} & \multicolumn{1}{c}{\# solution} & \multicolumn{1}{c}{} & \multicolumn{2}{c}{runtime (s)}\\\cline{15-16}
\multicolumn{10}{c}{$\Gamma$} & \multicolumn{1}{c}{$\mbox{\emph{min\_sup}}$} & \multicolumn{1}{c}{$\mbox{\emph{min\_size}}$} & \multicolumn{1}{c}{quasi-cliques} & \multicolumn{1}{c|}{$|V'|$} & \multicolumn{1}{c}{Corollary~\ref{cor:cliques}} & \multicolumn{1}{c}{\cite{jiang2009mining}}\\
\cline{1-16}
$.5$ & $.5$ & $.5$ & $.5$ & $.5$ & $.5$ & $.5$ & $.5$ & $.5$ & $.5$ & $0.2$ & $13$ & $1$ & $75$ & $\mathbf{0.2}$ & $26\,644$\\\cline{12-16}
\multicolumn{10}{c|}{} & \multicolumn{1}{c|}{} & $12$ & $1$ & $75$ & $\mathbf{0.2}$ & $27\,136$\\\cline{12-16}
\multicolumn{10}{c|}{} & \multicolumn{1}{c|}{} & $11$ & $8$ & $196$ & $\mathbf{0.7}$ & $26\,966$\\\cline{12-16}
\multicolumn{10}{c|}{} & \multicolumn{1}{c|}{} & $10$ & $10$ & $196$ & $\mathbf{0.7}$ & $27\,116$\\\cline{12-16}
\multicolumn{10}{c|}{} & \multicolumn{1}{c|}{} & $9$ & $116$ & $801$ & $\mathbf{18}$ & $32\,372$\\\cline{12-16}
\multicolumn{10}{c|}{} & \multicolumn{1}{c|}{} & $8$ & $121$ & $801$ & $\mathbf{18}$ & $31\,508$\\\cline{12-16}
\multicolumn{10}{c|}{} & \multicolumn{1}{c|}{} & $7$ & $1\,292$ & $3\,468$ & $\mathbf{181}$ & $113\,558$\\\cline{12-16}
\multicolumn{10}{c|}{} & \multicolumn{1}{c|}{} & $6$ & $1\,370$ & $3\,468$ & $\mathbf{198}$ & $113\,520$\\\cline{12-16}
\multicolumn{10}{c|}{} & \multicolumn{1}{c|}{} & $5$ & $7\,599$ & $15\,316$ & $\mathbf{3\,790}$ & $++$\\\cline{12-16}
\multicolumn{10}{c|}{} & \multicolumn{1}{c|}{} & $4$ & $8\,578$ & $15\,316$ & $\mathbf{3\,502}$ & $++$\\\cline{12-16}
\cline{1-16}
\end{tabular}
}

\end{table}

\subsection{Experimental results}
We show in Tables~\ref{tab:fcgqc_sacchcere}~and~\ref{tab:fcgqc_dblp} the experimental results about the comparison of the algorithm proposed by Jiang~\emph{et~al.}~\cite{jiang2009mining} and the more efficient extraction of frequent cross-graph quasi-cliques by Corollary~\ref{cor:cliques}.
Table~\ref{tab:fcgqc_sacchcere} refers to the \textsf{SacchCere} dataset, while Table~\ref{tab:fcgqc_dblp} to the \textsf{DBLP} dataset.
To evaluate the effect of the parameters, i.e., the function $\Gamma$, $\mbox{\emph{min\_sup}}$, and $\mbox{\emph{min\_size}}$, on the performance of the two approaches, we vary a parameter at a time keeping the other two fixed.
With regards to the values selected for $\Gamma$, we fix $\Gamma(\ell_5) = \Gamma(\ell_6) = 0.2$ in all the experiments involving the \textsf{SacchCere} dataset, due to the imbalance of the distribution of the edges in favor of the other five layers (i.e., layers $\ell_1,\ldots,\ell_4, \ell_7$).
Instead, given the uniformity of the edge density across the layers of the \textsf{DBLP} dataset, $\Gamma$ is modified coherently for all the layer in this latter case.
In addition to the execution times, for each configuration of the parameters, we also report the number of solution frequent cross-graph quasi-cliques and the number of vertices $|V'|$ of the subgraph $G'$ identified by Corollary~\ref{cor:cliques}.

The first thing to notice is that, in both datasets and for every configuration, our approach is faster than the algorithm by Jiang~\emph{et~al.}~\cite{jiang2009mining}.
The actual speed-up varies with the size of $|V'|$ (with respect to $|V|$) which, in turn, is affected by the mining parameters.
For the \textsf{SacchCere} dataset, we obtain the most extreme cases when varying $\mbox{\emph{min\_sup}}$ (middle table): our approach is able to prune from 30\% ($\mbox{\emph{min\_sup}} = 0.5$) up to 98\% ($\mbox{\emph{min\_sup}} = 1$) of the input multilayer graph.
For the \textsf{DBLP} dataset, the results are even stronger: in the worst case (i.e., $\Gamma(\ell) = 0.5 \ \forall \ell \in L$, $\mbox{\emph{min\_sup}} = 0.3$, and $\mbox{\emph{min\_size}} = 3$)  we prune the 95\% of the original vertex set.
The runtime of both our approach and Jiang~\emph{et~al.}'s~\cite{jiang2009mining} algorithm varies consistently according to parameters and to $|V'|$.
The speed-up that our method reaches ranges from $1.2$ to two orders of magnitude for the \textsf{SacchCere} dataset, and from one order up to six orders of magnitude for the \textsf{DBLP} dataset.

\section{Community search in multilayer networks}
\label{sec:communitysearch}
The idea here is very similar to that of the multilayer densest subgraph.

\setcounter{problem}{4}
\begin{problem}[Multilayer Community Search]
Given a multilayer graph $G=(V,E,L)$, a set of vertices $S \subseteq V$, and a set of layers $\hat{L} \subseteq L$, we define the minimum degree of a vertex in $S$, within the subgraph induced by $S$ and $\hat{L}$ as:
$$
\varphi(S,\hat{L}) = \min_{\ell \in \hat{L}} \min_{u \in S} deg_S(u,\ell).
$$

Given a positive real number $\beta$, we define a real-valued density function $\vartheta : 2^{V} \rightarrow \mathbb{R}^+$ as:
$$
\vartheta(S) = \max_{\hat{L} \subseteq L}  \varphi(S,\hat{L})  |\hat{L}|^{\beta}.
$$

Given a set $V_Q \subseteq V$ of query vertices, find a subgraph containing all the query vertices and maximizing the density function, i.e.,
\begin{equation}\label{eq:cssolution}
S^* = \argmax_{V_Q \subseteq S \subseteq V}  \vartheta(S).
\end{equation}
\end{problem}

Let $\coresset$ be the set of all non-empty multilayer cores of $G$.
For a core $C \in \coresset$ with \corenessvec\ $\vec{k} = [k_{\ell}]_{\ell \in L}$, we define the score
$$
\sigma(C) = \max_{\hat{L} \subseteq L} (\min_{\ell \in \hat{L}} k_{\ell}) |\hat{L}|^{\beta},
$$
and denote by $C^*$ a core that contains all query vertices in $V_Q$ and maximizes the score $\sigma$, i.e.,
\begin{equation}\label{eq:cscore}
C^* = \argmax_{C \in \mathbf{C},  V_Q \subseteq C} \sigma(C).
\end{equation}

As shown in the following theorem, $C^*$ is a (not necessarily unique) {\em exact} solution to Problem~\ref{prob:mlcs}.

\begin{mytheorem}
Given a multilayer graph $G = (V,E,L)$, and a set $V_Q \subseteq V$ of query vertices, let $S^*$ and $C^*$ be the vertex sets defined as in Equation~(\ref{eq:cssolution}) and Equation~(\ref{eq:cscore}), respectively.
It holds that $\vartheta(C^*) = \vartheta(S^*)$.
\end{mytheorem}
\begin{proof}
We prove the statement by contradiction, assuming that $\vartheta(C^*) < \vartheta(S^*)$.
Let $\mu_\ell = \min_{u \in S^*} deg_{S^*} (u,\ell)$, and $\vec{\mu} = [\mu_{\ell}]_{\ell \in L}$.
By definition of multilayer core, there exists a core $C \in \coresset$ of $G$ with coreness vector $\vec{\mu}$ such that $S^* \subseteq C$.
This means that 
$$
\sigma(C) =  \max_{\hat{L} \subseteq L} (\min_{\ell \in \hat{L}} \mu_{\ell}) |\hat{L}|^{\beta} = \max_{\hat{L} \subseteq L} (\min_{\ell \in \hat{L}} \min_{u \in S^*} deg_{S^*} (u,\ell)) |\hat{L}|^{\beta} = \vartheta(S^*).
$$
Thus, there exists a core $C \in \coresset$ whose $\vartheta(\cdot)$ score is equal to $\vartheta(S^*)$, which contradicts the original assumption $\vartheta(C^*) < \vartheta(S^*)$.
\end{proof}

\spara{Algorithms.}
The core $C^*$ can be straightforwardly found by running any of the proposed algorithms for multilayer core decomposition -- \bfs\ (Algorithm~\ref{alg:bfs}), \dfs\ (Algorithm~\ref{alg:dfs}), or \hybrid\ (Algorithm~\ref{alg:hybrid}) -- and taking from the overall output core set the core maximizing the $\sigma(\cdot)$ score.
However, thanks to the constraint about containment of query vertices $V_Q$, the various algorithms can be speeded up by preventively skipping the computation of cores that do not contain $V_Q$.
Specifically, this corresponds to the following simple modifications:
\begin{itemize}
\item
\bfs\ (Algorithm~\ref{alg:bfs}): replace the condition at Line~\ref{line:bfs:exists} with ``{\bf if} $V_Q \subseteq C_{\vec{k}}$ {\bf then}''.
\item
\dfs\ (Algorithm~\ref{alg:dfs}): stop the \kcorespathalg subroutine used at Lines~\ref{line:dfs:queue}~and~\ref{line:dfs:cores} as soon as a core not containing $V_Q$ is encountered and make the subroutine return only the cores containing $V_Q$.
\item
\hybrid\ (Algorithm~\ref{alg:hybrid}): replace the condition at Line~\ref{line:hybrid:exists} with ``{\bf if} $V_Q \subseteq C_{\vec{k}}$ {\bf then}''.
\end{itemize}

\begin{table}[t!]
\centering
\caption{Comparison of the average runtime (in seconds) between the original algorithms for multilayer core decomposition and modified methods for community search, with varying the number $|V_Q|$ of query vertices.
In each dataset and for each $|V_Q|$, the smallest runtime is bolded.\label{tab:cs_results}}

\centerline{
\setlength{\tabcolsep}{3.75pt}
\begin{tabular}{c|c|c|cccccccccc}
\multicolumn{1}{c}{} & \multicolumn{2}{c}{} & \multicolumn{10}{c}{$|V_Q|$}\\\cline{4-13}
\multicolumn{1}{c}{dataset} & \multicolumn{1}{c}{method} & \multicolumn{1}{c}{original} & $1$ & $2$ & $3$ & $4$ & $5$ & $6$ & $7$ & $8$ & $9$ & $10$\\
\hline
\textsf{Homo} & \textsc{bfs} & $13$ & $2$ & $1$ & $0.7$ & $0.7$ & $0.6$ & $0.6$ & $0.6$ & $0.6$ & $0.6$ & $0.6$\\
& \textsc{dfs} & $27$ & $3$ & $2$ & $1$ & $1$ & $1$ & $0.9$ & $0.9$ & $0.9$ & $0.9$ & $0.9$\\
& \textsc{h} & $12$ & $\mathbf{0.9}$ & $\mathbf{0.3}$ & $\mathbf{0.1}$ & $\mathbf{0.1}$ & $\mathbf{0.1}$ & $\mathbf{0.1}$ & $\mathbf{0.1}$ & $\mathbf{0.1}$ & $\mathbf{0.1}$ & $\mathbf{0.1}$\\
\hline
\textsf{SacchCere} & \textsc{bfs} & $1\,134$ & $\mathbf{162}$ & $\mathbf{25}$ & $6$ & $3$ & $1$ & $1$ & $0.7$ & $0.7$ & $0.5$ & $0.5$\\
& \textsc{dfs} & $2\,627$ & $390$ & $58$ & $13$ & $6$ & $2$ & $2$ & $1$ & $1$ & $0.7$ & $0.6$\\
& \textsc{h} & $1\,146$ & $166$ & $\mathbf{25}$ & $\mathbf{5}$ & $\mathbf{2}$ & $\mathbf{0.5}$ & $\mathbf{0.8}$ & $\mathbf{0.2}$ & $\mathbf{0.2}$ & $\mathbf{0.1}$ & $\mathbf{0.1}$\\
\hline
\textsf{DBLP} & \textsc{bfs} & $68$ & $35$ & $35$ & $34$ & $34$ & $34$ & $34$ & $35$ & $34$ & $35$ & $36$\\
& \textsc{dfs} & $282$ & $55$ & $42$ & $39$ & $39$ & $38$ & $38$ & $38$ & $38$ & $39$ & $39$\\
& \textsc{h} & $29$ & $\mathbf{5}$ & $\mathbf{5}$ & $\mathbf{5}$ & $\mathbf{5}$ & $\mathbf{5}$ & $\mathbf{6}$ & $\mathbf{6}$ & $\mathbf{6}$ & $\mathbf{6}$ & $\mathbf{6}$\\
\hline
\textsf{Obama} & \textsc{bfs} & $226$ & $42$ & $36$ & $34$ & $33$ & $31$ & $32$ & $32$ & $32$ & $32$ & $33$\\
\textsf{InIsrael} & \textsc{dfs} & $150$ & $51$ & $38$ & $34$ & $33$ & $31$ & $31$ & $31$ & $30$ & $31$ & $31$\\
& \textsc{h} & $177$ & $\mathbf{15}$ & $\mathbf{10}$ & $\mathbf{10}$ & $\mathbf{9}$ & $\mathbf{9}$ & $\mathbf{9}$ & $\mathbf{9}$ & $\mathbf{9}$ & $\mathbf{9}$ & $\mathbf{9}$\\
\hline
\textsf{Amazon} & \textsc{bfs} & $3\,981$ & $2\,125$ & $1\,364$ & $608$ & $582$ & $441$ & $234$ & $231$ & $192$ & $175$ & $167$\\
& \textsc{dfs} & $5\,278$ & $3\,103$ & $2\,105$ & $1\,198$ & $1\,072$ & $851$ & $523$ & $515$ & $434$ & $406$ & $371$\\
& \textsc{h} & $3\,913$ & $\mathbf{2\,109}$ & $\mathbf{1\,342}$ & $\mathbf{570}$ & $\mathbf{546}$ & $\mathbf{405}$ & $\mathbf{190}$ & $\mathbf{190}$ & $\mathbf{150}$ & $\mathbf{134}$ & $\mathbf{127}$\\
\hline
\textsf{Friendfeed} & \textsc{bfs} & $61\,113$ & $2\,464$ & $1\,004$ & $597$ & $333$ & $243$ & $185$ & $117$ & $108$ & $85$ & $59$\\
\textsf{Twitter} & \textsc{dfs} & $1\,973$ & $\mathbf{129}$ & $\mathbf{73}$ & $\mathbf{48}$ & $\mathbf{33}$ & $\mathbf{30}$ & $\mathbf{27}$ & $\mathbf{22}$ & $\mathbf{21}$ & $\mathbf{19}$ & $\mathbf{17}$\\
& \textsc{h} & $59\,520$ & $2\,340$ & $916$ & $523$ & $278$ & $193$ & $136$ & $78$ & $69$ & $49$ & $28$\\
\hline
\textsf{Higgs} & \textsc{bfs} & $2\,480$ & $351$ & $149$ & $91$ & $65$ & $62$ & $56$ & $50$ & $45$ & $40$ & $41$\\
& \textsc{dfs} & $640$ & $\mathbf{125}$ & $\mathbf{77}$ & $60$ & $52$ & $51$ & $46$ & $46$ & $42$ & $42$ & $39$\\
& \textsc{h} & $2\,169$ & $239$ & $80$ & $\mathbf{43}$ & $\mathbf{23}$ & $\mathbf{21}$ & $\mathbf{16}$ & $\mathbf{14}$ & $\mathbf{9}$ & $\mathbf{8}$ & $\mathbf{8}$\\
\hline
\textsf{Friendfeed} & \textsc{bfs} & $58\,278$ & $150$ & $51$ & $27$ & $25$ & $25$ & $24$ & $23$ & $23$ & $23$ & $23$\\
& \textsc{dfs} & $13\,356$ & $803$ & $220$ & $82$ & $68$ & $68$ & $66$ & $58$ & $58$ & $59$ & $57$\\
& \textsc{h} & $47\,179$ & $\mathbf{10}$ & $\mathbf{4}$ & $\mathbf{2}$ & $\mathbf{2}$ & $\mathbf{2}$ & $\mathbf{2}$ & $\mathbf{2}$ & $\mathbf{2}$ & $\mathbf{2}$ & $\mathbf{2}$\\
\hline
\end{tabular}
}

\end{table}

\subsection{Experimental results}
We experimentally prove the efficiency of the modifications adopted by our algorithms for multilayer community search by reporting a comparison against the original algorithms with no such modifications.
\revision{
Therefore, we consider as baselines the algorithms introduced in Section~\ref{sec:algorithms} for computing the entire multilayer core decomposition, i.e., the \bfs, \dfs, and \hybrid\ algorithms.
}
We vary the  size $|V_Q|$ of the query-vertex set from $1$ to $10$.
For every query-set size, we select -- uniformly at random -- a number of $100$ different query-vertex sets from the whole vertex set.
We also vary $\beta$ from $0.1$ and $100$.
The runtime with varying $|V_Q|$ is shown in Table~\ref{tab:cs_results}.
All results are averaged over the various query-vertex sets sampled.

In all datasets and for all algorithms, the modifications yield considerable improvement.
For $|V_Q| = 1$, which is the most demanding scenario in terms of runtime, we achieve from one to three orders of magnitude of speedup in all the cases (with the exception of \textsf{Amazon}).
As the number of query vertices increases, the modifications become even more effective: for $|V_Q| > 2$, we obtain at least one order of magnitude of speedup, up to a maximum of four orders of magnitude on the \textsf{Friendfeed} dataset.

\revision{
As a further insight, for a number of query vertices $|V_Q| \leq 2$, the runtime of the methods for multilayer community search is strongly dependent on the underlying algorithm for multilayer core decomposition.
For example, on the \textsf{SacchCere} and \textsf{Higgs} datasets, \textsc{h} is outperformed by \textsc{bfs} and \textsc{dfs}, respectively.
The picture is instead different for $|V_Q| > 2$: \textsc{h} turns out to be the fastest algorithm in all the datasets, with the exception of \textsf{FriendfeedTwitter}, for which \textsc{dfs} achieves better performance up to $10$ query vertices.
Therefore, in general, the core-lattice visit performed by \textsc{h} results to be more effective in identifying the solution multilayer core quickly. 
In the case of \textsf{FriendfeedTwitter} instead, the gap between the original runtime of \textsc{dfs} and \textsc{h} is so marked that, even if \textsc{h} yields better speedup, it is not able to outperform \textsc{dfs}.
This behavior is mainly motivated by the small number of layers of \textsf{FriendfeedTwitter} (only $2$), which, as already observed beforehand, favors \textsc{dfs} in terms of runtime.
}

\section{Conclusions}
\label{sec:conclusions}
Core decomposition has been proven to be a fundamental graph-analysis tool with plenty of applications.
In this work we study core decomposition in multilayer networks, characterizing its usefulness, its relation to other problems, and its  complexity.
We then devise three efficient algorithms for computing the whole core decomposition of a multilayer network  and we show a series of non-trivial applications of the core decomposition to solve related problems.
In particular:
\begin{itemize}
\item Given the large number of multilayer cores, we devise an efficient algorithm for efficiently computing the inner-most cores only.
\item We study densest-subgraph extraction in multilayer graphs as a proper optimization problem trading off between high density and layers exhibiting high density, and show how core decomposition can be used to approximate this problem with quality guarantees.
\item We show how the multilayer core-decomposition tool can be theoretically exploited as a data-reduction, preliminary step to speed up the extraction of frequent cross-graph quasi-cliques, and experimentally prove the effectiveness of our approach with respect to the original algorithm that searches for frequent cross-graph quasi-cliques in the whole input graph.
\item We generalize the multilayer community-search problem to the multilayer case and show how to exploit multilayer core decomposition to obtain optimal solutions to this problem.
\end{itemize}

In our on-going and future investigation we plan to employ multilayer core decomposition for the analysis of multilayer brain networks in which each layer represents a patient, vertices are brain regions, and edges are co-activation interactions measured by fMRI scans.
In this scenario the multilayer core-decomposition tool may be particularly powerful in the task of identifying common patterns to patients affected by diseases or under the assumption of drugs and, also, to select features in order to discriminate diseased patients from healthy individuals.


\begin{thebibliography}{88}


\ifx \showCODEN    \undefined \def \showCODEN     #1{\unskip}     \fi
\ifx \showDOI      \undefined \def \showDOI       #1{#1}\fi
\ifx \showISBNx    \undefined \def \showISBNx     #1{\unskip}     \fi
\ifx \showISBNxiii \undefined \def \showISBNxiii  #1{\unskip}     \fi
\ifx \showISSN     \undefined \def \showISSN      #1{\unskip}     \fi
\ifx \showLCCN     \undefined \def \showLCCN      #1{\unskip}     \fi
\ifx \shownote     \undefined \def \shownote      #1{#1}          \fi
\ifx \showarticletitle \undefined \def \showarticletitle #1{#1}   \fi
\ifx \showURL      \undefined \def \showURL       {\relax}        \fi
\providecommand\bibfield[2]{#2}
\providecommand\bibinfo[2]{#2}
\providecommand\natexlab[1]{#1}
\providecommand\showeprint[2][]{arXiv:#2}

\bibitem[\protect\citeauthoryear{Akoglu, Chau, Faloutsos, Tatti, Tong, and
  Vreeken}{Akoglu et~al\mbox{.}}{2013}]%
        {akoglu2013mining}
\bibfield{author}{\bibinfo{person}{Leman Akoglu}, \bibinfo{person}{Duen~Horng
  Chau}, \bibinfo{person}{Christos Faloutsos}, \bibinfo{person}{Nikolaj Tatti},
  \bibinfo{person}{Hanghang Tong}, {and} \bibinfo{person}{Jilles Vreeken}.}
  \bibinfo{year}{2013}\natexlab{}.
\newblock \showarticletitle{Mining Connection Pathways for Marked Nodes in
  Large Graphs}. In \bibinfo{booktitle}{\emph{Proc. of {SIAM} Int. Conf. on
  Data Mining {(SDM)}}}. \bibinfo{pages}{37--45}.
\newblock


\bibitem[\protect\citeauthoryear{Aksu, Canim, Chang, Korpeoglu, and
  Ulusoy}{Aksu et~al\mbox{.}}{2014}]%
        {AksuDistributed2014}
\bibfield{author}{\bibinfo{person}{Hidayet Aksu}, \bibinfo{person}{Mustafa
  Canim}, \bibinfo{person}{Yuan{-}Chi Chang}, \bibinfo{person}{Ibrahim
  Korpeoglu}, {and} \bibinfo{person}{{\"{O}}zg{\"{u}}r Ulusoy}.}
  \bibinfo{year}{2014}\natexlab{}.
\newblock \showarticletitle{Distributed $k$-core view materialization and
  maintenance for large dynamic graphs}.
\newblock \bibinfo{journal}{\emph{{IEEE} Transactions on Knowledge and Data
  Engineering {(TKDE)}}} \bibinfo{volume}{26}, \bibinfo{number}{10}
  (\bibinfo{year}{2014}), \bibinfo{pages}{2439--2452}.
\newblock


\bibitem[\protect\citeauthoryear{Alvarez-Hamelin, Dall'Asta, Barrat, and
  Vespignani}{Alvarez-Hamelin et~al\mbox{.}}{2005}]%
        {Alvarez-HamelinDBV05}
\bibfield{author}{\bibinfo{person}{J.~Ignacio Alvarez-Hamelin},
  \bibinfo{person}{Luca Dall'Asta}, \bibinfo{person}{Alain Barrat}, {and}
  \bibinfo{person}{Alessandro Vespignani}.} \bibinfo{year}{2005}\natexlab{}.
\newblock \showarticletitle{Large scale networks fingerprinting and
  visualization using the k-core decomposition}. In
  \bibinfo{booktitle}{\emph{Proc. of Conf. on Advances in Neural Information
  Processing Systems {(NIPS)}}}. \bibinfo{pages}{41--50}.
\newblock


\bibitem[\protect\citeauthoryear{Andersen and Chellapilla}{Andersen and
  Chellapilla}{2009}]%
        {AndersenC09}
\bibfield{author}{\bibinfo{person}{Reid Andersen} {and} \bibinfo{person}{Kumar
  Chellapilla}.} \bibinfo{year}{2009}\natexlab{}.
\newblock \showarticletitle{Finding dense subgraphs with size bounds}. In
  \bibinfo{booktitle}{\emph{Proc. of Int. Work. on Algorithms and Models for
  the Web-Graph {(WAW)}}}. \bibinfo{pages}{25--37}.
\newblock


\bibitem[\protect\citeauthoryear{Andersen and Lang}{Andersen and Lang}{2006}]%
        {Andersen1}
\bibfield{author}{\bibinfo{person}{Reid Andersen} {and}
  \bibinfo{person}{Kevin~J. Lang}.} \bibinfo{year}{2006}\natexlab{}.
\newblock \showarticletitle{Communities from seed sets}. In
  \bibinfo{booktitle}{\emph{Proc. of World Wide Web Conf. {(WWW)}}}.
  \bibinfo{pages}{223--232}.
\newblock


\bibitem[\protect\citeauthoryear{Arora, Karger, and Karpinski}{Arora
  et~al\mbox{.}}{1995}]%
        {AKK95}
\bibfield{author}{\bibinfo{person}{Sanjeev Arora}, \bibinfo{person}{David~R.
  Karger}, {and} \bibinfo{person}{Marek Karpinski}.}
  \bibinfo{year}{1995}\natexlab{}.
\newblock \showarticletitle{Polynomial time approximation schemes for dense
  instances of {NP}-hard problems}. In \bibinfo{booktitle}{\emph{Proc. of {ACM}
  Symp. on Theory of Computing {(STOC)}}}. \bibinfo{pages}{284--293}.
\newblock


\bibitem[\protect\citeauthoryear{Asahiro, Hassin, and Iwama}{Asahiro
  et~al\mbox{.}}{2002}]%
        {AHI02}
\bibfield{author}{\bibinfo{person}{Yuichi Asahiro}, \bibinfo{person}{Refael
  Hassin}, {and} \bibinfo{person}{Kazuo Iwama}.}
  \bibinfo{year}{2002}\natexlab{}.
\newblock \showarticletitle{Complexity of finding dense subgraphs}.
\newblock \bibinfo{journal}{\emph{Discrete Applied Mathematics}}
  \bibinfo{volume}{121}, \bibinfo{number}{1-3} (\bibinfo{year}{2002}),
  \bibinfo{pages}{15--26}.
\newblock


\bibitem[\protect\citeauthoryear{Asahiro, Iwama, Tamaki, and Tokuyama}{Asahiro
  et~al\mbox{.}}{2000}]%
        {AITT00}
\bibfield{author}{\bibinfo{person}{Yuichi Asahiro}, \bibinfo{person}{Kazuo
  Iwama}, \bibinfo{person}{Hisao Tamaki}, {and} \bibinfo{person}{Takeshi
  Tokuyama}.} \bibinfo{year}{2000}\natexlab{}.
\newblock \showarticletitle{Greedily finding a dense subgraph}.
\newblock \bibinfo{journal}{\emph{Journal of Algorithms}} \bibinfo{volume}{34},
  \bibinfo{number}{2} (\bibinfo{year}{2000}), \bibinfo{pages}{203--221}.
\newblock


\bibitem[\protect\citeauthoryear{Azimi-Tafreshi, G\'omez-Garde\~nes, and
  Dorogovtsev}{Azimi-Tafreshi et~al\mbox{.}}{2014}]%
        {MultiplexCores}
\bibfield{author}{\bibinfo{person}{N. Azimi-Tafreshi}, \bibinfo{person}{J.
  G\'omez-Garde\~nes}, {and} \bibinfo{person}{S.~N. Dorogovtsev}.}
  \bibinfo{year}{2014}\natexlab{}.
\newblock \showarticletitle{$k\text{\ensuremath{-}}\mathrm{core}$ percolation
  on multiplex networks}.
\newblock \bibinfo{journal}{\emph{Physical Review E}} \bibinfo{volume}{90},
  \bibinfo{number}{3} (\bibinfo{year}{2014}), \bibinfo{pages}{032816}.
\newblock


\bibitem[\protect\citeauthoryear{Bader and Hogue}{Bader and Hogue}{2003}]%
        {DBLP:journals/bmcbi/BaderH03}
\bibfield{author}{\bibinfo{person}{Gary~D. Bader} {and}
  \bibinfo{person}{Christopher W.~V. Hogue}.} \bibinfo{year}{2003}\natexlab{}.
\newblock \showarticletitle{An automated method for finding molecular complexes
  in large protein interaction networks}.
\newblock \bibinfo{journal}{\emph{BMC Bioinformatics}}  \bibinfo{volume}{4}
  (\bibinfo{year}{2003}), \bibinfo{pages}{2}.
\newblock


\bibitem[\protect\citeauthoryear{Bahmani, Kumar, and Vassilvitskii}{Bahmani
  et~al\mbox{.}}{2012}]%
        {DensestStreaming}
\bibfield{author}{\bibinfo{person}{Bahman Bahmani}, \bibinfo{person}{Ravi
  Kumar}, {and} \bibinfo{person}{Sergei Vassilvitskii}.}
  \bibinfo{year}{2012}\natexlab{}.
\newblock \showarticletitle{Densest subgraph in streaming and {M}ap{R}educe}.
\newblock \bibinfo{journal}{\emph{Proc. of the VLDB Endowment {(PVLDB)}}}
  \bibinfo{volume}{5}, \bibinfo{number}{5} (\bibinfo{year}{2012}),
  \bibinfo{pages}{454--465}.
\newblock


\bibitem[\protect\citeauthoryear{Balalau, Bonchi, Chan, Gullo, and
  Sozio}{Balalau et~al\mbox{.}}{2015}]%
        {Balalau15}
\bibfield{author}{\bibinfo{person}{Oana~Denisa Balalau},
  \bibinfo{person}{Francesco Bonchi}, \bibinfo{person}{TH Chan},
  \bibinfo{person}{Francesco Gullo}, {and} \bibinfo{person}{Mauro Sozio}.}
  \bibinfo{year}{2015}\natexlab{}.
\newblock \showarticletitle{Finding subgraphs with maximum total density and
  limited overlap}. In \bibinfo{booktitle}{\emph{Proc. of Int. Conf. on Web
  Search and Data Mining {(WSDM)}}}. \bibinfo{pages}{379--388}.
\newblock


\bibitem[\protect\citeauthoryear{Barbieri, Bonchi, Galimberti, and
  Gullo}{Barbieri et~al\mbox{.}}{2015}]%
        {BarbieriBGG15}
\bibfield{author}{\bibinfo{person}{Nicola Barbieri}, \bibinfo{person}{Francesco
  Bonchi}, \bibinfo{person}{Edoardo Galimberti}, {and}
  \bibinfo{person}{Francesco Gullo}.} \bibinfo{year}{2015}\natexlab{}.
\newblock \showarticletitle{Efficient and effective community search}.
\newblock \bibinfo{journal}{\emph{Data Mining and Knowledge Discovery
  {(DAMI)}}} \bibinfo{volume}{29}, \bibinfo{number}{5} (\bibinfo{year}{2015}),
  \bibinfo{pages}{1406--1433}.
\newblock


\bibitem[\protect\citeauthoryear{Batagelj, Mrvar, and Zaversnik}{Batagelj
  et~al\mbox{.}}{1999}]%
        {DBLP:conf/gd/BatageljMZ99}
\bibfield{author}{\bibinfo{person}{Vladimir Batagelj}, \bibinfo{person}{Andrej
  Mrvar}, {and} \bibinfo{person}{Matjaz Zaversnik}.}
  \bibinfo{year}{1999}\natexlab{}.
\newblock \showarticletitle{Partitioning approach to Vvsualization of large
  graphs}. In \bibinfo{booktitle}{\emph{Int. Symp. on Graph Drawing {(GD)}}}.
  \bibinfo{pages}{90--97}.
\newblock


\bibitem[\protect\citeauthoryear{Batagelj and Zaver{\v{s}}nik}{Batagelj and
  Zaver{\v{s}}nik}{2011}]%
        {batagelj2011fast}
\bibfield{author}{\bibinfo{person}{Vladimir Batagelj} {and}
  \bibinfo{person}{Matja{\v{z}} Zaver{\v{s}}nik}.}
  \bibinfo{year}{2011}\natexlab{}.
\newblock \showarticletitle{Fast algorithms for determining (generalized) core
  groups in social networks}.
\newblock \bibinfo{journal}{\emph{Advances in Data Analysis and Classification
  {(ADAC)}}} \bibinfo{volume}{5}, \bibinfo{number}{2} (\bibinfo{year}{2011}),
  \bibinfo{pages}{129--145}.
\newblock


\bibitem[\protect\citeauthoryear{Berlingerio, Coscia, and
  Giannotti}{Berlingerio et~al\mbox{.}}{2011}]%
        {berlingerio2011findingredundant}
\bibfield{author}{\bibinfo{person}{Michele Berlingerio},
  \bibinfo{person}{Michele Coscia}, {and} \bibinfo{person}{Fosca Giannotti}.}
  \bibinfo{year}{2011}\natexlab{}.
\newblock \showarticletitle{Finding redundant and complementary communities in
  multidimensional networks}. In \bibinfo{booktitle}{\emph{Proc. of Int. Conf.
  on Information and Knowledge Management {(CIKM)}}}.
  \bibinfo{pages}{2181--2184}.
\newblock


\bibitem[\protect\citeauthoryear{Bhattacharya, Henzinger, Nanongkai, and
  Tsourakakis}{Bhattacharya et~al\mbox{.}}{2015}]%
        {BhattacharyaHNT15}
\bibfield{author}{\bibinfo{person}{Sayan Bhattacharya}, \bibinfo{person}{Monika
  Henzinger}, \bibinfo{person}{Danupon Nanongkai}, {and}
  \bibinfo{person}{Charalampos~E. Tsourakakis}.}
  \bibinfo{year}{2015}\natexlab{}.
\newblock \showarticletitle{Space-and time-efficient algorithm for maintaining
  dense subgraphs on one-pass dynamic streams}. In
  \bibinfo{booktitle}{\emph{Proc. of {ACM} Symp. on Theory of Computing
  {(STOC)}}}. \bibinfo{pages}{173--182}.
\newblock


\bibitem[\protect\citeauthoryear{Blei, Ng, and Jordan}{Blei
  et~al\mbox{.}}{2003}]%
        {blei2003latent}
\bibfield{author}{\bibinfo{person}{David~M Blei}, \bibinfo{person}{Andrew~Y
  Ng}, {and} \bibinfo{person}{Michael~I Jordan}.}
  \bibinfo{year}{2003}\natexlab{}.
\newblock \showarticletitle{Latent {D}irichlet allocation}.
\newblock \bibinfo{journal}{\emph{Journal of Machine Learning Research
  {(JMLR)}}}  \bibinfo{volume}{3} (\bibinfo{year}{2003}),
  \bibinfo{pages}{993--1022}.
\newblock


\bibitem[\protect\citeauthoryear{Boden, G{\"u}nnemann, Hoffmann, and
  Seidl}{Boden et~al\mbox{.}}{2012}]%
        {boden2012mining}
\bibfield{author}{\bibinfo{person}{Brigitte Boden}, \bibinfo{person}{Stephan
  G{\"u}nnemann}, \bibinfo{person}{Holger Hoffmann}, {and}
  \bibinfo{person}{Thomas Seidl}.} \bibinfo{year}{2012}\natexlab{}.
\newblock \showarticletitle{Mining coherent subgraphs in multi-layer graphs
  with edge labels}. In \bibinfo{booktitle}{\emph{Proc. of {ACM} {SIGKDD} Int.
  Conf. on Knowledge Discovery and Data Mining}}. \bibinfo{pages}{1258--1266}.
\newblock


\bibitem[\protect\citeauthoryear{Bonchi, Gionis, Gullo, Tsourakakis, and
  Ukkonen}{Bonchi et~al\mbox{.}}{2015}]%
        {bonchi2015chromatic}
\bibfield{author}{\bibinfo{person}{Francesco Bonchi},
  \bibinfo{person}{Aristides Gionis}, \bibinfo{person}{Francesco Gullo},
  \bibinfo{person}{Charalampos~E. Tsourakakis}, {and} \bibinfo{person}{Antti
  Ukkonen}.} \bibinfo{year}{2015}\natexlab{}.
\newblock \showarticletitle{Chromatic correlation clustering}.
\newblock \bibinfo{journal}{\emph{ACM Transactions on Knowledge Discovery from
  Data {(TKDD)}}} \bibinfo{volume}{9}, \bibinfo{number}{4}
  (\bibinfo{year}{2015}), \bibinfo{pages}{34}.
\newblock


\bibitem[\protect\citeauthoryear{Bonchi, Gullo, and Kaltenbrunner}{Bonchi
  et~al\mbox{.}}{2018}]%
        {bonchi18core}
\bibfield{author}{\bibinfo{person}{Francesco Bonchi},
  \bibinfo{person}{Francesco Gullo}, {and} \bibinfo{person}{Andreas
  Kaltenbrunner}.} \bibinfo{year}{2018}\natexlab{}.
\newblock \showarticletitle{Core decomposition of massive, information-rich
  graphs}.
\newblock In \bibinfo{booktitle}{\emph{Encyclopedia of Social Network Analysis
  and Mining, 2nd Edition}}. \bibinfo{publisher}{Springer},
  \bibinfo{pages}{419--428}.
\newblock


\bibitem[\protect\citeauthoryear{Bonchi, Gullo, Kaltenbrunner, and
  Volkovich}{Bonchi et~al\mbox{.}}{2014}]%
        {bonchi14cores}
\bibfield{author}{\bibinfo{person}{Francesco Bonchi},
  \bibinfo{person}{Francesco Gullo}, \bibinfo{person}{Andreas Kaltenbrunner},
  {and} \bibinfo{person}{Yana Volkovich}.} \bibinfo{year}{2014}\natexlab{}.
\newblock \showarticletitle{Core decomposition of uncertain graphs}. In
  \bibinfo{booktitle}{\emph{Proc. of {ACM} {SIGKDD} Int. Conf. on Knowledge
  Discovery and Data Mining}}. \bibinfo{pages}{1316--1325}.
\newblock


\bibitem[\protect\citeauthoryear{Cai, Shao, He, Yan, and Han}{Cai
  et~al\mbox{.}}{2005}]%
        {cai2005community}
\bibfield{author}{\bibinfo{person}{Deng Cai}, \bibinfo{person}{Zheng Shao},
  \bibinfo{person}{Xiaofei He}, \bibinfo{person}{Xifeng Yan}, {and}
  \bibinfo{person}{Jiawei Han}.} \bibinfo{year}{2005}\natexlab{}.
\newblock \showarticletitle{Community mining from multi-relational networks}.
  In \bibinfo{booktitle}{\emph{Proc. of Europ. Machine Learning and Principles
  and Practice of Knowledge Discovery in Databases {(ECML PKDD)}}}.
  \bibinfo{pages}{445--452}.
\newblock


\bibitem[\protect\citeauthoryear{Charikar}{Charikar}{2000}]%
        {Char00}
\bibfield{author}{\bibinfo{person}{Moses Charikar}.}
  \bibinfo{year}{2000}\natexlab{}.
\newblock \showarticletitle{Greedy approximation algorithms for finding dense
  components in a graph}. In \bibinfo{booktitle}{\emph{Proc. of Int. Work. on
  Approximation Algorithms for Combinatorial Optimization Problems
  {(APPROX)}}}. \bibinfo{pages}{84--95}.
\newblock


\bibitem[\protect\citeauthoryear{Charikar, Naamad, and Wu}{Charikar
  et~al\mbox{.}}{2018}]%
        {charikar2018finding}
\bibfield{author}{\bibinfo{person}{Moses Charikar}, \bibinfo{person}{Yonatan
  Naamad}, {and} \bibinfo{person}{Jimmy Wu}.} \bibinfo{year}{2018}\natexlab{}.
\newblock \showarticletitle{On finding dense common subgraphs}.
\newblock \bibinfo{journal}{\emph{CoRR}}  \bibinfo{volume}{abs/1802.06361}
  (\bibinfo{year}{2018}).
\newblock


\bibitem[\protect\citeauthoryear{Cheng, Ke, Chu, and {\"O}zsu}{Cheng
  et~al\mbox{.}}{2011}]%
        {DiskCores}
\bibfield{author}{\bibinfo{person}{James Cheng}, \bibinfo{person}{Yiping Ke},
  \bibinfo{person}{Shumo Chu}, {and} \bibinfo{person}{M~Tamer {\"O}zsu}.}
  \bibinfo{year}{2011}\natexlab{}.
\newblock \showarticletitle{Efficient core decomposition in massive networks}.
  In \bibinfo{booktitle}{\emph{Proc. of {IEEE} Int. Conf. on Data Engineering
  {(ICDE)}}}. \bibinfo{pages}{51--62}.
\newblock


\bibitem[\protect\citeauthoryear{Cui, Xiao, Wang, and Wang}{Cui
  et~al\mbox{.}}{2014}]%
        {SozioLocalSIGMOD14}
\bibfield{author}{\bibinfo{person}{Wanyun Cui}, \bibinfo{person}{Yanghua Xiao},
  \bibinfo{person}{Haixun Wang}, {and} \bibinfo{person}{Wei Wang}.}
  \bibinfo{year}{2014}\natexlab{}.
\newblock \showarticletitle{Local search of communities in large graphs}. In
  \bibinfo{booktitle}{\emph{Proc. of {ACM} Int. Conf. on Management of Data
  {(SIGMOD)}}}. \bibinfo{pages}{991--1002}.
\newblock


\bibitem[\protect\citeauthoryear{Dickison, Magnani, and Rossi}{Dickison
  et~al\mbox{.}}{2016}]%
        {DickisonMagnaniRossi2016}
\bibfield{author}{\bibinfo{person}{Mark~E. Dickison}, \bibinfo{person}{Matteo
  Magnani}, {and} \bibinfo{person}{Luca Rossi}.}
  \bibinfo{year}{2016}\natexlab{}.
\newblock \bibinfo{booktitle}{\emph{{Multilayer social networks}}}.
\newblock \bibinfo{publisher}{Cambridge University Press}.
\newblock
\showISBNx{978-1107438750}


\bibitem[\protect\citeauthoryear{Du, Jin, Ding, Lee, and Thornton}{Du
  et~al\mbox{.}}{2009}]%
        {Du}
\bibfield{author}{\bibinfo{person}{Xiaoxi Du}, \bibinfo{person}{Ruoming Jin},
  \bibinfo{person}{Liang Ding}, \bibinfo{person}{Victor~E. Lee}, {and}
  \bibinfo{person}{John~H. Thornton, Jr.}} \bibinfo{year}{2009}\natexlab{}.
\newblock \showarticletitle{Migration motif: a spatial - temporal pattern
  mining approach for financial markets}. In \bibinfo{booktitle}{\emph{Proc. of
  {ACM} {SIGKDD} Int. Conf. on Knowledge Discovery and Data Mining}}.
  \bibinfo{pages}{1135--1144}.
\newblock


\bibitem[\protect\citeauthoryear{Epasto, Lattanzi, and Sozio}{Epasto
  et~al\mbox{.}}{2015}]%
        {EpastoEfficient2015}
\bibfield{author}{\bibinfo{person}{Alessandro Epasto}, \bibinfo{person}{Silvio
  Lattanzi}, {and} \bibinfo{person}{Mauro Sozio}.}
  \bibinfo{year}{2015}\natexlab{}.
\newblock \showarticletitle{Efficient densest subgraph computation in evolving
  graphs}. In \bibinfo{booktitle}{\emph{Proc. of World Wide Web Conf.
  {(WWW)}}}. \bibinfo{pages}{300--310}.
\newblock


\bibitem[\protect\citeauthoryear{Eppstein, L{\"o}ffler, and Strash}{Eppstein
  et~al\mbox{.}}{2010}]%
        {EppsteinLS10}
\bibfield{author}{\bibinfo{person}{David Eppstein}, \bibinfo{person}{Maarten
  L{\"o}ffler}, {and} \bibinfo{person}{Darren Strash}.}
  \bibinfo{year}{2010}\natexlab{}.
\newblock \showarticletitle{Listing all maximal cliques in sparse graphs in
  near-optimal time}. In \bibinfo{booktitle}{\emph{Proc. of Int. Symp. on
  Algorithms and Computation {(ISAAC)}}}. \bibinfo{pages}{403--414}.
\newblock


\bibitem[\protect\citeauthoryear{Faloutsos, McCurley, and Tomkins}{Faloutsos
  et~al\mbox{.}}{2004}]%
        {connect}
\bibfield{author}{\bibinfo{person}{Christos Faloutsos},
  \bibinfo{person}{Kevin~S. McCurley}, {and} \bibinfo{person}{Andrew Tomkins}.}
  \bibinfo{year}{2004}\natexlab{}.
\newblock \showarticletitle{Fast discovery of connection subgraphs}. In
  \bibinfo{booktitle}{\emph{Proc. of {ACM} {SIGKDD} Int. Conf. on Knowledge
  Discovery and Data Mining}}. \bibinfo{pages}{118--127}.
\newblock


\bibitem[\protect\citeauthoryear{Fang, Cheng, Chen, Luo, and Hu}{Fang
  et~al\mbox{.}}{2017a}]%
        {fang2017attributed}
\bibfield{author}{\bibinfo{person}{Yixiang Fang}, \bibinfo{person}{Reynold
  Cheng}, \bibinfo{person}{Yankai Chen}, \bibinfo{person}{Siqiang Luo}, {and}
  \bibinfo{person}{Jiafeng Hu}.} \bibinfo{year}{2017}\natexlab{a}.
\newblock \showarticletitle{Effective and efficient attributed community
  search}.
\newblock \bibinfo{journal}{\emph{The VLDB Journal}} \bibinfo{volume}{26},
  \bibinfo{number}{6} (\bibinfo{year}{2017}), \bibinfo{pages}{803--828}.
\newblock


\bibitem[\protect\citeauthoryear{Fang, Cheng, Li, Luo, and Hu}{Fang
  et~al\mbox{.}}{2017b}]%
        {fang2017spatial}
\bibfield{author}{\bibinfo{person}{Yixiang Fang}, \bibinfo{person}{Reynold
  Cheng}, \bibinfo{person}{Xiaodong Li}, \bibinfo{person}{Siqiang Luo}, {and}
  \bibinfo{person}{Jiafeng Hu}.} \bibinfo{year}{2017}\natexlab{b}.
\newblock \showarticletitle{Effective community search over large spatial
  graphs}.
\newblock \bibinfo{journal}{\emph{Proc. of the VLDB Endowment {(PVLDB)}}}
  \bibinfo{volume}{10}, \bibinfo{number}{6} (\bibinfo{year}{2017}),
  \bibinfo{pages}{709--720}.
\newblock


\bibitem[\protect\citeauthoryear{Feige, Kortsarz, and Peleg}{Feige
  et~al\mbox{.}}{2001}]%
        {FPK01}
\bibfield{author}{\bibinfo{person}{Uriel Feige}, \bibinfo{person}{Guy
  Kortsarz}, {and} \bibinfo{person}{David Peleg}.}
  \bibinfo{year}{2001}\natexlab{}.
\newblock \showarticletitle{The dense \emph{k}-subgraph problem}.
\newblock \bibinfo{journal}{\emph{Algorithmica}} \bibinfo{volume}{29},
  \bibinfo{number}{3} (\bibinfo{year}{2001}), \bibinfo{pages}{410--421}.
\newblock


\bibitem[\protect\citeauthoryear{Fratkin, Naughton, Brutlag, and
  Batzoglou}{Fratkin et~al\mbox{.}}{2006}]%
        {fratkin}
\bibfield{author}{\bibinfo{person}{Eugene Fratkin}, \bibinfo{person}{Brian~T.
  Naughton}, \bibinfo{person}{Douglas~L. Brutlag}, {and}
  \bibinfo{person}{Serafim Batzoglou}.} \bibinfo{year}{2006}\natexlab{}.
\newblock \showarticletitle{{MotifCut}: regulatory motifs finding with maximum
  density subgraphs}. In \bibinfo{booktitle}{\emph{Proc. of Int. Conf. on
  Intelligent Systems for Molecular Biology {(ISMB)}}}.
  \bibinfo{pages}{156--157}.
\newblock


\bibitem[\protect\citeauthoryear{Galbrun, Gionis, and Tatti}{Galbrun
  et~al\mbox{.}}{2016}]%
        {GalbrunGT16}
\bibfield{author}{\bibinfo{person}{Esther Galbrun}, \bibinfo{person}{Aristides
  Gionis}, {and} \bibinfo{person}{Nikolaj Tatti}.}
  \bibinfo{year}{2016}\natexlab{}.
\newblock \showarticletitle{Top-k overlapping densest subgraphs}.
\newblock \bibinfo{journal}{\emph{Data Mining and Knowledge Discovery
  {(DAMI)}}} \bibinfo{volume}{30}, \bibinfo{number}{5} (\bibinfo{year}{2016}),
  \bibinfo{pages}{1134--1165}.
\newblock


\bibitem[\protect\citeauthoryear{Galimberti, Barrat, Bonchi, Cattuto, and
  Gullo}{Galimberti et~al\mbox{.}}{2018}]%
        {galimberti2018mining}
\bibfield{author}{\bibinfo{person}{Edoardo Galimberti}, \bibinfo{person}{Alain
  Barrat}, \bibinfo{person}{Francesco Bonchi}, \bibinfo{person}{Ciro Cattuto},
  {and} \bibinfo{person}{Francesco Gullo}.} \bibinfo{year}{2018}\natexlab{}.
\newblock \showarticletitle{Mining (maximal) span-cores from temporal
  networks}. In \bibinfo{booktitle}{\emph{Proc. of Int. Conf. on Information
  and Knowledge Management {(CIKM)}}}. \bibinfo{pages}{107--116}.
\newblock


\bibitem[\protect\citeauthoryear{Galimberti, Bonchi, and Gullo}{Galimberti
  et~al\mbox{.}}{2017}]%
        {galimberti2017core}
\bibfield{author}{\bibinfo{person}{Edoardo Galimberti},
  \bibinfo{person}{Francesco Bonchi}, {and} \bibinfo{person}{Francesco Gullo}.}
  \bibinfo{year}{2017}\natexlab{}.
\newblock \showarticletitle{Core decomposition and densest subgraph in
  multilayer networks}. In \bibinfo{booktitle}{\emph{Proc. of Int. Conf. on
  Information and Knowledge Management {(CIKM)}}}. \bibinfo{pages}{1807--1816}.
\newblock


\bibitem[\protect\citeauthoryear{Garas, Schweitzer, and Havlin}{Garas
  et~al\mbox{.}}{2012}]%
        {WeigthedCores}
\bibfield{author}{\bibinfo{person}{Antonios Garas}, \bibinfo{person}{Frank
  Schweitzer}, {and} \bibinfo{person}{Shlomo Havlin}.}
  \bibinfo{year}{2012}\natexlab{}.
\newblock \showarticletitle{A \emph{k}-shell decomposition method for weighted
  networks}.
\newblock \bibinfo{journal}{\emph{New Journal of Physics}}
  \bibinfo{volume}{14}, \bibinfo{number}{8} (\bibinfo{year}{2012}),
  \bibinfo{pages}{083030}.
\newblock


\bibitem[\protect\citeauthoryear{Garc{\'{\i}}a, Mavrodiev, and
  Schweitzer}{Garc{\'{\i}}a et~al\mbox{.}}{2013}]%
        {GArcia2013}
\bibfield{author}{\bibinfo{person}{David Garc{\'{\i}}a},
  \bibinfo{person}{Pavlin Mavrodiev}, {and} \bibinfo{person}{Frank
  Schweitzer}.} \bibinfo{year}{2013}\natexlab{}.
\newblock \showarticletitle{Social resilience in online communities: the
  autopsy of {F}riendster}. In \bibinfo{booktitle}{\emph{Proc. of {ACM} Conf.
  on Online Social Networks {(COSN)}}}. \bibinfo{pages}{39--50}.
\newblock


\bibitem[\protect\citeauthoryear{Giatsidis, Thilikos, and
  Vazirgiannis}{Giatsidis et~al\mbox{.}}{2013}]%
        {DirectedCores}
\bibfield{author}{\bibinfo{person}{Christos Giatsidis},
  \bibinfo{person}{Dimitrios~M Thilikos}, {and} \bibinfo{person}{Michalis
  Vazirgiannis}.} \bibinfo{year}{2013}\natexlab{}.
\newblock \showarticletitle{D-cores: measuring collaboration of directed graphs
  based on degeneracy}.
\newblock \bibinfo{journal}{\emph{Knowledge and Information Systems {(KAIS)}}}
  \bibinfo{volume}{35}, \bibinfo{number}{2} (\bibinfo{year}{2013}),
  \bibinfo{pages}{311--343}.
\newblock


\bibitem[\protect\citeauthoryear{Gibson, Kumar, and Tomkins}{Gibson
  et~al\mbox{.}}{2005}]%
        {gibson}
\bibfield{author}{\bibinfo{person}{David Gibson}, \bibinfo{person}{Ravi Kumar},
  {and} \bibinfo{person}{Andrew Tomkins}.} \bibinfo{year}{2005}\natexlab{}.
\newblock \showarticletitle{Discovering large dense subgraphs in massive
  graphs}. In \bibinfo{booktitle}{\emph{Proc. of Int. Conf. on Very Large Data
  Bases {(VLDB)}}}. \bibinfo{pages}{721--732}.
\newblock


\bibitem[\protect\citeauthoryear{Goldberg}{Goldberg}{1984}]%
        {Goldberg84}
\bibfield{author}{\bibinfo{person}{A.~V. Goldberg}.}
  \bibinfo{year}{1984}\natexlab{}.
\newblock \bibinfo{booktitle}{\emph{Finding a Maximum Density Subgraph}}.
\newblock \bibinfo{type}{{T}echnical {R}eport}.
  \bibinfo{institution}{University of California at Berkeley}.
\newblock


\bibitem[\protect\citeauthoryear{Healy, Janssen, Milios, and Aiello}{Healy
  et~al\mbox{.}}{2006}]%
        {HealyJMA06}
\bibfield{author}{\bibinfo{person}{John Healy}, \bibinfo{person}{Jeannette
  Janssen}, \bibinfo{person}{Evangelos Milios}, {and} \bibinfo{person}{William
  Aiello}.} \bibinfo{year}{2006}\natexlab{}.
\newblock \showarticletitle{Characterization of graphs using degree cores}. In
  \bibinfo{booktitle}{\emph{Proc. of Int. Work. on Algorithms and Models for
  the Web-Graph {(WAW)}}}. \bibinfo{pages}{137--148}.
\newblock


\bibitem[\protect\citeauthoryear{Huang and Lakshmanan}{Huang and
  Lakshmanan}{2017}]%
        {huang2017attribute}
\bibfield{author}{\bibinfo{person}{Xin Huang} {and} \bibinfo{person}{Laks~V.S.
  Lakshmanan}.} \bibinfo{year}{2017}\natexlab{}.
\newblock \showarticletitle{Attribute-driven community search}.
\newblock \bibinfo{journal}{\emph{Proc. of the VLDB Endowment {(PVLDB)}}}
  \bibinfo{volume}{10}, \bibinfo{number}{9} (\bibinfo{year}{2017}),
  \bibinfo{pages}{949--960}.
\newblock


\bibitem[\protect\citeauthoryear{Huang, Lakshmanan, and Xu}{Huang
  et~al\mbox{.}}{2017}]%
        {HuangLX17}
\bibfield{author}{\bibinfo{person}{Xin Huang}, \bibinfo{person}{Laks V.~S.
  Lakshmanan}, {and} \bibinfo{person}{Jianliang Xu}.}
  \bibinfo{year}{2017}\natexlab{}.
\newblock \showarticletitle{Community search over big graphs: models,
  algorithms, and opportunities}. In \bibinfo{booktitle}{\emph{Proc. of {IEEE}
  Int. Conf. on Data Engineering {(ICDE)}}}. \bibinfo{pages}{1451--1454}.
\newblock


\bibitem[\protect\citeauthoryear{Interdonato, Tagarelli, Ienco, Sallaberry, and
  Poncelet}{Interdonato et~al\mbox{.}}{2017}]%
        {InterdonatoTISP17}
\bibfield{author}{\bibinfo{person}{Roberto Interdonato},
  \bibinfo{person}{Andrea Tagarelli}, \bibinfo{person}{Dino Ienco},
  \bibinfo{person}{Arnaud Sallaberry}, {and} \bibinfo{person}{Pascal
  Poncelet}.} \bibinfo{year}{2017}\natexlab{}.
\newblock \showarticletitle{Local community detection in multilayer networks}.
\newblock \bibinfo{journal}{\emph{Data Mining and Knowledge Discovery
  {(DAMI)}}} \bibinfo{volume}{31}, \bibinfo{number}{5} (\bibinfo{year}{2017}),
  \bibinfo{pages}{1444--1479}.
\newblock


\bibitem[\protect\citeauthoryear{Jethava and Beerenwinkel}{Jethava and
  Beerenwinkel}{2015}]%
        {jethava2015finding}
\bibfield{author}{\bibinfo{person}{Vinay Jethava} {and} \bibinfo{person}{Niko
  Beerenwinkel}.} \bibinfo{year}{2015}\natexlab{}.
\newblock \showarticletitle{Finding dense subgraphs in relational graphs}. In
  \bibinfo{booktitle}{\emph{Proc. of Europ. Machine Learning and Principles and
  Practice of Knowledge Discovery in Databases {(ECML PKDD)}}}.
  \bibinfo{pages}{641--654}.
\newblock


\bibitem[\protect\citeauthoryear{Jiang and Pei}{Jiang and Pei}{2009}]%
        {jiang2009mining}
\bibfield{author}{\bibinfo{person}{Daxin Jiang} {and} \bibinfo{person}{Jian
  Pei}.} \bibinfo{year}{2009}\natexlab{}.
\newblock \showarticletitle{Mining frequent cross-graph quasi-cliques}.
\newblock \bibinfo{journal}{\emph{ACM Transactions on Knowledge Discovery from
  Data {(TKDD)}}} \bibinfo{volume}{2}, \bibinfo{number}{4}
  (\bibinfo{year}{2009}), \bibinfo{pages}{16}.
\newblock


\bibitem[\protect\citeauthoryear{Khaouid, Barsky, Venkatesh, and Thomo}{Khaouid
  et~al\mbox{.}}{2015}]%
        {KhaouidKcore2015}
\bibfield{author}{\bibinfo{person}{Wissam Khaouid}, \bibinfo{person}{Marina
  Barsky}, \bibinfo{person}{S. Venkatesh}, {and} \bibinfo{person}{Alex Thomo}.}
  \bibinfo{year}{2015}\natexlab{}.
\newblock \showarticletitle{K-core decomposition of large networks on a single
  {PC}}.
\newblock \bibinfo{journal}{\emph{Proc. of the VLDB Endowment {(PVLDB)}}}
  \bibinfo{volume}{9}, \bibinfo{number}{1} (\bibinfo{year}{2015}),
  \bibinfo{pages}{13--23}.
\newblock


\bibitem[\protect\citeauthoryear{Kitsak, Gallos, Havlin, Liljeros, Muchnik,
  Stanley, and Makse}{Kitsak et~al\mbox{.}}{2010}]%
        {Kitsak2010}
\bibfield{author}{\bibinfo{person}{Maksim Kitsak}, \bibinfo{person}{Lazaros~K
  Gallos}, \bibinfo{person}{Shlomo Havlin}, \bibinfo{person}{Fredrik Liljeros},
  \bibinfo{person}{Lev Muchnik}, \bibinfo{person}{H~Eugene Stanley}, {and}
  \bibinfo{person}{Hern{\'a}n~A Makse}.} \bibinfo{year}{2010}\natexlab{}.
\newblock \showarticletitle{Identification of influential spreaders in complex
  networks}.
\newblock \bibinfo{journal}{\emph{Nature Physics}} \bibinfo{volume}{6},
  \bibinfo{number}{11} (\bibinfo{year}{2010}), \bibinfo{pages}{888--893}.
\newblock


\bibitem[\protect\citeauthoryear{Kloumann and Kleinberg}{Kloumann and
  Kleinberg}{2014}]%
        {Kloumann}
\bibfield{author}{\bibinfo{person}{Isabel~M. Kloumann} {and}
  \bibinfo{person}{Jon~M. Kleinberg}.} \bibinfo{year}{2014}\natexlab{}.
\newblock \showarticletitle{Community membership identification from small seed
  sets}. In \bibinfo{booktitle}{\emph{Proc. of {ACM} {SIGKDD} Int. Conf. on
  Knowledge Discovery and Data Mining}}. \bibinfo{pages}{1366--1375}.
\newblock


\bibitem[\protect\citeauthoryear{Kortsarz and Peleg}{Kortsarz and
  Peleg}{1994}]%
        {KortsarzP94}
\bibfield{author}{\bibinfo{person}{Guy Kortsarz} {and} \bibinfo{person}{David
  Peleg}.} \bibinfo{year}{1994}\natexlab{}.
\newblock \showarticletitle{Generating sparse 2-spanners}.
\newblock \bibinfo{journal}{\emph{Journal of Algorithms}} \bibinfo{volume}{17},
  \bibinfo{number}{2} (\bibinfo{year}{1994}), \bibinfo{pages}{222--236}.
\newblock


\bibitem[\protect\citeauthoryear{Langston and et~al.}{Langston and
  et~al.}{2005}]%
        {langston}
\bibfield{author}{\bibinfo{person}{Michael~A. Langston} {and}
  \bibinfo{person}{et al.}} \bibinfo{year}{2005}\natexlab{}.
\newblock \showarticletitle{A combinatorial approach to the analysis of
  differential gene expression data: the use of graph algorithms for disease
  prediction and screening}.
\newblock In \bibinfo{booktitle}{\emph{Methods of Microarray Data Analysis
  IV}}. \bibinfo{pages}{223--238}.
\newblock


\bibitem[\protect\citeauthoryear{Lee, Min, and Goh}{Lee et~al\mbox{.}}{2015}]%
        {lee2015towards}
\bibfield{author}{\bibinfo{person}{Kyu-Min Lee}, \bibinfo{person}{Byungjoon
  Min}, {and} \bibinfo{person}{Kwang-Il Goh}.} \bibinfo{year}{2015}\natexlab{}.
\newblock \showarticletitle{Towards real-world complexity: an introduction to
  multiplex networks}.
\newblock \bibinfo{journal}{\emph{The European Physical Journal B}}
  \bibinfo{volume}{88}, \bibinfo{number}{2} (\bibinfo{year}{2015}),
  \bibinfo{pages}{48}.
\newblock


\bibitem[\protect\citeauthoryear{Lee, Ruan, Jin, and Aggarwal}{Lee
  et~al\mbox{.}}{2010}]%
        {aggarwal}
\bibfield{author}{\bibinfo{person}{Victor~E. Lee}, \bibinfo{person}{Ning Ruan},
  \bibinfo{person}{Ruoming Jin}, {and} \bibinfo{person}{Charu~C. Aggarwal}.}
  \bibinfo{year}{2010}\natexlab{}.
\newblock \showarticletitle{A survey of algorithms for dense subgraph
  discovery}.
\newblock In \bibinfo{booktitle}{\emph{Managing and Mining Graph Data}}.
\newblock


\bibitem[\protect\citeauthoryear{Li, Yu, and Mao}{Li et~al\mbox{.}}{2014}]%
        {LiEfficient2014}
\bibfield{author}{\bibinfo{person}{Rong{-}Hua Li}, \bibinfo{person}{Jeffrey~Xu
  Yu}, {and} \bibinfo{person}{Rui Mao}.} \bibinfo{year}{2014}\natexlab{}.
\newblock \showarticletitle{Efficient core maintenance in large dynamic
  graphs}.
\newblock \bibinfo{journal}{\emph{{IEEE} Transactions on Knowledge and Data
  Engineering {(TKDE)}}} \bibinfo{volume}{26}, \bibinfo{number}{10}
  (\bibinfo{year}{2014}), \bibinfo{pages}{2453--2465}.
\newblock


\bibitem[\protect\citeauthoryear{Matula and Beck}{Matula and Beck}{1983}]%
        {MatulaB83}
\bibfield{author}{\bibinfo{person}{David~W. Matula} {and}
  \bibinfo{person}{Leland~L. Beck}.} \bibinfo{year}{1983}\natexlab{}.
\newblock \showarticletitle{Smallest-last ordering and clustering and graph
  coloring algorithms}.
\newblock \bibinfo{journal}{\emph{Journal of the {ACM} {(JACM)}}}
  \bibinfo{volume}{30}, \bibinfo{number}{3} (\bibinfo{year}{1983}).
\newblock


\bibitem[\protect\citeauthoryear{Montresor, De~Pellegrini, and
  Miorandi}{Montresor et~al\mbox{.}}{2013}]%
        {DistributedCores1}
\bibfield{author}{\bibinfo{person}{Alberto Montresor},
  \bibinfo{person}{Francesco De~Pellegrini}, {and} \bibinfo{person}{Daniele
  Miorandi}.} \bibinfo{year}{2013}\natexlab{}.
\newblock \showarticletitle{Distributed k-core decomposition}.
\newblock \bibinfo{journal}{\emph{{IEEE} Transactions on Parallel and
  Distributed Systems {(TPDS)}}} \bibinfo{volume}{24}, \bibinfo{number}{2}
  (\bibinfo{year}{2013}), \bibinfo{pages}{288--300}.
\newblock


\bibitem[\protect\citeauthoryear{Mucha, Richardson, Macon, Porter, and
  Onnela}{Mucha et~al\mbox{.}}{2010}]%
        {mucha2010community}
\bibfield{author}{\bibinfo{person}{Peter~J Mucha}, \bibinfo{person}{Thomas
  Richardson}, \bibinfo{person}{Kevin Macon}, \bibinfo{person}{Mason~A Porter},
  {and} \bibinfo{person}{Jukka-Pekka Onnela}.} \bibinfo{year}{2010}\natexlab{}.
\newblock \showarticletitle{Community structure in time-dependent, multiscale,
  and multiplex networks}.
\newblock \bibinfo{journal}{\emph{Science}} \bibinfo{volume}{328},
  \bibinfo{number}{5980} (\bibinfo{year}{2010}), \bibinfo{pages}{876--878}.
\newblock


\bibitem[\protect\citeauthoryear{Nasir, Gionis, Morales, and
  Girdzijauskas}{Nasir et~al\mbox{.}}{2017}]%
        {nasir2017fully}
\bibfield{author}{\bibinfo{person}{Muhammad Anis~Uddin Nasir},
  \bibinfo{person}{Aristides Gionis}, \bibinfo{person}{Gianmarco De~Francisci
  Morales}, {and} \bibinfo{person}{Sarunas Girdzijauskas}.}
  \bibinfo{year}{2017}\natexlab{}.
\newblock \showarticletitle{Fully dynamic algorithm for top-\emph{k} densest
  subgraphs}. In \bibinfo{booktitle}{\emph{Proc. of Int. Conf. on Information
  and Knowledge Management {(CIKM)}}}. \bibinfo{pages}{1817--1826}.
\newblock


\bibitem[\protect\citeauthoryear{Papalexakis, Akoglu, and Ience}{Papalexakis
  et~al\mbox{.}}{2013}]%
        {papalexakis2013more}
\bibfield{author}{\bibinfo{person}{Evangelos~E Papalexakis},
  \bibinfo{person}{Leman Akoglu}, {and} \bibinfo{person}{Dino Ience}.}
  \bibinfo{year}{2013}\natexlab{}.
\newblock \showarticletitle{Do more views of a graph help? community detection
  and clustering in multi-graphs}. In \bibinfo{booktitle}{\emph{Proc. of Int.
  Conf. on Information Fusion {(FUSION)}}}. \bibinfo{pages}{899--905}.
\newblock


\bibitem[\protect\citeauthoryear{Pechlivanidou, Katsaros, and
  Tassiulas}{Pechlivanidou et~al\mbox{.}}{2014}]%
        {DistributedCores2}
\bibfield{author}{\bibinfo{person}{Katerina Pechlivanidou},
  \bibinfo{person}{Dimitrios Katsaros}, {and} \bibinfo{person}{Leandros
  Tassiulas}.} \bibinfo{year}{2014}\natexlab{}.
\newblock \showarticletitle{{M}ap{R}educe-based distributed \emph{k}-shell
  decomposition for online social networks}. In
  \bibinfo{booktitle}{\emph{{IEEE} World Congress on Services {(SERVICES)}}}.
  IEEE, \bibinfo{pages}{30--37}.
\newblock


\bibitem[\protect\citeauthoryear{Reinthal, Andersson, Norlander,
  St{\aa}lhammar, Norlin, et~al\mbox{.}}{Reinthal et~al\mbox{.}}{2016}]%
        {reinthal2016finding}
\bibfield{author}{\bibinfo{person}{Alexander Reinthal}, \bibinfo{person}{Arvid
  Andersson}, \bibinfo{person}{Erik Norlander}, \bibinfo{person}{Philip
  St{\aa}lhammar}, \bibinfo{person}{Sebastian Norlin}, {et~al\mbox{.}}}
  \bibinfo{year}{2016}\natexlab{}.
\newblock \showarticletitle{Finding the densest common subgraph with linear
  programming}.
\newblock  (\bibinfo{year}{2016}).
\newblock


\bibitem[\protect\citeauthoryear{Ruchansky, Bonchi, Garc{\'{\i}}a{-}Soriano,
  Gullo, and Kourtellis}{Ruchansky et~al\mbox{.}}{[n. d.]}]%
        {RuchanskyBGGK17}
\bibfield{author}{\bibinfo{person}{Natali Ruchansky},
  \bibinfo{person}{Francesco Bonchi}, \bibinfo{person}{David
  Garc{\'{\i}}a{-}Soriano}, \bibinfo{person}{Francesco Gullo}, {and}
  \bibinfo{person}{Nicolas Kourtellis}.} \bibinfo{year}{[n. d.]}\natexlab{}.
\newblock \showarticletitle{To be connected, or not to be connected: that is
  the minimum inefficiency subgraph problem}. In
  \bibinfo{booktitle}{\emph{Proc. of Int. Conf. on Information and Knowledge
  Management {(CIKM)}}}. \bibinfo{pages}{879--888}.
\newblock


\bibitem[\protect\citeauthoryear{Ruchansky, Bonchi, Garc{\'\i}a-Soriano, Gullo,
  and Kourtellis}{Ruchansky et~al\mbox{.}}{2015}]%
        {ruchansky2015minimum}
\bibfield{author}{\bibinfo{person}{Natali Ruchansky},
  \bibinfo{person}{Francesco Bonchi}, \bibinfo{person}{David
  Garc{\'\i}a-Soriano}, \bibinfo{person}{Francesco Gullo}, {and}
  \bibinfo{person}{Nicolas Kourtellis}.} \bibinfo{year}{2015}\natexlab{}.
\newblock \showarticletitle{The minimum {W}iener connector problem}. In
  \bibinfo{booktitle}{\emph{Proc. of {ACM} Int. Conf. on Management of Data
  {(SIGMOD)}}}. \bibinfo{pages}{1587--1602}.
\newblock


\bibitem[\protect\citeauthoryear{Sar{\'\i}y{\"u}ce, Gedik, Jacques-Silva, Wu,
  and {\c{C}}ataly{\"u}rek}{Sar{\'\i}y{\"u}ce et~al\mbox{.}}{2013}]%
        {StreamingCores}
\bibfield{author}{\bibinfo{person}{Ahmet~Erdem Sar{\'\i}y{\"u}ce},
  \bibinfo{person}{Bu{\u{g}}ra Gedik}, \bibinfo{person}{Gabriela
  Jacques-Silva}, \bibinfo{person}{Kun-Lung Wu}, {and}
  \bibinfo{person}{{\"U}mit~V {\c{C}}ataly{\"u}rek}.}
  \bibinfo{year}{2013}\natexlab{}.
\newblock \showarticletitle{Streaming algorithms for k-core decomposition}.
\newblock \bibinfo{journal}{\emph{Proc. of the VLDB Endowment {(PVLDB)}}}
  \bibinfo{volume}{6}, \bibinfo{number}{6} (\bibinfo{year}{2013}),
  \bibinfo{pages}{433--444}.
\newblock


\bibitem[\protect\citeauthoryear{Seidman}{Seidman}{1983}]%
        {Seidman1983k-cores}
\bibfield{author}{\bibinfo{person}{S.~B. Seidman}.}
  \bibinfo{year}{1983}\natexlab{}.
\newblock \showarticletitle{Network structure and minimum degree}.
\newblock \bibinfo{journal}{\emph{Social Networks}} \bibinfo{volume}{5},
  \bibinfo{number}{3} (\bibinfo{year}{1983}), \bibinfo{pages}{269--287}.
\newblock


\bibitem[\protect\citeauthoryear{Semertzidis, Pitoura, Terzi, and
  Tsaparas}{Semertzidis et~al\mbox{.}}{2019}]%
        {semertzidis2016best}
\bibfield{author}{\bibinfo{person}{Konstantinos Semertzidis},
  \bibinfo{person}{Evaggelia Pitoura}, \bibinfo{person}{Evimaria Terzi}, {and}
  \bibinfo{person}{Panayiotis Tsaparas}.} \bibinfo{year}{2019}\natexlab{}.
\newblock \showarticletitle{Finding lasting dense subgraphs}.
\newblock \bibinfo{journal}{\emph{Data Mining and Knowledge Discovery
  {(DAMI)}}} \bibinfo{volume}{33}, \bibinfo{number}{5} (\bibinfo{year}{2019}),
  \bibinfo{pages}{1417--1445}.
\newblock


\bibitem[\protect\citeauthoryear{Shen, Shuai, Yang, Lan, Lee, Yu, and
  Chen}{Shen et~al\mbox{.}}{2015}]%
        {shen2015forming}
\bibfield{author}{\bibinfo{person}{Chih-Ya Shen}, \bibinfo{person}{Hong-Han
  Shuai}, \bibinfo{person}{De-Nian Yang}, \bibinfo{person}{Yi-Feng Lan},
  \bibinfo{person}{Wang-Chien Lee}, \bibinfo{person}{Philip~S Yu}, {and}
  \bibinfo{person}{Ming-Syan Chen}.} \bibinfo{year}{2015}\natexlab{}.
\newblock \showarticletitle{Forming online support groups for internet and
  behavior related addictions}. In \bibinfo{booktitle}{\emph{Proc. of Int.
  Conf. on Information and Knowledge Management {(CIKM)}}}.
  \bibinfo{pages}{163--172}.
\newblock


\bibitem[\protect\citeauthoryear{Sozio and Gionis}{Sozio and Gionis}{2010}]%
        {Sozio}
\bibfield{author}{\bibinfo{person}{Mauro Sozio} {and}
  \bibinfo{person}{Aristides Gionis}.} \bibinfo{year}{2010}\natexlab{}.
\newblock \showarticletitle{The community-search problem and how to plan a
  successful cocktail party}. In \bibinfo{booktitle}{\emph{Proc. of {ACM}
  {SIGKDD} Int. Conf. on Knowledge Discovery and Data Mining}}.
  \bibinfo{pages}{939--948}.
\newblock


\bibitem[\protect\citeauthoryear{Tagarelli, Amelio, and Gullo}{Tagarelli
  et~al\mbox{.}}{2017}]%
        {TagarelliEnsemble2017}
\bibfield{author}{\bibinfo{person}{Andrea Tagarelli}, \bibinfo{person}{Alessia
  Amelio}, {and} \bibinfo{person}{Francesco Gullo}.}
  \bibinfo{year}{2017}\natexlab{}.
\newblock \showarticletitle{Ensemble-based community detection in multilayer
  networks}.
\newblock \bibinfo{journal}{\emph{Data Mining and Knowledge Discovery
  {(DAMI)}}} \bibinfo{volume}{31}, \bibinfo{number}{5} (\bibinfo{year}{2017}),
  \bibinfo{pages}{1506--1543}.
\newblock


\bibitem[\protect\citeauthoryear{Tang, Wang, and Liu}{Tang
  et~al\mbox{.}}{2010}]%
        {tang2010community}
\bibfield{author}{\bibinfo{person}{Lei Tang}, \bibinfo{person}{Xufei Wang},
  {and} \bibinfo{person}{Huan Liu}.} \bibinfo{year}{2010}\natexlab{}.
\newblock \bibinfo{booktitle}{\emph{Community detection in multi-dimensional
  networks}}.
\newblock \bibinfo{type}{{T}echnical {R}eport}. \bibinfo{institution}{DTIC
  Document}.
\newblock


\bibitem[\protect\citeauthoryear{Tong and Faloutsos}{Tong and
  Faloutsos}{2006}]%
        {CenterpieceKDD06}
\bibfield{author}{\bibinfo{person}{Hanghang Tong} {and}
  \bibinfo{person}{Christos Faloutsos}.} \bibinfo{year}{2006}\natexlab{}.
\newblock \showarticletitle{Center-piece subgraphs: problem definition and fast
  solutions}. In \bibinfo{booktitle}{\emph{Proc. of {ACM} {SIGKDD} Int. Conf.
  on Knowledge Discovery and Data Mining}}. \bibinfo{pages}{404--413}.
\newblock


\bibitem[\protect\citeauthoryear{Tsourakakis}{Tsourakakis}{2015}]%
        {Tsourakakis15a}
\bibfield{author}{\bibinfo{person}{Charalampos Tsourakakis}.}
  \bibinfo{year}{2015}\natexlab{}.
\newblock \showarticletitle{The \emph{k}-clique densest subgraph problem}. In
  \bibinfo{booktitle}{\emph{Proc. of World Wide Web Conf. {(WWW)}}}.
  \bibinfo{pages}{1122--1132}.
\newblock


\bibitem[\protect\citeauthoryear{Tsourakakis, Bonchi, Gionis, Gullo, and
  Tsiarli}{Tsourakakis et~al\mbox{.}}{2013}]%
        {BabisKDD13}
\bibfield{author}{\bibinfo{person}{Charalampos Tsourakakis},
  \bibinfo{person}{Francesco Bonchi}, \bibinfo{person}{Aristides Gionis},
  \bibinfo{person}{Francesco Gullo}, {and} \bibinfo{person}{Maria Tsiarli}.}
  \bibinfo{year}{2013}\natexlab{}.
\newblock \showarticletitle{Denser than the densest subgraph: extracting
  optimal quasi-cliques with quality guarantees}. In
  \bibinfo{booktitle}{\emph{Proceedings of the 19th ACM SIGKDD international
  conference on Knowledge discovery and data mining}}. ACM,
  \bibinfo{pages}{104--112}.
\newblock


\bibitem[\protect\citeauthoryear{Wang, Zhang, Tan, and Tung}{Wang
  et~al\mbox{.}}{2010}]%
        {TriangleDensePVLDB10}
\bibfield{author}{\bibinfo{person}{Nan Wang}, \bibinfo{person}{Jingbo Zhang},
  \bibinfo{person}{Kian-Lee Tan}, {and} \bibinfo{person}{Anthony~KH Tung}.}
  \bibinfo{year}{2010}\natexlab{}.
\newblock \showarticletitle{On triangulation-based dense neighborhood graph
  discovery}.
\newblock \bibinfo{journal}{\emph{Proc. of the VLDB Endowment {(PVLDB)}}}
  \bibinfo{volume}{4}, \bibinfo{number}{2} (\bibinfo{year}{2010}),
  \bibinfo{pages}{58--68}.
\newblock


\bibitem[\protect\citeauthoryear{Wen, Qin, Zhang, Lin, and Yu}{Wen
  et~al\mbox{.}}{2016}]%
        {WenIO2016}
\bibfield{author}{\bibinfo{person}{Dong Wen}, \bibinfo{person}{Lu Qin},
  \bibinfo{person}{Ying Zhang}, \bibinfo{person}{Xuemin Lin}, {and}
  \bibinfo{person}{Jeffrey~Xu Yu}.} \bibinfo{year}{2016}\natexlab{}.
\newblock \showarticletitle{{I/O} efficient core graph decomposition at {W}eb
  scale}. In \bibinfo{booktitle}{\emph{Proc. of {IEEE} Int. Conf. on Data
  Engineering {(ICDE)}}}. \bibinfo{pages}{133--144}.
\newblock


\bibitem[\protect\citeauthoryear{Wu, Cheng, Lu, Ke, Huang, Yan, and Wu}{Wu
  et~al\mbox{.}}{2015a}]%
        {wu2015core}
\bibfield{author}{\bibinfo{person}{Huanhuan Wu}, \bibinfo{person}{James Cheng},
  \bibinfo{person}{Yi Lu}, \bibinfo{person}{Yiping Ke}, \bibinfo{person}{Yuzhen
  Huang}, \bibinfo{person}{Da Yan}, {and} \bibinfo{person}{Hejun Wu}.}
  \bibinfo{year}{2015}\natexlab{a}.
\newblock \showarticletitle{Core decomposition in large temporal graphs}. In
  \bibinfo{booktitle}{\emph{Proc. of {IEEE} Int. Conf. on Big Data}}.
  \bibinfo{pages}{649--658}.
\newblock


\bibitem[\protect\citeauthoryear{Wu, Jin, Zhu, and Zhang}{Wu
  et~al\mbox{.}}{2015b}]%
        {wu2015finding}
\bibfield{author}{\bibinfo{person}{Yubao Wu}, \bibinfo{person}{Ruoming Jin},
  \bibinfo{person}{Xiaofeng Zhu}, {and} \bibinfo{person}{Xiang Zhang}.}
  \bibinfo{year}{2015}\natexlab{b}.
\newblock \showarticletitle{Finding dense and connected subgraphs in dual
  networks}. In \bibinfo{booktitle}{\emph{Proc. of {IEEE} Int. Conf. on Data
  Engineering {(ICDE)}}}. \bibinfo{pages}{915--926}.
\newblock


\bibitem[\protect\citeauthoryear{Wuchty and Almaas}{Wuchty and Almaas}{2005}]%
        {citeulike:298147}
\bibfield{author}{\bibinfo{person}{Stefan Wuchty} {and} \bibinfo{person}{Eivind
  Almaas}.} \bibinfo{year}{2005}\natexlab{}.
\newblock \showarticletitle{Peeling the yeast protein network}.
\newblock \bibinfo{journal}{\emph{Proteomics}} \bibinfo{volume}{5},
  \bibinfo{number}{2} (\bibinfo{year}{2005}), \bibinfo{pages}{444--449}.
\newblock


\bibitem[\protect\citeauthoryear{Yan, Zhou, and Han}{Yan et~al\mbox{.}}{2005}]%
        {yan2005mining}
\bibfield{author}{\bibinfo{person}{Xifeng Yan}, \bibinfo{person}{X Zhou}, {and}
  \bibinfo{person}{Jiawei Han}.} \bibinfo{year}{2005}\natexlab{}.
\newblock \showarticletitle{Mining closed relational graphs with connectivity
  constraints}. In \bibinfo{booktitle}{\emph{Proc. of {ACM} {SIGKDD} Int. Conf.
  on Knowledge Discovery and Data Mining}}. \bibinfo{pages}{324--333}.
\newblock


\bibitem[\protect\citeauthoryear{Yin and Khaing}{Yin and Khaing}{2013}]%
        {yin2013multi}
\bibfield{author}{\bibinfo{person}{Zin~Mar Yin} {and} \bibinfo{person}{Soe~Soe
  Khaing}.} \bibinfo{year}{2013}\natexlab{}.
\newblock \showarticletitle{Multi-layered Graph Clustering in Finding the
  Community Cores}.
\newblock \bibinfo{journal}{\emph{International Journal of Advanced Research in
  Computer Engineering \& Technology {(IJARCET)}}}  \bibinfo{volume}{2}
  (\bibinfo{year}{2013}), \bibinfo{pages}{2674--2677}.
\newblock


\bibitem[\protect\citeauthoryear{Zhang, Zhang, Qin, Zhang, and Lin}{Zhang
  et~al\mbox{.}}{2017b}]%
        {zhang2017engagement}
\bibfield{author}{\bibinfo{person}{Fan Zhang}, \bibinfo{person}{Ying Zhang},
  \bibinfo{person}{Lu Qin}, \bibinfo{person}{Wenjie Zhang}, {and}
  \bibinfo{person}{Xuemin Lin}.} \bibinfo{year}{2017}\natexlab{b}.
\newblock \showarticletitle{When engagement meets similarity: efficient (k,
  r)-core computation on social networks}.
\newblock \bibinfo{journal}{\emph{Proc. of the VLDB Endowment {(PVLDB)}}}
  \bibinfo{volume}{10}, \bibinfo{number}{10} (\bibinfo{year}{2017}),
  \bibinfo{pages}{998--1009}.
\newblock


\bibitem[\protect\citeauthoryear{Zhang, Zhao, Cai, Liu, and Zhou}{Zhang
  et~al\mbox{.}}{2010}]%
        {DBLP:journals/tjs/ZhangZCLZ10}
\bibfield{author}{\bibinfo{person}{Haohua Zhang}, \bibinfo{person}{Hai Zhao},
  \bibinfo{person}{Wei Cai}, \bibinfo{person}{Jie Liu}, {and}
  \bibinfo{person}{Wanlei Zhou}.} \bibinfo{year}{2010}\natexlab{}.
\newblock \showarticletitle{Using the k-core decomposition to analyze the
  static structure of large-scale software systems}.
\newblock \bibinfo{journal}{\emph{J. Supercomputing}} \bibinfo{volume}{53},
  \bibinfo{number}{2} (\bibinfo{year}{2010}).
\newblock


\bibitem[\protect\citeauthoryear{Zhang, Yu, Zhang, and Qin}{Zhang
  et~al\mbox{.}}{2017a}]%
        {ZhangFast2017}
\bibfield{author}{\bibinfo{person}{Yikai Zhang}, \bibinfo{person}{Jeffrey~Xu
  Yu}, \bibinfo{person}{Ying Zhang}, {and} \bibinfo{person}{Lu Qin}.}
  \bibinfo{year}{2017}\natexlab{a}.
\newblock \showarticletitle{A Fast Order-Based Approach for Core Maintenance}.
  In \bibinfo{booktitle}{\emph{Proc. of {IEEE} Int. Conf. on Data Engineering
  {(ICDE)}}}. \bibinfo{pages}{337--348}.
\newblock


\bibitem[\protect\citeauthoryear{Zhu, Zou, and Li}{Zhu et~al\mbox{.}}{2018}]%
        {zhu2018diversified}
\bibfield{author}{\bibinfo{person}{Rong Zhu}, \bibinfo{person}{Zhaonian Zou},
  {and} \bibinfo{person}{Jianzhong Li}.} \bibinfo{year}{2018}\natexlab{}.
\newblock \showarticletitle{Diversified Coherent Core Search on Multi-Layer
  Graphs}. In \bibinfo{booktitle}{\emph{Proc. of {IEEE} Int. Conf. on Data
  Engineering {(ICDE)}}}. \bibinfo{pages}{701--712}.
\newblock


\end{thebibliography}


\end{document}